\documentclass[aps,prx,floatfix,nofootinbib,superscriptaddress,longbibliography]{revtex4-2}
\usepackage{graphicx} 
\usepackage{natbib}
\usepackage{graphicx}
\usepackage{amsmath}
\usepackage{amssymb}
\usepackage{amsfonts}
\usepackage{amsthm}
\usepackage{subfigure}
\usepackage{dsfont}
\usepackage{appendix}
\usepackage{tikz}
\usepackage{braket}
\usepackage{listings}
\usepackage[margin=0.8in]{geometry}
\usepackage{comment}
\usepackage{array}
\usepackage[normalem]{ulem}

\newtheorem{theorem}{Theorem}
\newtheorem*{theorem*}{Theorem}

\newtheorem{lemma}[theorem]{Lemma}
\newtheorem*{lemma*}{Lemma}

\newtheorem{corollary}[theorem]{Corollary}
\newtheorem*{corollary*}{Corollary}
\theoremstyle{definition}
\newtheorem{definition}[theorem]{Definition}

\newcommand{\n}{\notag \\}

\newcommand{\half}{\frac{1}{2}}

\newcommand{\dd}{\text{d}}

\newcommand{\tr}{\text{tr}}
\newcommand{\Wg}{\text{Wg}}

\newcommand{\ceil}[1]{\left \lceil #1 \right \rceil}
\newcommand{\Tstair}{T_\text{stair}}
\newcommand{\estair}{\lambda_\text{stair}}
\DeclareMathOperator{\vecspan}{span}
\DeclareMathOperator{\image}{image}

\begin{document}

\title{Conditional $t$-independent spectral gap for random quantum circuits and implications for $t$-design depths}

\def\urbana{
Institute for Condensed Matter Theory and IQUIST and NCSA Center for Artificial Intelligence Innovation and Department of Physics, University of Illinois at Urbana-Champaign, IL 61801, USA
}

\def\chicago{
Department of Computer Science, The University of Chicago
}

\author{James Allen}
\thanks{Co-first authors.}
\affiliation{\urbana}
\author{Daniel Belkin}
\thanks{Co-first authors.}
\affiliation{\urbana}
\author{Bryan K. Clark}
\email{bkclark@illinois.edu}
\affiliation{\urbana}
 
\begin{abstract}
    A fundamental question is understanding the rate at which random quantum circuits converge to the Haar measure. One quantity which is important in establishing this rate is the spectral gap of a random quantum ensemble.   In this work we establish a new bound on the spectral gap of the $t$-th moment of a one-dimensional brickwork architecture on $N$ qudits. This bound is independent of both $t$ and $N$, provided $t$ does not exceed the qudit dimension $q$. We also show that the bound is nearly optimal. The improved spectral gaps gives large improvements to the constant factors in known results on the approximate $t$-design depths of the 1D brickwork, of generic circuit architectures, and of specially-constructed architectures which scramble in depth $O(\log N)$.  We moreover show that the spectral gap gives the dominant $\epsilon$-dependence of the t-design depth at small-$\epsilon$.
    Our spectral gap bound is obtained by bounding the $N$-site 1D brickwork architecture by the spectra of 3-site operators. We then exploit a block-triangular hierarchy and a global symmetry in these operators in order to efficiently bound them.
    The technical methods used are a qualitatively different approach for bounding spectral gaps and and have little in common with previous techniques. 
\end{abstract}

\maketitle

\section{Introduction}
Random quantum circuits are an important model in quantum information theory and beyond. They are used in quantum advantage experiments \cite{Boixo2018, Arute2019,Movassagh2019}, as a representation of unstructured dynamics in physical systems like black holes \cite{Brown2018,Hayden2007}, and as a setting for studying measurement-induced phase transitions \cite{Skinner2019,Bao2020} as well as the AdS/CFT correspondence \cite{Jafferis2022}. Central to most of these applications is the approximate $t$-design depth, which characterizes the depth at which the moments of a random quantum circuit ensemble become difficult to distinguish from those of the Haar measure. 

Many of the known results on $t$-design depths are based on the spectral gap (see Definition~\ref{def:spectral_gap}) of the 1D brickwork architecture. 
Refs.~\onlinecite{Brandao2016,Haferkamp2022,Chen2024b} use bounds on the spectral gap to bound the additive-error and multiplicative-error approximate t-design depths of the 1D brickwork architecture. These bounds scale as $O(N)$, where $N$ is the number of sites. Refs.~\onlinecite{Schuster2024,LaRacuente2024}, meanwhile, use results on the spectral gap of the brickwork to construct families of circuits that form approximate $t$-designs in depth $O(\log N)$. Ref.~\onlinecite{Belkin2023} gives a reduction from spectral gaps of arbitrary architectures to that of the 1D brickwork to show that all well-connected circuit architectures form approximate $t$-designs in depth $\text{poly}(N)$. 

In this work we give improved bounds on the spectral gap of the 1D brickwork when the qudit dimension $q$ exceeds $t$. 
We further show that our bounds on the spectral gap are nearly tight, providing a matching lower bound which differs from our upper bound by at most 15\%. We also show that the spectral gap is the sole quantity controlling the asymptotic small-$\epsilon$ scaling of the $\epsilon$-approximate $t$-design depth. More precisely we provide both upper and lower bounds on the depth of the form $\log\left(\frac{1}{\sqrt{1-\Delta}}\right)^{-1}\log\left(\frac{1}{\epsilon}\right) + o\left[\log\left(\frac{1}{\epsilon}\right)\right]$
, where $\Delta$ is the spectral gap. It follows that our upper and lower bounds on the spectral gap give upper and lower bounds on the approximate $t$-design depth, and that the true formula for the approximate $t$-design depth must involve the spectral gap.

Our spectral gap bound also improves the approximate $t-$design depth for each of the five bounds listed above.
The asymptotic $N$-dependencies do not change, but the constant factors often change dramatically. We show, for example, that the 100-site 1D brickwork of local dimension $4$ forms a $10^{-4}$-approximate $4$-design in at most 3030 layers. The best previously-known bound for this case was $1.77 \cdot 10^{15}$ layers\cite{Brandao2016}.
Our bound does not vary with $t$, but it is valid only when the local Hilbert space dimension $q \geq t$. If our proof technique can be extended to $q < t$, the asymptotic $t$-dependence of these bounds will also improve by a factor of $\text{polylog}(t)$.

The first step of our proof is a reduction from the $N$-site brickwork to a series of $3$-site operators. A bound on the spectral gaps of the $3$-site operators implies a bound on the spectral gap of the full operator. Moreover, this relationship between the spectral gap bounds is independent of $N,q,$ and $t$. 
We then show that the spectral gaps of the $3$-site operators do not depend on $t$ so long as $t \leq q$. The restriction on $t$ enters into our proof quite late, and so it is possible that many of the techniques we develop will be useful for establishing the spectral gap in the $t > q$ case. Some of the mathematical structure we develop is also relevant to non-brickwork architectures or properties of the random circuit other than the spectral gap, such as anticoncentration~\cite{Dalzell2022}. 

Our proof exploits three types of structure in the eigenspaces of the moment operator. We first describe a block-triangular structure due to the 1D geometry of the circuit. This is similar in flavor to the known $t = 2$ block-triangular structure exploited by ref.~\cite{Deneris2024}, but not equivalent even in the $t = 2$ case. Using this structure we reduce the problem of bounding the $N$-site brickwork to that of bounding a family of $3$-site circuits. Second, we develop a block-triangular structure from the subgroup structure of the permutation group. This allows us to decompose the eigenspaces of the $t$\textsuperscript{th} moment operator, isolating new eigenvalues from those inherited from the $(t-1)$\textsuperscript{th} moment operator.
Third, we find a block-diagonal structure which corresponds to the irreducible representations of the symmetric group. We believe this proof technique is independently interesting beyond our new bounds on the spectral gap. It has very little in common with prior methods, such as detectability lemma-based approaches that bound the spectral gap through the Nachtergaele method\cite{Brandao2016}, Knabe bounds\cite{Haferkamp2021}, or relating to the PFC ensemble\cite{Chen2024b}. Some of our strategies may prove useful for stronger results on the 1D brickwork spectral gap or other properties of random circuits in the future.

\subsection{Prior work}
For a one-dimensional brickwork architecture (see Definition~\ref{def:1d_brickwork_v2}) over $N$ qudits with local Hilbert space dimension $q$, ref.~\onlinecite{Brandao2016} established a lower bound on the spectral gap (see Definition~\ref{def:spectral_gap}):
\begin{gather}
    \Delta(N,q,t) \geq \left[553000\lceil \log_q(4t)\rceil^2 q^2 t^{5+\frac{3.1}{\log q}}\right]^{-1}
\end{gather}

Ref.~\onlinecite{Haferkamp2022} introduced an improvement on the $t$ dependence of the previous bound in the case of qubits:
\begin{gather}
    \Delta(N,q,t) \geq \left[3\times 10^{13} \log^5(t) t^{4+\frac{3}{\sqrt{\log_2(t)}}}\right]^{-1}\quad  \text{for } q = 2     
\end{gather}

Ref.~\onlinecite{Chen2024b} recently improved these $q=2$ bounds to one independent of $N$ and only weakly dependent on $t$:
\begin{gather}
    \Delta(N,q, t) \geq \Omega(\log(t)^{-7})\quad \text{for } q = 2, t \leq \Theta(2^{2n/5})
\end{gather}
with an alternate bound for a slightly larger range of $t$:
\begin{gather}
    \Delta(N,q, t) \geq \Omega\left(\frac{N^{-5}}{\log(N)}\right)\quad \text{    for } q = 2, t \leq \Theta(2^{n/2})
\end{gather}

These bounds leave two areas of improvement.  Firstly, establishing a tighter bound on the exact value of the spectral gap may have useful practical implications in the construction of these circuits on real quantum systems. The constant factors in these bounds are for the most part impractically large. 
Second, it seems likely that the polylogarithmic scaling in $t$ established by ref.~\onlinecite{Chen2024b} is not optimal.

Ref.~\onlinecite{Hunter-Jones2019} showed that
\begin{gather}
    \Delta(N,q, t) \geq 1 - \left(\frac{2}{q}\right)^2 \quad \text{for } q \rightarrow \infty
\end{gather}
Ref.~\onlinecite{Haferkamp2021} showed that $t$-independence holds even for finite values of $q$ above some cutoff $q_0(t)$. In particular, they found
\begin{gather}
\label{eq:eighteenth_bound}
    \Delta(N, q, t) \geq \frac{1}{18} \quad \text{for } q \geq 6t^2
\end{gather}

How far can we reduce this limit? It has been conjectured\cite{Hunter-Jones2019,Belkin2023} that
\begin{gather}
    \label{eq:conjectured_bound}
    \Delta(N, q, t) = 1 - \left(\frac{2q}{q^2 + 1}\right)^2 \quad \text{for } N \rightarrow \infty
\end{gather}

In this paper, we will establish a near-optimal, $t$-independent spectral gap bound for one-dimensional random circuits for any $q \geq t$. Our result is in the spirit of Equation \ref{eq:eighteenth_bound}, but with the cutoff improved to $q_0(t) = t$ and the spectral gap tightened to nearly that of Equation \ref{eq:conjectured_bound}.

\subsection{Main results}
Our core results are about the spectral gap (see Definition \ref{def:spectral_gap}) of the 1D brickwork architecture (see Definition \ref{def:1d_brickwork_v2}). 
\begin{theorem}
    \label{thm:gap_t_leq_q}
    When \(t \leq q\), the spectral gap of the 1D brickwork architecture with \(N\) sites of local Hilbert space dimension \(q\) is at least
    \begin{gather}
        \Delta(N,q,t) \geq 1 - \left(\frac{2q}{q^2 + 1}\frac{1 + \sqrt{1 + \frac{1}{q^2}}}{2}\right)^2
    \end{gather}
\end{theorem}

In addition, we show that known results on the $t=2$ case can be extended to a lower bound for all $t$.
\begin{theorem}
    \label{thm:upper_bound}
    The spectral gap of the 1D brickwork architecture with \(N\) sites of local Hilbert space dimension \(q\) is at most
    \begin{gather}
        \Delta(N,q,t) \leq 1 - \left(\frac{2q}{q^2 + 1} \cos \frac{\pi}{N}\right)^2
    \end{gather}
\end{theorem}
The proof of this theorem follows from Corollary \ref{corollary:deranged_block_eigvals} and the results of refs.~\onlinecite{Deneris2024}  and \onlinecite{Znidaric2022}. 
When $N$ is large, this lower bound approaches $1 - \left(\frac{2q}{q^2 + 1}\right)^2 $. Figure \ref{fig:bound_comparison} shows that our upper and lower bounds differ by at most $0.0473$, with convergence as $q \rightarrow \infty$.

\begin{figure}
    \label{fig:bound_comparison}
    \centering
    \begin{tikzpicture}
        \begin{scope}
            \node[anchor=north west,inner sep=0] (image_a) at (0.4,0)
            {\includegraphics[width=0.6\columnwidth]{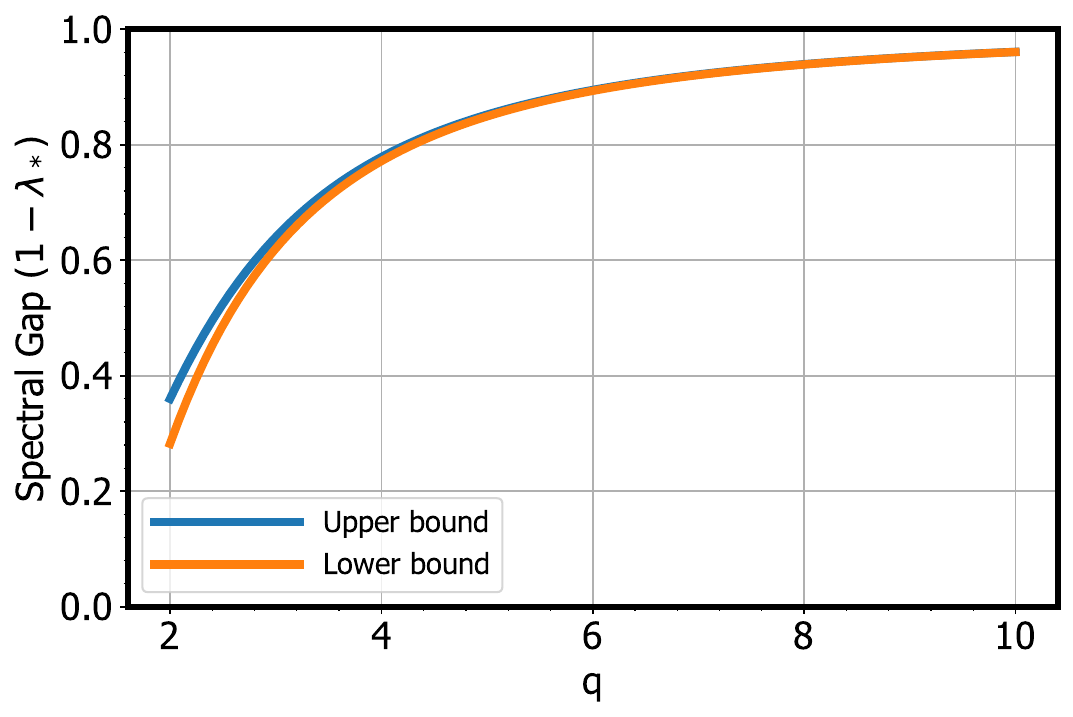}};
        \end{scope}
    \end{tikzpicture}
    \vspace*{-0.4cm}
    \caption{Comparison of our upper and lower bounds on the spectral gap when $N \rightarrow \infty$. }
\end{figure}

These results on spectral gaps imply new bounds on the approximate $t$-design depth (see Definition~\ref{def:atd}). 
\begin{corollary}
    Let $\ell_*(N,q,t,\epsilon)$ be the number of layers at which the 1D brickwork architecture with \(N\) sites of local Hilbert space dimension \(q\) first forms an \(\epsilon\)-approximate \(t\)-design.
    When \(t \leq q\), 
\begin{gather}
    \ell_* \leq 1 + C\left(2 N t \log q + \log \frac{1}{\epsilon}\right) 
\end{gather}
where
\begin{gather}
    C = \left[\log \frac{q^2 + 1}{2q} + \log\frac{2}{1 + \sqrt{1 + \frac{1}{q^2}}}\right]^{-1} \leq 6.032
\end{gather}
\end{corollary}
    This corollary is a consequence of Theorem~\ref{thm:depth_from_eigval}, in Appendix~\ref{app:depth_from_eigval}. Theorem \ref{thm:depth_from_eigval} also gives a similar bound on the multiplicative-error $t$-design depth in terms of the spectral gap.
    
\begin{corollary}
    \label{corollary:depth_bound_upper_lower}
    There exists a constant $C(N,q,t)$ such that
    \begin{gather}
        \ell_*(N,q,t,\epsilon) = \left(C(N,q,t) + o(1)\right) \log \frac{1}{\epsilon}
    \end{gather}
    as \(\epsilon \rightarrow 0\). If $t \leq q $, then
    \begin{gather}
        \left[\log \frac{q^2 + 1}{2q} + \log\frac{1}{\cos \frac{\pi}{N}}\right]^{-1} \leq C(N,q,t) \leq \left[\log \frac{q^2 + 1}{2q} + \log\frac{2}{1 + \sqrt{1 + \frac{1}{q^2}}}\right]^{-1}
    \end{gather}
\end{corollary}
This corollary is proven in Appendix~\ref{app:depth_bound_upper_lower}. This implies that the behavior of the random quantum circuit has rather precise bounds in the large depth/low-$\epsilon$ regime. In particular, the depth at which the circuit reaches an approximate $t$-design is \textit{independent} of both $N$ and $t$ in this regime (although the location of the regime itself might be dependent on $N$ and $t$).

In addition, we have a semi-numerical bound on the spectral gap for finite system sizes:
\begin{theorem}
    \label{thm:gap_finite_N}
    When \(q = 2\), \(t \leq 6\), and \(N \leq 1000\), the spectral gap of the 1D brickwork architecture with \(N\) sites of local Hilbert space dimension \(q\) is at least
    \begin{gather}
        1 - \left(\frac{2q}{q^2 + 1}\frac{1 + \sqrt{1 + \frac{1}{q^2}}}{2}\right)^2
    \end{gather}
\end{theorem}
These results will be proved in Appendix \ref{app:gap_finite_N}. They can be extended to other $N$, $t$, and $q$ with additional computational effort. 

\subsection{Proof overview}
The bulk of this work is devoted to proving Theorem \ref{thm:gap_t_leq_q}. The first phase is a reduction from the brickwork architecture to a family of three-site operator. We begin by relating the brickwork architecture to the staircase architecture (Figure \ref{fig:staircase_architecture}). We then define an orthogonal decomposition of the vector space into subspaces. This corresponds to a specific block-matrix decomposition of the transfer matrix. Next, we show that the spectral gap of the architecture can be bounded in terms of the norms of each block in this matrix. We complete this phase by relating each block norms to the spectral gap of an effective three-site operator.

The second phase of our proof is concerned with bounds on these three-site spectral gaps. We show that it suffices to restrict these operators to a certain subspace, which we term the \textbf{complete derangement space}. Furthermore, the representation theory of the symmetric group allows us to decompose each of these subspaces into isotypic components. 

At this point our bounds split into cases. The general strategy is to obtain analytic bounds when $q$ or $t$ is sufficiently large and numerically bound the spectral gap for smaller $t$. We use multiple numerical bounds depending on the size of $t$, with larger $t$ values using looser bounds that are easier to compute, while smaller $t$ values use tighter bounds that require more detailed analysis.

\section{Preliminary work}
\subsection{Definitions}
Taking an ensemble $\varepsilon$ of random quantum circuits $U_\varepsilon$ acting on a Hilbert space $\mathcal{H}$, we want to understand the conditions under which this ensemble approaches within $\epsilon$ distance of a global Haar random distribution given any $t$-body measurement. 
\begin{definition}
    The \textbf{$t$\textsuperscript{th} moment operator} of a random quantum circuit ensemble $U_\varepsilon$ is the $L(\mathcal{H}^{\otimes t}) \rightarrow L(\mathcal{H}^{\otimes t})$ channel formed by averaging the action of $t$ copies of the circuit as follows:
    \begin{gather}
        \Phi_\varepsilon^{(t)}[\rho] = \int_\varepsilon (U_\varepsilon^\dagger)^{\otimes t} \! \rho (U_\varepsilon)^{\otimes t} \dd U_\varepsilon
    \end{gather}
    We can also write this as the expectation of the circuit ensemble over $2t$ copies
    \begin{gather}
        \Phi_\varepsilon^{(t)} = \left \langle U_\varepsilon^{\otimes t,t} \right\rangle_{U_\varepsilon}
    \end{gather}
    where $X^{\otimes t,t} \equiv X^{\otimes t}  \otimes (X^*)^{\otimes t}$
\end{definition}

\begin{definition}
\label{def:spectral_gap}
The \textbf{spectral gap} $\Delta(\varepsilon,t)$ of a random quantum circuit ensemble $\varepsilon$ and moment $t$ is the difference between the largest and second-largest eigenvalues of the corresponding $t$-th moment operator. For a one-dimensional brickwork of $N$ qubits with local Hilbert space dimension $q$, we will refer to the spectral gap as $\Delta(N,q,t)$ instead.
\end{definition}
We call the second-largest eigenvalue the \textbf{subleading eigenvalue} or \textbf{SEV}. 

A primary use of the spectral gap is in determining the $\epsilon$-approximate $t$-design depth of a random quantum circuit, which is the depth required for the statistical average of any specific $t$-body measurement of the random quantum circuit to be equal to to a measurement on the global Haar distribution, within an error of $\epsilon$:
\begin{definition}
    \label{def:atd}
    A random quantum circuit ensemble $\varepsilon$ is an \textbf{$\epsilon$-approximate $t$-design} if the diamond norm distance between the ensemble's $t$\textsuperscript{th} moment operator $\Phi_\epsilon^{(t)}$ and the Haar measure's $t$\textsuperscript{th} moment operator is at most $\epsilon$:
    \begin{gather}
        ||\Phi_\varepsilon^{(t)} - \Phi_\text{Haar}^{(t)}||_\diamond \leq \epsilon
    \end{gather}
\end{definition}

\subsection{Mapping moment operators to tensor networks}
In this section, we will give a mapping from the spectrum of the moment operator corresponding to a random quantum circuit architecture to the spectrum of a particular tensor network. 

For a circuit ensemble $U_\varepsilon$ composed of 2-local\footnote{By "2-local" we mean that the unitary gate acts nontrivially only on some pair of sites.} Haar random unitaries $U_i$, we may use the independence of the gates to write
\begin{gather}
    \Phi_{\varepsilon}^{(t)} = \left\langle \prod_i U_i^{\otimes t,t} \right \rangle = \prod_i \left \langle U_i^{\otimes t,t} \right\rangle_{U_i}
\end{gather}
The result is a tensor network in the same form as the original circuit, with each random gate replaced with a copy of \(\left \langle U_i^{\otimes t,t} \right\rangle\) (see Figure~\ref{fig:stat_mech_process}.). 

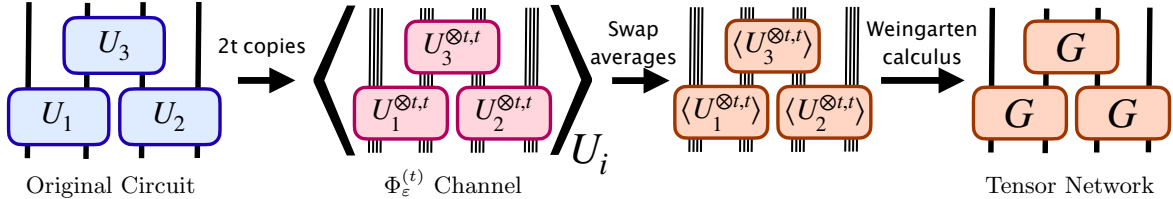
\begin{figure}
    \centering
    \begin{tikzpicture}
        \begin{scope}
            \node[anchor=north west,inner sep=0] (image_a) at (0.4,0)
            {\includegraphics[width=0.9\columnwidth]{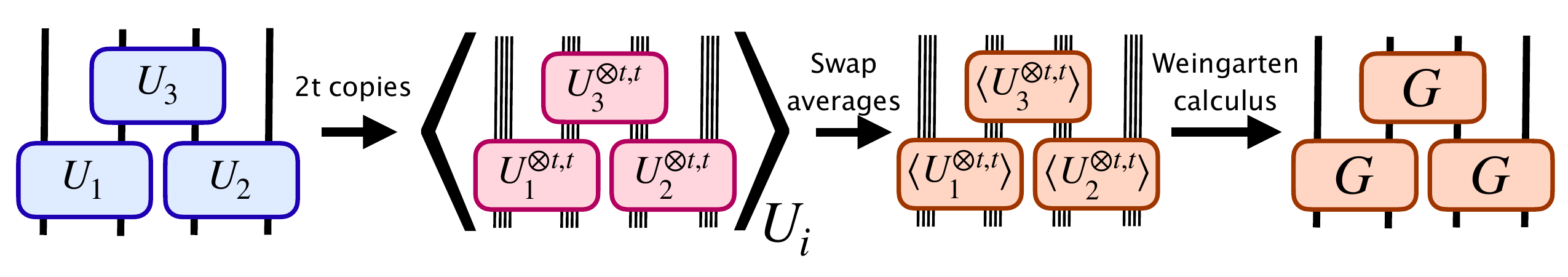}};
            \node [anchor=north] (note) at (1.95,-2.5) {\small{Original Circuit}};
            \node [anchor=north] (note) at (6.5,-2.4) {\small{$\Phi_\varepsilon^{(t)}$ Channel}};
            \node [anchor=north] (note) at (14.7,-2.5) {\small{Tensor Network}};
        \end{scope}
    \end{tikzpicture}
    \vspace*{-0.4cm}
    \caption{Converting a random circuit ensemble $\varepsilon$ consisting of 2-local Haar random unitaries $U_i$ to its $t$-fold channel $\Phi_\varepsilon^{(t)}$, which can be written as a tensor network of identical tensors $G$ in the same position as the original unitaries.}
    \label{fig:stat_mech_process}
\end{figure}

The averaged tensors $\left \langle U_i^{\otimes t,t} \right\rangle_{U_i}$ are projectors on to the commutant. Let us define the state
\begin{gather}
\ket{\sigma} = \frac{1}{\sqrt{q}^t} \sum_{\vec{i} \in \mathbb{Z}_q^t} \ket{\vec{i}} \otimes \ket{\sigma(\vec{i})}
\end{gather}
where \(\sigma\) is some permutation in the symmetric group $S_t$. The permutation \(\sigma\) acts on the multi-index $\vec{i}$ by reordering its elements. In terms of these states, the moment operator corresponding to a random 2-site gate may be written
\begin{gather}
    G \equiv \left \langle U_i^{\otimes t,t} \right\rangle = \sum_{\sigma \in S_t} \Wg(\sigma^{-1} \tau, q^2)\ket{\sigma}^{(t)\otimes 2}\bra{\tau}^{(t)\otimes 2}
\end{gather}
where \(\Wg\) is the Weingarten function defined in ref.~\onlinecite{Collins2006}, multiplied by a pseudo-normalizing\footnote{In this paper we define the Weingarten function with an extra factor of $d^t$ compared to the traditional definition. We call this the \textit{pseudo-normalized} version of the Weingarten function, because this way $\Wg(\mathbb{I}, q)$ approaches $1$ as $d \rightarrow \infty$.} factor $d^{t}$. It will also be useful to note that
\[\braket{\sigma|\tau} = q^{-|\sigma^{-1}\tau|}\]
where $|\sigma^{-1}\tau|$ is the minimal number of transpositions required to form $\sigma^{-1}\tau$.

\subsection{Brickwork and staircase}
\begin{definition}
    \label{def:1d_brickwork_v2}
    The \textbf{1D brickwork architecture} is a random quantum circuit architecture. Odd-numbered layers have gates between sites $i$ and $i+1$ for each odd $i < N$, while even-numbered layers have gates between sites $i$ and $i+1$ for each even $i < N$.
\end{definition}
\begin{definition}
    \label{def:staircase_v2}
    The \textbf{staircase architecture}, illustrated in Figure \ref{fig:staircase_architecture}, is the random quantum circuit architecture in which layer $i$ consist of a single gate between sites $i$ and $i+1$.
\end{definition}

Let $L_O$ and $L_E$ refer to the ``even'' and ``odd'' layers of the 1D brickwork architecture, respectively. The transfer matrix associated with a single period of the architecture is then \(T_{1DB} = L_O L_E\). Let $\Tstair$ refer to the moment operator of the staircase architecture.
\begin{lemma}
    \label{lemma:brickwork_is_staircase}
    The eigenspectra of \(L_O L_E\) and the staircase transfer matrix $\Tstair$ are the same.
\end{lemma}
This is Theorem 1 of ref.~\onlinecite{Bensa2021}. It is proven by observing that the brickwork and staircase are related by a cyclic permutation.

\section{Bounding the staircase by a single gate}\label{section:single_gate_bounding}

\begin{figure}
    \centering
    \begin{tikzpicture}
        \begin{scope}
            \node[anchor=north west,inner sep=0] (image_a) at (0.4,0)
            {\includegraphics[width=0.2\columnwidth]{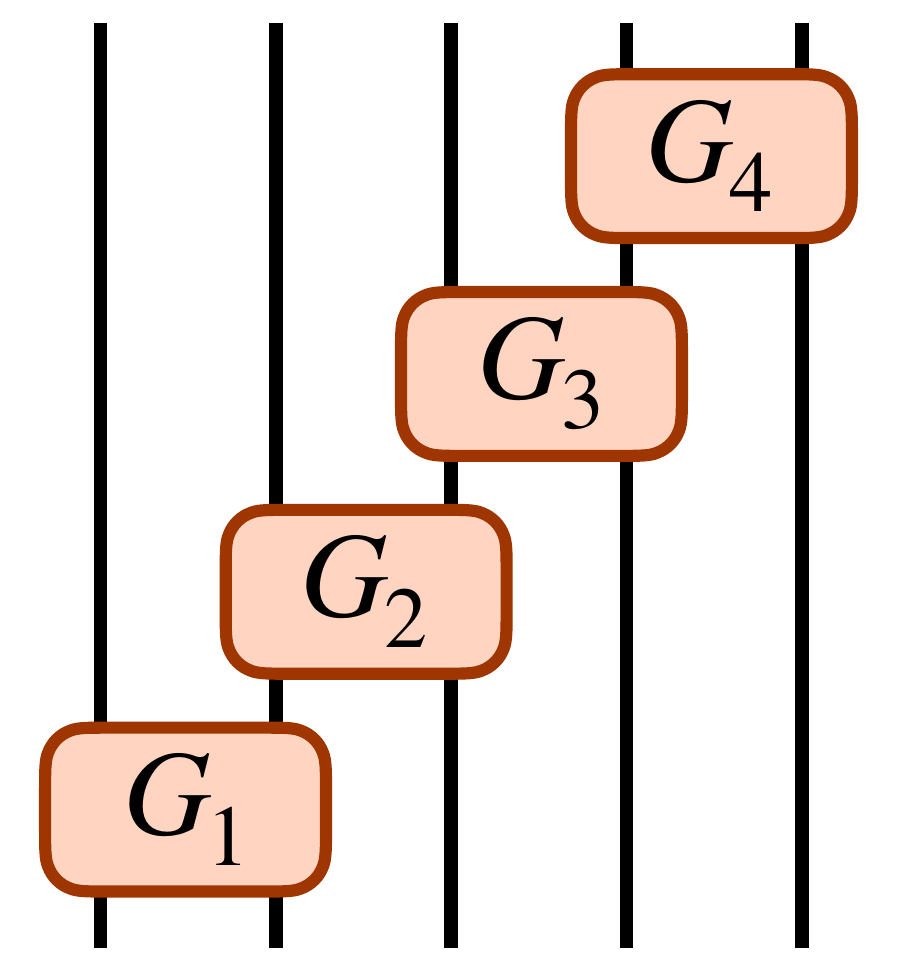}};
            \node [anchor=north west] (note) at (0,-0.1) {\small{\textbf{a)}}};
        \end{scope}
        \begin{scope}
            \node[anchor=north west,inner sep=0] (image_b) at (4.7,0)
            {\includegraphics[width=0.17\columnwidth]{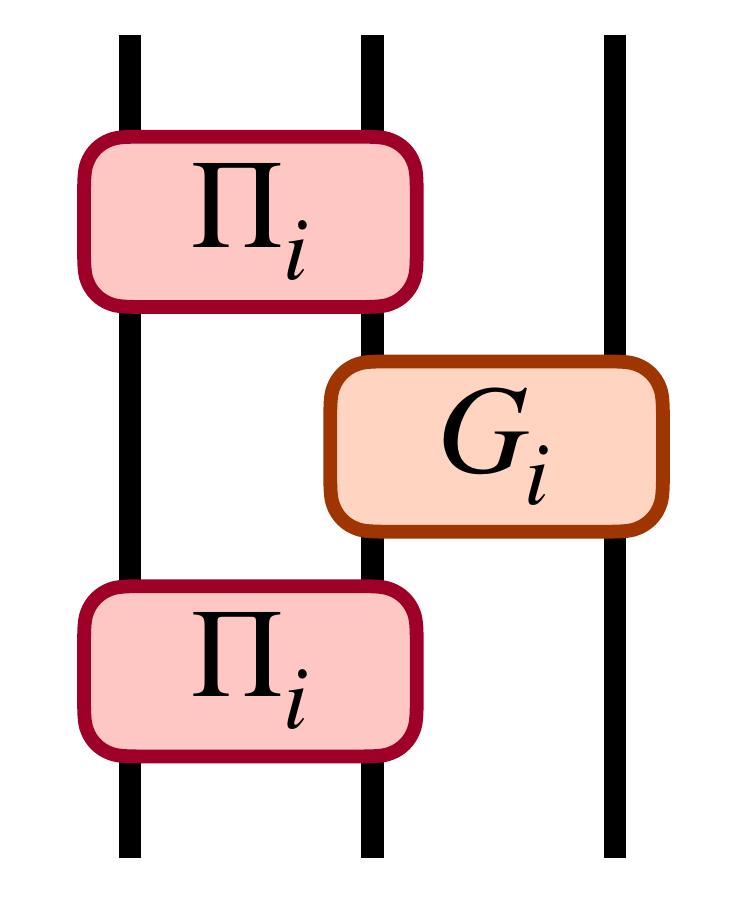}};
            \node [anchor=north west] (note) at (4.3,-0.1) {\small{\textbf{b)}}};
        \end{scope}
    \end{tikzpicture}
    \vspace*{-0.4cm}
    \caption{a) The staircase architecture on 5 sites. b) The effective 3-site gate operator $K_i \equiv \Pi_i G_i \Pi_i$.}
    \label{fig:staircase_architecture}
\end{figure}

In this section, we will develop a reduction from an arbitrary staircase to the three-site staircase. The staircase architecture consists of $N-1$ gates $G_1, ..., G_{N-1}$, where gate $G_i$ acts on sites $i,i+1$ (see Figure~\ref{fig:staircase_architecture}a.)

Let \(\Pi_i\) be the orthogonal projector to the span of the uniform permutations on the first $i$ sites. Since the permutation states which are uniform on the first $i$ sites are a subset of the permutation states which are uniform on the first $i+1$ sites, we have a hierarchy on the \(\Pi_i\). 
\begin{lemma} \label{lemma:projector_redundancy}
    The projectors and gate follow a subspace hierarchy:
    \begin{itemize}
        \item $\Pi_i \Pi_j = \Pi_j \Pi_i = \Pi_j$ for all $i \leq j$.
        \item $\Pi_j G_i = G_i \Pi_j = \Pi_j$ for all $i < j$. 
    \end{itemize}
\end{lemma}
With this lemma, we define
\begin{itemize}
    \item $P_i = \Pi_i - \Pi_{i+1}$ is the projector onto states which are uniform on the first $i$ sites, but orthogonally not uniform on the first $i+1$ sites. We will define $\Pi_{N+1} = 0$, so that $P_N = \Pi_N$. This is consistent with  Lemma~\ref{lemma:projector_redundancy}.
    \item $\lambda_i = ||P_i G_i P_i||_{\infty} = ||\Pi_i G_i \Pi_i - \Pi_{i+1}||_{\infty}$ is the leading singular value (and eigenvalue) of $G_i$ restricted to the image of $P_i$. Here $||A||_{\infty}$ denotes the spectral norm of $A$.
    \item $\mu_i = ||(I-\Pi_i) G_i \Pi_i||_{\infty} = ||\Pi_i G_i (I-\Pi_i)||_{\infty}$.
\end{itemize}
It will also be useful to define the gate $G_i$ restricted to the space projected by $\Pi_i$ as well:
\begin{itemize}
    \item $K_i = \Pi_i G_i \Pi_i$ (i.e. $\lambda_i$ is the largest singular value of $K_i - \Pi_{i+1}$)
\end{itemize}
\begin{lemma} \label{lemma:mu_bound}
    Suppose $\lambda_i \leq \frac{1}{2}$. Then
    $\mu_i = \sqrt{\lambda_i(1-\lambda_i)}$
\end{lemma}
\begin{proof}
    We use the facts that projectors are Hermitian and $||f(H)||_{\infty} = \max_{\lambda} f(\lambda)$ for analytic functions $f$ and Hermitian operators $H$ with eigenvalues $\lambda$. 
    This gives
    \begin{align*}
        \mu_i &= ||(I-\Pi_i) G_i \Pi_i||_{\infty}\\
        \mu_i^2 &= ||\Pi_i G_i (I-\Pi_i)^2 G_i \Pi_i||_{\infty}\\
        &= ||\Pi_i G_i \Pi_i - (\Pi_i G_i \Pi_i)^2||_{\infty}\\
        &= ||(\Pi_i G_i \Pi_i - \Pi_{i+1}) - (\Pi_i G_i \Pi_i - \Pi_{i+1})^2||_{\infty}\\
        &= ||(K_i - \Pi_{i+1}) - (K_i - \Pi_{i+1})^2||_{\infty}\\
        &= \max_{\lambda} \lambda(1-\lambda)
    \end{align*}
    for eigenvalues $\lambda$ of $K_i - \Pi_{i+1}$. If all eigenvalues are below $0.5$, then the function \(\lambda(1-\lambda)\) is monotonically  increasing, so the maximum is attained by the largest eigenvalue. 
\end{proof}
Next, we note that we can establish the following resolution of the identity:
\begin{gather}
    I = \Pi_1 = P_1 + \Pi_2 = ... = \sum_{i=1}^N P_i + \Pi_{N+1} = \sum_{i=1}^N P_i
\end{gather}
Moreover, for all $j > i$,
\begin{gather}
    P_j P_i = P_i P_j = \Pi_j - \Pi_{j+1} - (\Pi_j - \Pi_{j+1}) = 0
\end{gather}
so each $P_i$ is an orthogonal idempotent. It follows that for any vector $v$ in the permutation space, we have
\begin{gather}
    ||v||^2 = \sum_i ||P_i v||^2
\end{gather}
Finally, we note as a corollary of Lemma~\ref{lemma:projector_redundancy},
\begin{corollary} \label{lemma:commutation_relations}
    $P_j$ absorbs all $G_i$ for $i < j$, similar to $\Pi_j$ in Lemma~\ref{lemma:projector_redundancy}. In addition, $P_j$ commutes with all $G_i$ for $i > j+1$. In other words,
    \begin{align}
        P_j G_i &= G_i P_j = P_j \qquad i < j \\
        P_j G_i &= G_i P_j  \qquad \qquad  i > j+1
    \end{align}
\end{corollary}
Define the partial staircase operator
\begin{gather}
    T_j = G_j G_{j-1} ... G_1
\end{gather}
with $T_{N-1} \equiv T$ being the full staircase operator. We also define $s_{ij} = ||P_i T P_j||_{\infty}$. Because of Lemma~\ref{lemma:commutation_relations}, 
\begin{align}
    s_{ij} &= ||P_i G_{N-1} ... G_1 P_j||_{\infty}\n
    &= ||P_i G_{N-1} ... G_j P_j G_{j-1} ... G_1||_{\infty}\n
    &\leq ||G_{N-1} ... G_{i+2} P_i G_{i+1}G_i...G_j P_j||_{\infty}\n
    &\leq ||P_i G_{i+1}G_i...G_j P_j||_{\infty} \label{eq:sij_expansion}
\end{align}
for all $1 \leq i < N-1$, $1 \leq j < N$. If $i+1 < j$, there are no gates in between the two projectors, as every gate commutes with one or the other.. This leads to the following relations:
\begin{lemma}
    For all $1 \leq i,j < N$, the elements $s_{ij}$ are bounded as follows:
    \begin{itemize}
        \item For $i < j-1$, $s_{ij} = 0$.
        \item For $i = j-1$, $s_{ij} \leq \mu_j$.
        \item For $i \geq j$, $s_{ij} \leq \lambda_j \prod_{k=j+1}^i \mu_k$.
    \end{itemize}
\end{lemma}
\begin{proof}
    Starting with the first case, if $i < j-1$ there are no gates in Equation~\ref{eq:sij_expansion}, so
    \begin{gather}
        s_{ij} \leq ||P_i G_{i+1}...G_j P_j||_{\infty} = ||P_i P_j||_{\infty} = 0
    \end{gather}
    For $i = j-1$, we have
    \begin{align*}
        s_{j-1,j} &\leq ||P_{j-1} G_j P_j||_{\infty} \\
        &\leq ||\Pi_{j-1} G_j \Pi_j - \Pi_j G_j \Pi_j||_{\infty}\\
        &\leq ||G_j \Pi_j - \Pi_j G_j \Pi_j||_{\infty}\\
        &\leq ||(I-\Pi_j) G_j \Pi_j ||_{\infty}\\
        &\leq \mu_j
    \end{align*}
    Finally, for $i \geq j$, we first assume $i < N-1$. Then we have
    \begin{align*}
        s_{i,j} &\leq ||P_i G_{i+1}...G_j P_j||_{\infty}\\
        &\leq ||(I-\Pi_{i+1})\Pi_i ... \Pi_j G_{i+1} ... G_j P_j||_{\infty}\\
        &\leq ||(I-\Pi_{i+1})G_{i+1}\Pi_i G_i \Pi_{i-1} ... \Pi_j G_j P_j||_{\infty}\\
        &\leq ||(I-\Pi_{i+1})G_{i+1}\Pi_i G_i \Pi_{i-1} ... \Pi_j G_j (I-\Pi_{i+1}) (I-\Pi_{i}) ... (I-\Pi_{j+1})\Pi_j||_{\infty}\\
        &\leq \bigg|\bigg|\left[(I-\Pi_{i+1})G_{i+1}(I-\Pi_{i+1})\right]\left[\Pi_i G_i (I-\Pi_i)\right]\left[\Pi_{i-1} G_{i-1} (I-\Pi_{i-1})\right] ... \left[\Pi_{j+1} G_{j+1} (I-\Pi_{j+1})\right]\Pi_j G_j \Pi_j\bigg|\bigg|_{\infty}\\
        &\leq \big|\big|(I-\Pi_{i+1})G_{i+1}(I-\Pi_{i+1})\big|\big|_{\infty} \mu_i\mu_{i-1} ... \mu_{j+1} \lambda_j\\
        &\leq \lambda_j \prod_{k=j+1}^i \mu_k.
    \end{align*}
    Where we have used the facts that the projector $(I-\Pi_i)$ absorbs $(I-\Pi_j)$ for all $i < j$, and commutes with $G_j$ for all $j < i$. 
    If $i = N-1$, gate $G_{i+1}$ does not exist, so $s_{ij} = ||P_i G_i ... G_j P_j||_{\infty}$ instead. This removes the $\big|\big|(I-\Pi_{i+1})G_{i+1}(I-\Pi_{i+1})\big|\big|_{\infty}$ term from the penultimate line of the previous derivation, which did not change the inequality anyways. 
\end{proof}
\begin{corollary} \label{corollary:gate_bound}
    Suppose there exists $\lambda \leq 0.5$ such that $\lambda_i \leq \lambda$ for all $i$ (and therefore by Lemma~\ref{lemma:mu_bound}, $\mu_i \leq \mu \equiv \sqrt{\lambda(1-\lambda)}$ for all $i$). Then for all $1 \leq i,j < N$,
    \begin{itemize}
        \item If $i < j-1$, $s_{ij} = 0$.
        \item If $i = j-1$, $s_{ij} \leq \mu$.
        \item If $i \geq j$, $s_{ij} \leq \lambda \mu^{i-j}$.
    \end{itemize}
\end{corollary}

For the rest of this section we will suppose that such an upper bound $\lambda\leq 0.5$ exists. We thus have
\[||P_i T P_j ||_{\infty} \leq A_{ij}\]
where
\begin{equation}
\label{eq:aformat}
    A = \begin{bmatrix}
    \lambda & \lambda \mu & \lambda \mu^2 & \dots & \lambda \mu^{N-1} & \lambda \mu^N \\
    \mu & \lambda & \lambda \mu & \dots & \lambda \mu^{N-2} & \lambda \mu^{N-1}\\
    0 & \mu & \lambda & \dots & \lambda \mu^{N-3} & \lambda \mu^{N-2} \\
    \vdots & \vdots & \vdots & \ddots  & \vdots & \vdots \\
    0 & 0 & 0 & \dots & \mu & \lambda \\
\end{bmatrix}
\end{equation}

\begin{lemma}
    The largest eigenvalue of $A$ is an upper bound on the largest non-unit eigenvalue of $T$.
\end{lemma}
\begin{proof}
    This follows from Theorem \ref{thm:block_norm_bound} of Appendix \ref{app:block_norm_bound}. 
\end{proof}

\begin{lemma}
    \label{lemma:gershgorin_block_bound}
    The largest eigenvalue of $A$ is at most $(1 + \sqrt{1 - \lambda})^2 \lambda$.
\end{lemma}
\begin{proof}
Consider the diagonal matrix\footnote{This matrix was obtained by an optimization procedure, the details of which are not necessary for this proof. Numerics suggest that the resulting bound is tight in the limit of large $N$.} with elements
$$M_{ij} = \delta_{ij}(\mu + \sqrt{\lambda})^j$$
The eigenvalues of $B = M A M^{-1}$ are the same as those of $A$. If we apply the Gershgorin circle theorem to $B$, we obtain a bound on the largest eigenvalue of
\begin{align}
\max_i \left[\lambda + (1 - \delta_{iN})(\mu + \sqrt{\lambda})\mu + \sum_{j=1}^i \lambda \left(\frac{\mu}{\mu + \sqrt{\lambda}}\right)^j\right]
\\
< (\mu + \sqrt{\lambda})\mu + \sum_{j=0}^\infty \lambda \left(\frac{\mu}{\mu + \sqrt{\lambda}}\right)^j
\\ =
(\mu + \sqrt{\lambda})\mu + \frac{\lambda}{1 - \left(\frac{\mu}{\mu + \sqrt{\lambda}}\right)}
\\ = (\mu + \sqrt{\lambda})^2
\\ = (1 + \sqrt{1 - \lambda})^2 \lambda
\end{align}
\end{proof}

Combining these lemmas, we obtain a relationship between the staircase spectral gap and the spectral gaps of a family of smaller operators:
\begin{theorem}
    \label{thm:staircase_to_single}
    Let $T_N$ be the $N$-site staircase transfer matrix and define
    \begin{equation}
        \lambda = \sup_m ||K_m - \Pi_{m+1}||_{2}
    \end{equation}
    Then
    \begin{gather}
        \estair \leq (1 + \sqrt{1 - \lambda})^2 \lambda
    \end{gather}
\end{theorem}
Note that this bound is nontrivial only when 
$$\lambda < \frac{1}{3}\left[2^\frac{1}{3}\left(13 + 3 \sqrt{33}\right)^\frac{1}{3} - 2^\frac{8}{3} \left(13 + 3 \sqrt{33}\right)^{-\frac{1}{3}} - 1\right] \approx 0.2956$$
Over this region, we have \(3.382 < (1 + \sqrt{1 - \lambda})^2 < 4\), so the bound is nearly proportional to \(\lambda\). Also note that $\Pi_{m+1}$ is the projector onto the unit eigenspace of $K_m$, so $||K_m - \Pi_{m+1}||_{2}$ is also the largest non-unit eigenvalue of $K_m$. 

\section{Properties of the eigenspaces}
We will now discuss some useful aspects of the structure of the gate and transfer operators. The first is a symmetry under \(S_t \times S_t\). 
\begin{definition}
    For \(\pi \in S_t\), define the \textit{global left-action}
    \begin{gather}
        \pi_L: \ket{\sigma_1, \sigma_2 ... \sigma_N } \rightarrow \ket{\pi \circ \sigma_1, \pi \circ \sigma_2 ... \pi \circ \sigma_N }
    \end{gather}
    and the analogous \textit{global right-action}
    \begin{gather}
        \pi_R: \ket{\sigma_1, \sigma_2 ... \sigma_N } \rightarrow \ket{\sigma_1 \circ \pi, \sigma_2 \circ \pi ... \sigma_N \circ \pi }
    \end{gather}
\end{definition}

\begin{lemma}
\label{lemma:symmetry}
The moment operator $G$ corresponding to a random unitary acting on $k$ of the sites is invariant under global left and right action:
    \[[G, \pi_R] = [G, \pi_L] = 0\]
\end{lemma}
\begin{proof}
    A $k$-site gate $G$ is a projector on to the span of the $k$-site uniform permutation states $\langle \ket{\sigma}^{\otimes k} : \sigma \in S_t\rangle$. The metric in this subspace, $\langle \sigma | \tau \rangle^k = q^{-k |\sigma^{-1}\tau|}$, is invariant under global left or right action. In addition, the space itself is closed under global left or right action. Therefore, the image of $G$, as a subspace, is preserved under global left or right action. The commutativity follows.
\end{proof}
Next we show that there is a natural embedding of transfer matrices at some moment $t$ and transfer matrices at any higher moment $t' > t$. Since we will be working with states and gates at different values of $t$, we will here note the $t$-value with a superscript, so that 
\begin{gather}
\ket{\sigma}^{(t)} = \frac{1}{\sqrt{q}^t} \sum_{\vec{i} \in \mathbb{Z}_q^t} \ket{\vec{i}} \otimes \ket{\sigma(\vec{i})}
\end{gather}
Here the tensor product is between copies of the random circuit and copies of its adjoint, not sites. Similarly we define
\begin{gather}
    G^{(t)} = \sum_{\sigma \in S_t} \Wg(\sigma^{-1} \tau, q^2)\ket{\sigma}^{(t)\otimes 2}\bra{\tau}^{(t)\otimes 2}
\end{gather}
where this tensor product is between sites, not copies.
Then
\begin{lemma}
\label{lemma:gate_t_blindness}
Let $k < t$ and \(\vec{\sigma} \in S_k^{\times N}\), where $S_k^{\times N}$ is the set of tuples$\{(\sigma_1, \sigma_2, ..., \sigma_N) : \sigma_i \in S_k\}$. Suppose
\begin{gather}
G^{(k)} \ket{\vec{\sigma}}^{(k)} = \sum_{\vec{\tau} \in S_k^N} c_{\vec{\tau}} \ket{\vec{\tau}}^{(k)}
\end{gather}
Then
\begin{gather}
    G^{(t)} \ket{\vec{\sigma}}^{(t)} = \sum_{\vec{\tau} \in S_k^N} c_{\vec{\tau}} \ket{\vec{\tau}}^{(t)}
\end{gather}
where the coefficients \(c_{\vec{\tau}}\) are the same.
\end{lemma}
\begin{proof} \textit{(Informal)}
    The basic insight here is to return to the original picture of averaging $t$ copies of a Haar random unitary over a density matrix given by a permutation state $\ket{\vec{\sigma}}^{(k)}$ in the $S_k$ symmetric group. In this case, we can cancel out the unitaries on copies on which the permutation acts trivially. A full proof may be found in Appendix \ref{app:gate_t_blindness}. 
\end{proof}

\begin{corollary}
    \label{corollary:gate_t_blindness}
    Let \(k < t\). Define \(\vecspan(S_k^{\times N}) =\vecspan\{ \ket{\sigma_1, ...\sigma_N}^{(t)}, \sigma_i \in S_k\}\)
    Then \(\vecspan(S_k^{\times N})\) is an invariant subspace of \(G_i\).
\end{corollary}
\begin{corollary}
    \label{corollary:}
    The subspace \(\pi_L \rho_R \vecspan(S_k^{\times N}) =\vecspan\{ \ket{\pi \circ \sigma_1 \circ \rho, ...\pi \circ \sigma_N \circ \rho}^{(t)}, \sigma_i \in S_k\}\) is also an invariant subspace of \(G_i\) for every \(\pi,\rho \in S_t\).
\end{corollary}
Call the direct sum of all subspaces of this form \(\vecspan(S_t^2 \circ S_k^{\times N})\). The arguments above generalize straightforwardly to non-two-site gates, so these subspaces are also invariant under the \(\Pi_i\) operators. The subspaces have a nested structure, with \(\vecspan (S_t^2 \circ S_{k-1}^N) \subset \vecspan (S_t^2 \circ S_k^N)\).
Since they are invariant subspaces of each gate, they are also invariant subspaces of any circuit. In other words, the circuit may be diagonalized separately in each subspace. 
\begin{lemma}
    $T$ is block-diagonal in \(\vecspan (S_t^2 \circ S_k^{\times N})\) and its orthogonal complement.
\end{lemma}
\begin{proof}
    Let $P$ be the orthogonal projector on to \(\vecspan (S_t^2 \circ S_k^{\times N})\). We first show that $G$ and $P$ commute. Since $\text{Im}(GP) \subseteq \vecspan (S_t^2 \circ S_k^{\times N})$, we see that $GP = PGP$. Since both $G$ and $P$ are Hermitian,
    $$PG = (GP)^\dagger = (PGP)^\dagger = PGP = GP$$
    It follows that $(I-P)GP = 0$, so each gate operator is block-diagonal. $T$ is a product of gate operators, so it inherits their block-diagonal structure.
\end{proof}
The nonzero eigenvalues of $PG^{(t)}P$ are of course exactly those of $G^{(k)}$. It follows that as we increase $t$, new eigenvalues can appear only in the orthogonal complement of \(\vecspan (S_t^2 \circ S_k^{\times N})\). To proceed, however, we will need a slightly stronger theorem. We now show that even a nonorthogonal complement will do:
\begin{theorem} 
\label{thm:eigenstate_independence}
Let $P$ be any projector (not necessarily orthogonal) onto \(\vecspan (S_t^2 \circ S_k^{\times N})\). Then $T^{(t)}$ has a block-triangular structure in $P$ and $(I-P)$ and the former block contains the same eigenvalues as $T^{(k)}$. More precisely, let \(\text{eig}^*(M)\) denote the set of nonzero eigenvalues of $M$. Then
\begin{gather}
    \label{eq:eigenvalues_are_subset}
    \text{eig}^*\left(T^{(k)}\right) = \text{eig}^* \left(P T^{(t)} P\right)
\end{gather}
and
\begin{gather}
    \label{eq:block_triangular_eigenvalues}
    \text{eig}^*\left(T^{(t)}\right) = 
    \text{eig}^* \left((I - P) T^{(t)}(I - P)\right) \cup \text{eig}^*\left(T^{(k)}\right)
\end{gather}

\end{theorem}
\begin{proof}
\textit{(Informal)}
To prove Equation \ref{eq:eigenvalues_are_subset}, we use Lemma \ref{lemma:gate_t_blindness} to show that \(P T^{(t)} P\) is essentially several copies of $T^{(k)}$. To prove Equation \ref{eq:block_triangular_eigenvalues}, we show that $T^{(t)}$ is block-triangular in the images of $P$ and $I-P$. We then use the fact that the eigenvalues of a block-triangular matrix are those of the diagonal blocks. A full proof may be found in Appendix \ref{app:eigenstate_independence}. 
\end{proof}

We now define and characterize two useful subspaces. Our analytical proofs will make use of the \textbf{deranged subspace}, while our numerical work is done in terms of the \textbf{coderanged subspace}.
\subsection{Deranged subspace}
\begin{definition}
    \label{def:deranged}
    A state $\ket{\sigma_1, ... \sigma_N}$ is called a \textbf{complete derangement state} if and only if there does not exist $i \in \{1...t\}$
     such that \(\sigma_1(i) = \sigma_2(i) = ... = \sigma_N(i)\). 
    Equivalently, $\ket{\sigma_1, ... \sigma_N}$ is a \textbf{complete derangement state} if and only if the permutations \(\sigma_1^{-1} \circ \sigma_2, \sigma_1^{-1} \circ \sigma_3, ... \sigma_1^{-1} \circ \sigma_N\) have no common fixed point. Let \(\mathfrak{D}_N^{(t)}\) denote the set of complete derangement states on $t$ copies of $N$ sites. We will usually suppress the $t$.
\end{definition}

In the two-site case, things become a bit simpler. A state \(\ket{\sigma, \tau} \in \mathfrak{D}_2\) if and only if \(\sigma^{-1} \circ \tau\) is a derangement. Let \(S_t^{(\mathfrak{D})}\) refer to the set of derangements of order $t$.

\begin{lemma}
Suppose $t \leq q$. Then \(\vecspan \left({\mathfrak{D}_N^{(t)}}\right)\) is a complement of \(\vecspan \left(S_t^2 \circ S_{t-1}^{\times N}\right)\).
\end{lemma}
\begin{proof}
    When $t \leq q$, the basis states are linearly independent. We will prove that every basis state $\ket{\sigma_1, ... \sigma_N}$ is either in $\mathfrak{D}_N^{(t)}$ or $S_t^2 \circ S_{t-1}^{\times N}$, that is, the complement of $\mathfrak{D}_N^{(t)}$ is exactly $S_t^2 \circ S_{t-1}^{\times N}$.
    
    Firstly, suppose that $\ket{\sigma_1, ... \sigma_N}$ is not in $\mathfrak{D}_N^{(t)}$, that is, there is a common fixed point $j$ among $\sigma_1^{-1} \circ \sigma_2, \sigma_1^{-1} \circ \sigma_3, ... \sigma_1^{-1} \circ \sigma_N$. Let \(\tau_{jt}\) denote the permutation which swaps copy \(j\) and copy \(t\). Then $(\tau_{jt} \sigma_1^{-1} ) \cdot \sigma_i \cdot \tau_{jt} \in S_{t-1}$ for all $i$, as this permutation must fix the $t$th element. The complement of $\mathfrak{D}_N^{(t)}$ is therefore contained in $S_t^2 \circ S_{t-1}^{\times N}$. 

    Now take a basis element $\ket{\sigma_1, ... \sigma_N} \in S_t^2 \circ S_{t-1}^{\times N}$. We must have some $\pi, \rho, \tau_1, ..., \tau_N$ such that $\sigma_i = \pi \tau_i \rho$ for all $i$. Therefore
    \begin{gather}
        \sigma_1^{-1} \sigma_i = \rho^{-1} \left(\tau_1^{-1} \tau_i\right) \rho
    \end{gather}
    is a conjugation by $\rho$ of a permutation in $S_{t-1}$ - therefore, it must leave the element $\rho^{-1}(t)$ fixed. This is true for all $i > 1$, however, so $\ket{\sigma_1, ..., \sigma_N}$ is not in $\mathfrak{D}_N^{(t)}$. Hence, the complement of $\mathfrak{D}_N^{(t)}$ also contains all of $S_t^2 \circ S_{t-1}^{\times N}$, so the two sets are identical.
\end{proof}
We will refer to \(\vecspan \left({\mathfrak{D}_N^{(t)}}\right)\) and \(\vecspan \left(S_t^2 \circ S_{t-1}^{\times N}\right)\) as the \textbf{deranged subspace} and \textbf{non-deranged subspace}, respectively.

\begin{corollary}
    \label{corollary:deranged_block_eigvals}
    Suppose $t \leq q$. Define an operator $P_{\mathfrak{D}}$ on basis states by
    \begin{gather}
        P_{\mathfrak{D}}\ket{\vec{\sigma}} = \begin{cases}
        \ket{\vec{\sigma}} & \vec{\sigma} \in \mathfrak{D}_N \\
        \mathbf{0} & \text{otherwise} \\
        \end{cases}
    \end{gather}
    Then
    \begin{gather}
    \label{eq:deranged_eigvals}
    \text{eig}^*\left(T^{(t)}\right) = 
    \text{eig}^* \left(P_{\mathfrak{D}} T^{(t)}P_{\mathfrak{D}}\right) \cup \text{eig}^*\left(T^{(t-1)}\right)
\end{gather}
\end{corollary}
\begin{proof}
    When $t \leq q$, the basis states are linearly independent, so $P_{\mathfrak{D}}$ extends to a well-defined linear operator. Furthermore, $I - P_{\mathfrak{D}}$ is a projector to the non-deranged subspace. We may thus apply Theorem~\ref{thm:eigenstate_independence} with $k = t-1$.
\end{proof}

By inductively applying Corollary~\ref{corollary:deranged_block_eigvals}, we arrive at the following:
\begin{lemma}\label{lemma:K_m_bounded_by_deranged}
When $t \leq q$, 
    \begin{gather}
        \text{eig}^*\left(K_m^{(t)}\right) = \bigcup_{k=1}^t \text{eig}^* \left(P_{\mathfrak{D}} K_m^{(k)}P_{\mathfrak{D}}\right)
    \end{gather}
    As the unit eigenvalues make up $\text{eig}^*(K_m^{(1)})$, the subleading eigenvalue of $K_m$ is bounded by the leading eigenvalues of $P_{\mathfrak{D}}K_m^{(k)} P_{\mathfrak{D}}$ for all $2 \leq k \leq t$.
\end{lemma}
It therefore suffices to check the leading eigenvalues of $K_m$ restricted to deranged subspaces. These eigenvalues are bounded by the spectral norm of $K_m$ in these subspaces with respect to any inner product. In particular, we will work in terms of operator norm induced by the \textit{basis norm} 
\begin{gather}
    \left|\left|\sum_{\sigma \in S_t} c_\sigma \ket{\sigma}\right|\right|_\text{basis}^2 \equiv \sum_{\sigma \in S_t} |c_\sigma|^2
\end{gather}
Note that this norm is not equivalent to the norm inherited from the inner product on the original Hilbert space. Then, we define the \textbf{deranged subspace basis norm}
\begin{gather}
    ||K_m||_{\mathfrak{D}} \equiv ||P_{\mathfrak{D}} K_m P_{\mathfrak{D}}||_{\text{basis}}
\end{gather}
The subleading eigenvalue of $K_m^{(t)}$ is then bounded by largest deranged subspace basis norm of $K_m^{(k)}$ over all $k \leq t$, i.e.
\begin{gather}
    ||K_m^{(t)} - \Pi_{m + 1}||_\infty \leq \max_{2 \leq k \leq t} ||K_m^{(k)}||_{\mathfrak{D}}
\end{gather}

Finally, we note that the symmetry of Lemma \ref{lemma:symmetry} applies separately within each subspace.
\begin{lemma}
    \label{lemma:deranged_symmetry}
    The subspaces $\vecspan \left(\mathfrak{D}_N^{(t)}\right)$,  \(\vecspan \left(S_t^2 \circ S_{t-1}^{\times N}\right)\), and the orthogonal complement of \(\vecspan \left(S_t^2 \circ S_{t-1}^{\times N}\right)\) are each invariant under the global right-action \(\sigma_R\) of any \(\sigma \in S_t\).
\end{lemma}
The proof is in Appendix \ref{app:deranged_symmetry}. It follows that when $t \leq q$, new eigenvalues may be classified according to the irreducible representations of $S_t$. The lemma also holds for left-action and thus for $S_t \times S_t$, but we will need only right-action.

\subsection{Co-deranged subspace}
We will now give a formal description of the orthogonal complement of \(\vecspan \left(S_t^2 \circ S_{t-1}^{\times N}\right)\) which is consistent for all $q,t$. The results of this section are not necessary for the proof of Theorem \ref{thm:gap_t_leq_q}, where $t \leq q$, but are relied on by the numerics used to establish Theorem \ref{thm:gap_finite_N} in the $t > q$ regime. 
\begin{definition}
    Define the \textbf{permutation cobasis}\footnote{It should be noted that the permutation cobasis is only a true cobasis for $t \leq q$. For $t > q$, the linear dependence of the original basis does not allow us to define a true cobasis, but we are still allowed to write down the set of states $\{\ket{\widetilde{\sigma}} : \sigma \in S_t\}$.} states by
    \begin{gather}
        \ket{\widetilde{\sigma}} = \sum_{\tau} \Wg(\sigma^{-1}\tau)\ket{\tau}
    \end{gather}
    Then the \textbf{coderanged states} are states of the form $\ket{\widetilde{\sigma}_1, ... \widetilde{\sigma}_N}$, where \(\sigma_1^{-1} \circ \sigma_2, \sigma_1^{-1} \circ \sigma_3, ... \sigma_1^{-1} \circ \sigma_N\) have no common fixed point. Denote the set of such states by $\widetilde{\mathfrak{D}}_N.$ 
\end{definition}

\begin{corollary}
    If $t \leq q$, then \(\vecspan \left({\widetilde{\mathfrak{D}}_N^{(t)}}\right)\) is the orthogonal complement of \(\vecspan \left(S_t^2 \circ S_{t-1}^{\times N}\right)\). If $t > q$, then $\vecspan \left({\widetilde{\mathfrak{D}}_N^{(t)}}\right)$ contains the orthogonal complement of \(\vecspan \left(S_t^2 \circ S_{t-1}^{\times N}\right)\).
\end{corollary}
This corollary will follow from Theorem \ref{thm:coderanged_characterization} below. It follows immediately that when $t \leq q$, a version of Lemma \ref{lemma:K_m_bounded_by_deranged} holds with $P_\mathfrak{D}$ replaced by the orthogonal projector on to the co-deranged subspace. However, when $t > q$ the co-deranged subspace has nontrivial intersection with the non-deranged subspace, so we must restrict ourselves the subspace of $\vecspan \left({\widetilde{\mathfrak{D}}_N^{(t)}}\right)$ which is orthogonal to the non-deranged subspace.

It turns out that we can find a more concrete description of this subspace. Define
\begin{gather}
    X_{\sigma \tau}(q,t) = \braket{\sigma|\widetilde{\tau}}
\end{gather}
By abuse of notation we also define the multi-site version $X(\vec{q},t) = \bigotimes_i X(q_i,t)$. 
\begin{theorem}\label{thm:coderanged_characterization}
    Define the intersection space to be the intersection of $\image(X(\vec{q},t))$ with the preimage of the coderanged subspace, 
    \begin{gather}
        \mathcal{I} = \left\{\mathbf{v} \in  \mathbb{C}^{t!^N}: \left(v_{\vec{\sigma}} = 0 \forall \vec{\sigma} \in S_t^2 \circ S_{t-1}^{\times N}\right) \wedge \left(\mathbf{v} \in \image(X)\right) \right\}
    \end{gather}
    Define the embedding
    \begin{gather}
    E(\mathbf{v}) = \sum_{\vec{\sigma} \in S_t^N} v_{\vec{\sigma}} \ket{\widetilde{\sigma_1,  \dots \sigma_N}}
    \end{gather}
    The restriction of $E$ to $\mathcal{I}$
    gives an isomorphism between $\mathcal{I}$ and the orthogonal complement of $\vecspan \left(S_t^2 \circ S_{t-1}^{\times N}\right)$.
\end{theorem}
This theorem is proven in Appendix \ref{app:coderanged_characterization}. Since $K_m$ is block-diagonal in $\vecspan \left(S_t^2 \circ S_{t-1}^{\times N}\right)$ and $E(\mathcal{I})$, this isomorphism allows us to run more efficient numerical eigensolving.

\section{Structure of the 3-Site Operator}
Our goal is to to bound the subleading eigenvalue $\lambda_m$ of the effective three-site gate operator (Fig.~\ref{fig:staircase_architecture}b)
\begin{gather}
    K_m = \Pi_m G_{m} \Pi_m 
\end{gather}
Since each $\Pi_m$ is a projector, we can restrict our search to the image of $\Pi_m$. This is a space of dimension $q^2$, spanned by the (non-orthogonal) basis
\begin{gather}
    \left\{|\sigma_1\rangle_{1}\ket{\sigma_1}_2 ... \ket{\sigma_1}_m |\sigma_2\rangle_{m+1} \big | \sigma_1, \sigma_2 \in S_t \right \}\label{eq:two_site_basis}
\end{gather}
The intermediate state $G_m \Pi_m \ket{\psi}$ lies in $\image (G_m \Pi_m)$, which is spanned by states of the form $\ket{\sigma_1}_1 ...\ket{\sigma_1}_{m-1}\ket{\sigma_2}_{m}\ket{\sigma_2}_{m+1}$. We will work in the span of the union of these two spaces
\begin{gather} 
    \label{eq:three_site_basis}
    \vecspan \left( \image [\Pi_m] \cup \image [G_m \Pi_m] \right) = \left \{ \ket{\sigma_1}_1,...\ket{\sigma_1}_{m-1}\ket{\sigma_2}_{m}\ket{\sigma_3}_{m+1} \big | \sigma_1, \sigma_2, \sigma_3 \in S_t \right\}
\end{gather}
This is an effective three-site basis. 

There is a natural isomorphism of the form $\ket{\sigma, \sigma, ...\sigma} \rightarrow \ket{\sigma}$ between the uniform permutation states on $k$ sites and the permutation states on a single site. If the original $k$ sites had local dimension $q$ and the new site has physical dimension $q^k$, then this map is also an isometry. Applying this map to the basis of Equation \ref{eq:three_site_basis}, we obtain an effective three-site basis
\begin{gather}
    \left\{\ket{\sigma_1, \sigma_2, \sigma_3} \big|\sigma_1, \sigma_2, \sigma_3 \in S_t\right\}
\end{gather}
where the first site has effective local dimension $Q_1 = q^m$ and the second and third sites have $Q_2 = Q_3 = q$. 
In this generalization, the projectors $\Pi_m$ and $G_m$ have analogous definitions: $\Pi_m$ is the uniform projector on sites $1$ and $2$, while $G_m$ is the uniform projector on sites $2$ and $3$.

\subsection{Factorization of \(K_m\)}
Let's review the sequence of reductions we have passed through. We began by looking for the dominant eigenvalue of \(\Pi_n G_n \Pi_n\) in $\mathbb{C}^{2tqn}$ for all $ n \leq N$. By averaging, we restricted our search space to $\vecspan \left(S_t^{\times n}\right)$. By the argument above we have reduced it to $\vecspan  \left(S_t^{\times 3}\right)$. By Corollary~\ref{corollary:deranged_block_eigvals}, it furthermore suffices to check $\vecspan  \left(\mathcal{D}_3^{(k)}\right)$ for all \(k \leq t\). And by Lemma~\ref{lemma:symmetry}, it in fact suffices to check each isotypic component under the global left-action of $S_t$ separately. We will now introduce a convenient basis for this last space.

Irreducible representations of $S_t$ are labeled by the partitions of $t$. For a particular partition $\nu$, let \(V_\nu: S_t \rightarrow \text{End}(\mathbb{R}^d)\) be the corresponding irreducible representation in some basis.
\begin{lemma}
\label{lemma:isotype_basis}
Consider a global right action operator in an irreducible representation of $S_t$ indexed by a partition $\nu$:
\begin{gather}
    R_\nu^{ij} = \sum_{\rho \in S_t} V_\nu(\rho^{-1})^{ij} \rho_R
\end{gather}
A basis for the isotypic component of \(\image(\Pi_m)\) corresponding to \(\nu\) is given by
\begin{gather}
    \left\{\ket{e_{\nu}^{ij,\sigma}} = R_{\nu}^{ij} \ket{I,I, \sigma} \bigg|\sigma \in S_t, i,j \in \{1...d_{\nu}\} \right\}
\end{gather}
Furthermore, the subset of such states for which $\sigma$ is a complete derangement is a basis for the isotypic component of \(\vecspan  \left(\mathfrak{D}_3^{(t)}\right) \cap \image(\Pi_m)\) corresponding to \(\nu\).
\end{lemma}

Note that $\ket{e_\nu^{ij, \sigma}} \in \vecspan \left(\mathfrak{D}_3^{(t)}\right)$ if and only if $\ket{I, I, \sigma}$ is a complete derangement. The fact that $\ket{e_\nu^{ij, \sigma}}$ lies in the isotypic component corresponding to $\nu$ follows from the Schur orthogonality relations, while the fact that they span the space follows from the form of the canonical idempotent
\begin{gather}
    P_\nu = \frac{\chi_\nu(1)}{t!} \sum_\tau \chi_\nu(\tau^{-1}) \tau_R
\end{gather}
A full proof may be found in Appendix \ref{app:isotype_basis}. 

Using this basis, we can obtain a convenient factorization of the gate operator. 
\begin{lemma} \label{lemma:main_product}
    In the basis defined by Lemma \ref{lemma:isotype_basis},
    \begin{align}
        K_m \ket{e_{\nu}^{ij,\sigma}} = \sum_{k,\xi \in \mathbb{D}_t} M_{m,\xi,\sigma}^{\nu,kj} \ket{e_{\nu}^{ik,\xi}}
    \end{align}
    Note that \(\nu\), $i$, and $j$ indices are \textit{not} implicitly summed over. Furthermore,
    \begin{align} \label{eq:main_product_v2}
        M^{\nu}_m = D_\nu^{-1} W(Q_1 Q_2) D(Q_2) C(Q_1) D_\nu W(Q_2 Q_3) D(Q_2) C(Q_3)
    \end{align}
    where
    \begin{gather}
        Q_1 = q^m \\
        Q_2 = Q_3 = q \\
        W(d)_{\sigma, \tau}^{ij} = \Wg(\sigma^{-1} \tau, d) \delta_{ij}\\
        D(d)_{\sigma, \tau}^{ij} = d^{-|\sigma|}\delta_{\sigma, \tau} \delta_{ij}\\
        C(d)_{\sigma, \tau}^{ij} = d^{-|\sigma^{-1} \tau|} \delta_{ij}\\
        D^{ij}_{\nu,\sigma, \tau} = V_\nu(\sigma)^{ij} \delta_{\sigma, \tau}
    \end{gather}
    and the product above is matrix multiplication.
\end{lemma}
\begin{proof}
Consider the action of \(K_m = \Pi_m G_m \Pi_m\) on one of the basis states. Since $\Pi_m$ and $G_m$ commute with right-action, 
\begin{gather}
    \Pi_m G_m \Pi_m \ket{e_{\nu}^{ij, \sigma}} = R_{\nu}^{ij} \Pi_m G_m \Pi_m \ket{I, \sigma}
\end{gather}
We can compute
\begin{align*}
    \Pi_m G_m \Pi_m \ket{I, I, \sigma} &= \Pi_m G_m \ket{I, I, \sigma} \\
    &= \Pi_m \sum_{\tau, \pi \in S_dst} \Wg(\tau^{-1} \pi, Q_2 Q_3) Q_2^{-|\pi|} Q_3^{-|\sigma^{-1} \pi|} \ket{I,\tau , \tau}
\end{align*}
Meanwhile,
\begin{align*}
    \Pi_m \ket{I,\tau , \tau} &= \sum_{\xi, \omega \in S_t} \Wg(\xi^{-1} \omega, Q_1 Q_2) Q_1^{-|\omega|} Q_2^{-|\tau^{-1} \omega|} \ket{\xi,\xi , \tau}
    \\ &= \sum_{\xi, \omega \in S_t} \Wg(\xi^{-1} \omega, Q_1 Q_2) Q_1^{-|\omega|} Q_2^{-|\tau^{-1} \omega|} \xi_R \ket{I,I , \tau \xi^{-1}}
\end{align*}
Combining these, we see
\begin{align}
    K_m \ket{I,I,\sigma} &= \sum_{\xi, \omega,\tau,\pi} \Wg(\xi^{-1} \omega, Q_1 Q_2) Q_1^{-|\omega|}Q_2^{-|\tau^{-1} \omega|}\Wg(\tau^{-1} \pi, Q_2 Q_3) Q_2^{-|\pi|} Q_3^{-|\sigma^{-1} \pi|} \xi_R \ket{I,I , \tau \xi^{-1}} 
    \n &= 
    \sum_{\xi, \omega,\tau,\pi} \Wg(\tau^{-1}\xi \omega, Q_1 Q_2) Q_1^{-|\omega|}Q_2^{-|\tau^{-1} \omega|}\Wg(\tau^{-1} \pi, Q_2 Q_3) Q_2^{-|\pi|} Q_3^{-|\sigma^{-1} \pi|} (\xi^{-1}\tau)_R \ket{I,I , \xi} \quad (\xi \rightarrow \xi^{-1}\tau)\n
    &= \sum_{\xi, \omega,\tau,\pi} \Wg(\xi \omega^{-1}, Q_1 Q_2) Q_1^{-|\omega^{-1}\tau|}Q_2^{-| \omega|}\Wg(\tau^{-1} \pi, Q_2 Q_3) Q_2^{-|\pi|} Q_3^{-|\sigma^{-1} \pi|} (\xi^{-1}\tau)_R \ket{I,I , \xi} \quad (\omega \rightarrow \omega^{-1}\tau)
\end{align}
In Lemma \ref{lemma:pr_properties}, we will show
\begin{align}
    R_\nu^{ij} \sigma_R &= \sum_k R_\nu^{ik} V_\nu(\sigma)^{kj} 
\end{align}
Applying this result, 
\begin{align}
    K_m \ket{e_{\nu}^{ij,\sigma}} &= R_\nu^{ij} K_m\ket{I,I,\sigma}\n
    &= \sum_{\xi, \omega,\tau,\pi} \Wg(\xi \omega^{-1}, Q_1 Q_2) Q_1^{-|\omega^{-1}\tau|}Q_2^{-| \omega|}\Wg(\tau^{-1} \pi, Q_2 Q_3) Q_2^{-|\pi|} Q_3^{-|\sigma^{-1} \pi|} R_\nu^{ij} (\xi^{-1}\tau)_R \ket{I,I , \xi}\n 
    &= \sum_{\xi, \omega,\tau,\pi}  \Wg(\xi \omega^{-1}, Q_1 Q_2) Q_1^{-|\omega^{-1}\tau|}Q_2^{-| \omega|}\Wg(\tau^{-1} \pi, Q_2 Q_3) Q_2^{-|\pi|} Q_3^{-|\sigma^{-1} \pi|} \sum_{k} V_\nu(\xi^{-1}\tau )^{kj} R_\nu^{i k} \ket{I,I , \xi}\n
    &= \sum_{\xi, \omega,\tau,\pi}  \Wg(\xi \omega^{-1}, Q_1 Q_2) Q_1^{-|\omega^{-1}\tau|}Q_2^{-| \omega|}\Wg(\tau^{-1} \pi, Q_2 Q_3) Q_2^{-|\pi|} Q_3^{-|\sigma^{-1} \pi|} \sum_{k,\ell} V_\nu(\xi^{-1})^{k\ell}V_\nu(\tau )^{\ell j} R_\nu^{i k} \ket{I,I , \xi}\n
    &= \sum_{\xi, \omega,\tau,\pi,k,\ell} \left[(D_\nu^{-1})^{k\ell}_{\xi \xi} W(Q_1 Q_2)^{\ell \ell}_{\xi \omega}  D(Q_2)^{\ell \ell}_{\omega \omega} C(Q_1)^{\ell \ell}_{\omega \tau} (D_\nu)^{\ell j}_{\tau \tau} W(Q_2 Q_3)^{j j}_{\tau \pi} D(Q_2)^{j j }_{\pi \pi} C(Q_3)^{j j }_{\pi \sigma}\right] R_\nu^{i k} \ket{I,I , \xi}\n 
    &= \sum_{k} M^{\nu, k j}_{m, \xi \sigma} \ket{e_\nu^{i k, \xi}} \label{eq:sum_over_k_for_m_matrix}
\end{align}
\end{proof}

We can reduce the isotypical subspace even further by considering global left-action symmetries as well as global right-action. In Appendix~\ref{app:lr_irreps}, we consider the nature of this reduced subspace as well as the form of the gate $K_m$ in this subspace. For now, it is sufficient to prove Theorem~\ref{thm:gap_t_leq_q} with just global right-action symmetries.

\subsection{Computing \(\lambda\) when $t=2$}\label{section:t2_3_site_op}
We will begin by computing the subleading eigenvalue of $K_m$ in the special case $t=2$. 
\begin{theorem}
    \label{thm:eigval_t2}
    When $t = 2$, 
    \[\lambda_m = \frac{q^2 (q^{2m} - 1)}{(q^2 + 1)(q^{2m + 2} - 1)} \leq \frac{1}{1 + q^2}\]
\end{theorem}
\begin{proof}
    The only partitions of $2$ are $(1,1)$ and $2$, corresponding to the alternating and trivial representations, respectively. Each is one-dimensional, so the sum over $k$ in Equation~\ref{eq:sum_over_k_for_m_matrix} contains only the $k=1$ term. The two elements of the group are identity and the $(1,2)$ transposition. These are referred to below as $I$ and $S$, respectively. 

    In the trivial irrep, we have
    
    \[W(d) = \frac{1}{1 - d^{-2}}\begin{bmatrix}
        1 & -1/d \\
        -1/d & 1
    \end{bmatrix}\]
    \[D(d) = \begin{bmatrix}
        1 & 0 \\
        0 & 1/d
    \end{bmatrix}\]
    \[C(d) = \begin{bmatrix}
        1 & 1/d \\
        1/d & 1
    \end{bmatrix}\]
    \[D_{(2)}^{11} =\begin{bmatrix}
        1 & 0 \\
        0 & 1 \\
    \end{bmatrix}\]
    and
    \begin{align}
        M^{(2)} &= D_{(2)}^{-1} 
        W(q^{m+1}) D(q) C(q^m) D_{(2)} W(q^2) D(q) C(q)
        \\ &= \begin{bmatrix}
            1 & \frac{q(q^m + 1)(q^{m + 2}-1)}{(q^2 + 1)(q^{2m+2}-1)} \\
            0 & \frac{q^2 (q^{2m} - 1)}{(q^2 + 1)(q^{2m + 2} - 1)}
        \end{bmatrix}
    \end{align}
    Since this matrix is upper-triangular, its eigenvalues are the diagonal elements. The unit eigenvalue corresponds of course to the uniform \(\ket{II} + \ket{SS}\) state. The second diagonal element, which corresponds to the deranged subspace, is the eigenvalue we are looking for. 
    
    In the alternating case, 
    \[D_{(1,1)} =\begin{bmatrix}
        1 & 0 \\
        0 & -1 \\
    \end{bmatrix}\]
    but because the minus sign appears twice, it cancels out and we find \(M^{(1,1)} = M^{(2)}\).

    We may now compute
    \begin{align}
        \sup_m \frac{q^2 (q^{2m} - 1)}{(q^2 + 1)(q^{2m + 2} - 1)} &= \frac{q^2}{q^2 + 1} \sup_m\frac{q^{2m} - 1}{q^{2m + 2} - 1} \\
        &= \frac{q^2}{q^2 + 1} \lim_{m \rightarrow \infty} \frac{q^{2m} - 1}{q^{2m + 2} - 1} \\
        &= \frac{1}{q^2 + 1}
    \end{align}
    where in the second line we use the fact that \(q > 1\).
\end{proof}

\subsection{Computing $\lambda$ when $t > 2$}
The remainder of this work is devoted to establishing the following fact.
\begin{theorem}
    \label{thm:eigval_t_leq_q}
    When $t \leq q$, 
    \begin{gather}
        ||K_m ||_\infty \leq \frac{1}{1 + q^2}
    \end{gather}
\end{theorem}
\begin{proof}
From Corollary \ref{corollary:deranged_block_eigvals}, it suffices to check the eigenvalues of $K_m$ in each deranged block separately. From Theorem \ref{thm:eigval_t2} ($t=2$) and Lemmas~\ref{lemma:Kbound_t34} ($t \leq 6$), \ref{lemma:numerical_bound} ($7 \leq t \leq 28$), and \ref{lemma:analytic_bound} ($t > 28$), we will see there is no value of $t$ for which the eigenvalue of $K_m$ in the deranged block exceeds $\frac{1}{1 + q^2}$.
\end{proof}

\section{Bounding the complete derangement block}
\begin{figure}
    \centering
    \begin{tikzpicture}
        \begin{scope}
            \node[anchor=north west,inner sep=0] (image_a) at (0.4,0)
            {\includegraphics[width=0.4\columnwidth]{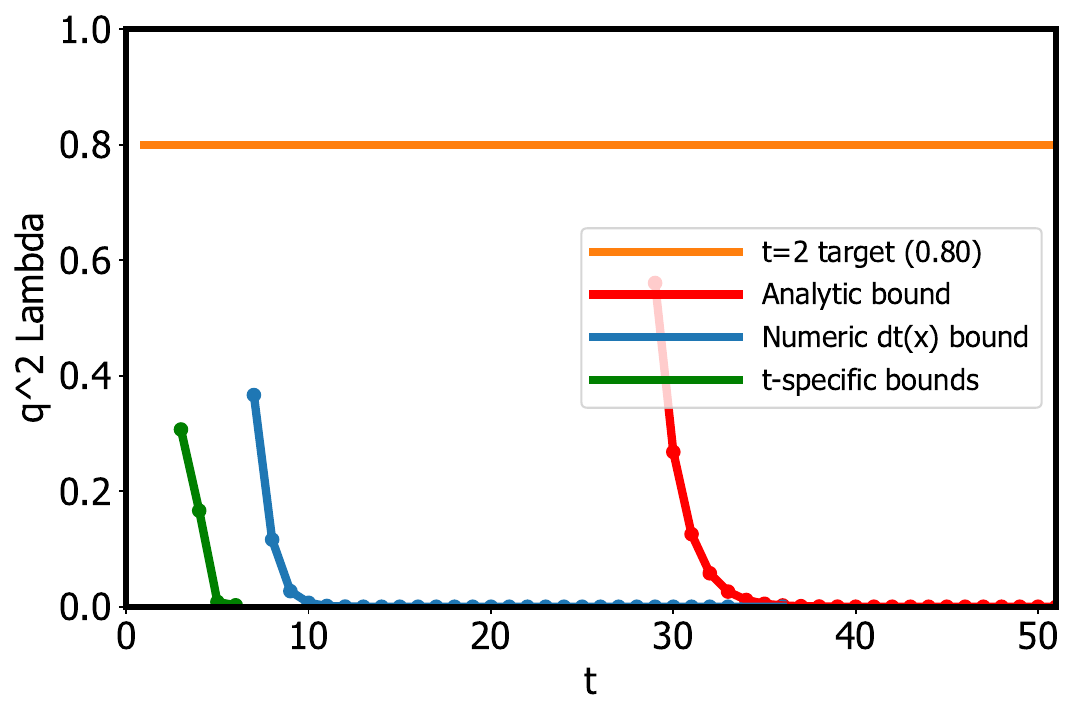}};
        \end{scope}
    \end{tikzpicture}
    \vspace*{-0.4cm}
    \caption{Bounds on $q^2 ||K_m||_{\mathfrak{D}}$, that is, the largest eigenvalue of the complete derangement block scaled by $q^2$, through Lemmas~\ref{lemma:analytic_bound} (red line), \ref{lemma:numerical_bound} (blue line), and \ref{lemma:Kbound_t34} (green line). For each $t > 2$, at least one of our bounds implies $||K_m||_{\mathfrak{D}} < \frac{4}{5} q^2 \leq \frac{1}{q^2+1}$. It follows that the $t=2$ eigenvalue always dominates.}
    \label{fig:derangement_eigs}
\end{figure}

In this section we consider the case $2 < t \leq q$. Our goal is to obtain bounds on the largest eigenvalue of $K_m$ in the completely deranged subspace. Let \(||\cdot ||_\mathfrak{D}\) refer to the spectral norm of an operator's restriction to the deranged subspace.  There are several cases which must be handled independently.
\begin{lemma}\label{lemma:analytic_bound}
    When $t \leq q$ and $t > 28$,
    \begin{gather}
        ||K_m ||_{\mathfrak{D}} \leq 
         \frac{t^2}{2} e^{\frac{1}{2t^2} + 2} \left[\left(\frac{1}{2} + \frac{3}{t}\right)^t + \left(\frac{4}{\sqrt{t}}\right)^t\right] \left(\frac{t}{q}\right)^{t-2} q^{-2} \leq \frac{1}{1 + q^2}
    \end{gather}
\end{lemma}
\begin{proof}
    We have
    \begin{gather}
        ||K_m ||_{\mathfrak{D}} \leq \max_\nu ||D_{\nu}^{-1}|| \cdot
        ||W(q^{m+1})|| \cdot ||D(q)C(q^m)P_{\mathfrak{D}}|| \cdot ||D_{\nu}|| \cdot ||W(q^2)|| \cdot ||D(q)C(q)P_{\mathfrak{D}}||
    \end{gather}
    We have $||D_{\nu}|| = ||D_{\nu}^{-1}|| = 1$, since the representations are unitary. Lemmas \ref{lemma:weingarten_bound}, \ref{lemma:dc_bound}, and \ref{lemma:dt_bound_analytic} give bounds on the remaining terms.
\end{proof}
\begin{lemma}\label{lemma:numerical_bound}
    When $t \leq q$ and $7 \leq t \leq 28$, 
    \begin{gather}
        ||K_m ||_{\mathfrak{D}} \leq \frac{0.367}{q^2} \left(\frac{t}{q}\right)^{t-2} \leq \frac{1}{1 + q^2}
    \end{gather}
\end{lemma}
\begin{proof}
This will follow from Lemmas \ref{lemma:weingarten_bound}, \ref{lemma:dc_bound}, and \ref{lemma:dt_bound_numerical}.
\end{proof}

\begin{lemma}
    \label{lemma:Kbound_t34}
    When $t \leq q$ and $2 < t < 7$, 
    \begin{gather}
        ||K_m||_{\mathfrak{D}} \leq \begin{cases}
        \frac{0.307}{q^2}\left(\frac{t}{q}\right)^2 & t = 3 \\
        \frac{0.167}{q^2}\left(\frac{t}{q}\right)^2 & t = 4 \\
        \frac{8.27\times 10^{-3}}{q^2}\left(\frac{t}{q}\right)^4 & t = 5 \\
        \frac{2.48 \times 10^{-3}}{q^2}\left(\frac{t}{q}\right)^4 & t = 6
        \end{cases}
        \leq \frac{1}{1 + q^2}
    \end{gather}
\end{lemma}
A proof is found in Section \ref{app:Kbound_t34} - in short, we can obtain these bounds by diagonalizing the half-matrix $W(Q_1 Q_2) D(Q_2) C(Q_1)$ directly, using Perron-Frobenius to replace matrix elements with more positive expressions that are largely independent of $Q_1, Q_2$. 
The bounds of Lemmas~\ref{lemma:analytic_bound} through \ref{lemma:Kbound_t34} are illustrated in Fig.~\ref{fig:derangement_eigs}, where $||K_m||_{\mathfrak{D}}$ is compared to $\frac{0.8}{q^2}$, which in turn is bounded by $\frac{1}{1+q^2}$ for all $q \geq 2$. 
\subsection{Bounds by factorization of $K_m$}
We saw in Lemma \ref{lemma:main_product} that $K_m$ may be decomposed as a product of several operators. Here we will bound the norms of the terms in that product independently. The product of these norms of the terms then gives a bound on the norm of the product. We will work in terms of operator norm induced by the \textbf{basis norm} 
\begin{gather}
    \left|\left|\sum_{\sigma \in S_t} c_\sigma \ket{\sigma}\right|\right|_\text{basis}^2 \equiv \sum_{\sigma \in S_t} |c_\sigma|^2
\end{gather}
Note that this norm is not equivalent to the norm inherited from the inner product on the original Hilbert space. 

The first term in the product is of the form $W(Q_1 Q_2)$. 
\begin{lemma} \label{lemma:weingarten_bound}
    For \(3 \leq t \leq q\) and \(Q_1, Q_2 \geq q\), 
    \begin{gather}
        ||W(Q_1 Q_2)||_\text{basis} \leq e
    \end{gather}
\end{lemma}
A proof may be found in Appendix \ref{app:weingarten_bound}. Note that this bound holds in either the deranged subspace or the whole space.

\begin{lemma} \label{lemma:dc_bound}
Define the \textit{derangement polynomial}
    \begin{gather}\label{eq:dt_definition_copy}
        d_t(x) = \sum_{\sigma \in S_t^{(\mathfrak{D})}} \sum_{\tau \in S_t} x^{|\sigma^{-1} \tau| + |\tau|}
    \end{gather}
    Let $Q_1, Q_2 \geq q$. Then
    \begin{gather}
        ||D(Q_2) C(Q_1)||_{\mathfrak{D}, \text{basis}}  \leq \sqrt{d_t(q^{-2})}
    \end{gather}
\end{lemma}
The proof involves using the Frobenius norm to upper-bound the operator norm of $DC$. A formal proof may be found in Appendix~\ref{app:proof_of_dc_bound}.

To make use of this result, we establish the following pair of bounds on $d_t$. 
\begin{lemma}
\label{lemma:dt_bound_analytic}
When $t \leq q$, 
\begin{gather}
    d_t(q^{-2}) \leq \frac{t^2}{2} e^{\frac{1}{2t^2}} \left[\left(\frac{1}{2} + \frac{3}{t}\right)^t + \left(\frac{4}{\sqrt{t}}\right)^t\right] \left(\frac{t}{q}\right)^{t-2} \frac{1}{q^2}
\end{gather}
\end{lemma}
The proof of this lemma may be found in Appendix~\ref{app:proof_of_dt_bound}. This bound is adequate only for $t \geq 28$. For smaller $t$, we may evaluate the derangement polynomial numerically to obtain the following bound.
\begin{lemma} \label{lemma:dt_bound_numerical}
When $7 \leq t \leq 28$ and $q \geq t$, 
\begin{gather}
    d_t(q^{-2}) \leq \frac{0.0496}{q^2}\left(\frac{t}{q}\right)^{t-2} 
\end{gather}
\end{lemma}
This lemma is proved in Appendix \ref{app:dt_numerical_proof}.
This bound, in turn, is actually adequate only for $t \geq 7$. For $t < 7$, we will need to consider the whole product together instead of bounding it term-by-term.

\subsection{Bounds for $2 < t < 7$ and $q \geq t$ (Lemma \ref{lemma:Kbound_t34})}
\label{app:Kbound_t34}
The above bounds cover the case $7 \leq t \leq q$. It remains to prove Lemma \ref{lemma:Kbound_t34}. Our goal is to bound the dominant eigenvalue of $M_{\sigma,\tau}^{ij}$, as defined in Lemma \ref{lemma:main_product}. An upper bound on the dominant eigenvalue of $M$ is given by the operator norm of $M$ with respect to \textit{any} vector norm. 
With respect to the basis norm, $D_\nu$ is a unitary operator. 
It follows that 
\begin{gather}
    ||M_m||_{\text{basis}} \leq ||W(Q_1 Q_2)D(Q_2)C(Q_1)||_{\text{basis}} ||W(Q_2 Q_3)D(Q_2)C(Q_3)||_{\text{basis}} = ||H(Q_1, Q_2)||_{\text{basis}} ||H(Q_3, Q_2)||_{\text{basis}}
\end{gather}
where we define the half-operator 
\begin{gather} \label{eq:matrix_half_product}
    H(Q_1, Q_2) = W(Q_1 Q_2)D(Q_2)C(Q_1)
\end{gather}
In order to bound this operator, consider the Weingarten function. There exist polynomials $f_t$ and $g_{\sigma^{-1} \tau}$ such that
\begin{gather}
    \Wg\left(\sigma^{-1}\tau, d\right) = \frac{g_{\sigma^{-1} \tau}(d^{-1})}{f_t(d^{-1})}
\end{gather}
In Appendix \ref{app:elementwise_hbound}, we establish the following elementwise upper bound.
\begin{lemma}
\label{lemma:elementwise_hbound}
Let $Q_1, Q_2 \geq q \geq t$. Let $f_t, g_{\sigma^{-1}\tau}$ be polynomials such that
\begin{gather}
    \Wg\left(\sigma^{-1}\tau, d\right) = \frac{g_{\sigma^{-1} \tau}(d^{-1})}{f_t(d^{-1})}
\end{gather}
Define the polynomial 
\begin{gather}
    h_{\sigma \tau}(x,y) = \sum_{\rho \in S_t} x^{|\rho \sigma^{-1}|} y^{|\rho|} g_{\rho^{-1} \tau}(xy)
\end{gather}
and let $h_{\sigma \tau}^{(i,j)}$ be the corresponding coefficients. Define the matrix
\begin{gather}
    \overline{H}(t)_{\sigma \tau} = \frac{1}{f_t(t^{-2})}\sum_{i,j} \left|h_{\sigma \tau}^{(ij)}\right| t^{-(i+j)}
\end{gather}
Then we may bound the matrix elements of $H$ by
\begin{gather}
        |H_{\sigma \tau}(Q_1, Q_2)| \leq \left(\frac{t}{q}\right)^{\left\lceil \frac{t}{2}\right \rceil} 
        \overline{H}(t)_{\sigma \tau} 
\end{gather}
\end{lemma}
For sufficiently small $t$, we can find the operator norm of $\overline{H}$ numerically. By the Perron-Frobenius theorem, this gives us a bound on the operator norm of $H$. This bound will be of the form
$h(t) q^{-\left\lceil \frac{t}{2}\right \rceil}$. The $t=2$ eigenvalue will dominate so long as 
\begin{gather}
    h^2(t) \leq \frac{t^{2 \left\lceil \frac{t}{2}\right \rceil}}{1 + t^2}
\end{gather}
Numerical values of $||\overline{H}(t)||_\text{basis}^{\mathfrak{D}}$ are presented in Table \ref{tab:ft_t2}. We see that the bound is satisfied for all $t \leq 6$, which completes the proof of Theorem \ref{thm:gap_t_leq_q}.

\section{Conclusion}
We have established a nearly-optimal $t$-independent, $N$-independent lower bound on the spectral gap when $t \leq q$. This implies improvements to known bounds on approximate $t$-design depths of various architectures, including the brickwork\cite{Brandao2016,Haferkamp2022,Chen2024b}, arbitrary well-connected architectures\cite{Belkin2023}, and specific rapidly-scrambling architectures\cite{Schuster2024,LaRacuente2024}. Although the asymptotics of these results do not change, our improvements to the constant factors are often dramatic. The scaling with $t$ is also improved for the regime $t \leq q < 6t^2$ in which the bound of ref.~\onlinecite{Haferkamp2021} did not apply. We also show that the resulting bound on the 1D brickwork $t$-design depth has nearly optimal small-$\epsilon$ behavior.

Our first technical innovation is our strategy to relate the 1D staircase SEV to a sequence of 3-site operators. This differs from previously established methods, and we believe it has the potential to produce even stronger results with further work. As a primary example, we would like to extend our bound to the case $t > q$. From numerical evidence (see Appendix~\ref{app:gap_finite_N}), it appears that the eigenvalues of $K_m$ remain bounded by $\frac{1}{q^2+1}$ even when $t > q$, with the sole exception of the case $t = 4$, $q = 2$, $m = 2$. There are, roughly speaking, two obstacles to extending our proof to the $t > q$ regime. The first is the degeneracy of the permutation basis at $t > q$. This is partially resolved by Theorem \ref{thm:coderanged_characterization}, but the intersection space $\mathcal{I}$ is difficult to work with. Second, we bound $||K||$ by factorization in terms of $||W|| \cdot ||DC||$. But for $t > q$ both $||W||$ and $||DC||$ grow large even as $||WDC||$ remains small, so bounds by factorization cannot succeed.

We would also like to extend our bounds to the regime of non-integer $q < t$. In this setting, the metric acquires negative eigenvalues, so additional techniques may be needed. Bounding these regions would allow us to extend the depth bounds on one-dimensional systems to all architectures through ref.~\onlinecite{Belkin2023}.

Recent works have shown that architectures very similar to the 1D brickwork form $\epsilon$-approximate $t$-designs in depth $O(\log N)$. However, a the best known scaling for the 1D brickwork itself remains $O(N)$. In this work we use our reduction to $3$-site operators only to bound the spectral gap of the 1D brickwork, but the implied block-triangular structure may be helpful for other analyses of the behavior of the moment operator. The deranged-subspace and irreducible-representation structures of both the original brickwork and of the three-site operators could provide additional insight to this behavior as well. If these decompositions can be used to study the entire eigenspectrum, it could lead to a better understanding of the true $N$-scaling of the 1D brickwork approximate $t$-design depth.

\begin{acknowledgments}
This material is based upon work supported by the U.S. Department of Energy, Office of Science, National Quantum Information Science Research Centers. 
\end{acknowledgments}

\bibliographystyle{unsrt}
\bibliography{references2}

\begin{thebibliography}{10}

\bibitem{Boixo2018}
Sergio Boixo, Sergei~V. Isakov, Vadim~N. Smelyanskiy, Ryan Babbush, Nan Ding, Zhang Jiang, Michael~J. Bremner, John~M. Martinis, and Hartmut Neven.
\newblock {Characterizing quantum supremacy in near-term devices}.
\newblock {\em Nature Physics}, 14(6):595--600, 2018.

\bibitem{Arute2019}
Frank Arute, Kunal Arya, Ryan Babbush, Dave Bacon, Joseph~C. Bardin, Rami Barends, Rupak Biswas, Sergio Boixo, Fernando~G.S.L. Brandao, David~A. Buell, Brian Burkett, Yu~Chen, Zijun Chen, Ben Chiaro, Roberto Collins, William Courtney, Andrew Dunsworth, Edward Farhi, Brooks Foxen, Austin Fowler, Craig Gidney, Marissa Giustina, Rob Graff, Keith Guerin, Steve Habegger, Matthew~P. Harrigan, Michael~J. Hartmann, Alan Ho, Markus Hoffmann, Trent Huang, Travis~S. Humble, Sergei~V. Isakov, Evan Jeffrey, Zhang Jiang, Dvir Kafri, Kostyantyn Kechedzhi, Julian Kelly, Paul~V. Klimov, Sergey Knysh, Alexander Korotkov, Fedor Kostritsa, David Landhuis, Mike Lindmark, Erik Lucero, Dmitry Lyakh, Salvatore Mandr{\`{a}}, Jarrod~R. McClean, Matthew McEwen, Anthony Megrant, Xiao Mi, Kristel Michielsen, Masoud Mohseni, Josh Mutus, Ofer Naaman, Matthew Neeley, Charles Neill, Murphy~Yuezhen Niu, Eric Ostby, Andre Petukhov, John~C. Platt, Chris Quintana, Eleanor~G. Rieffel, Pedram Roushan, Nicholas~C. Rubin, Daniel Sank, Kevin~J.
  Satzinger, Vadim Smelyanskiy, Kevin~J. Sung, Matthew~D. Trevithick, Amit Vainsencher, Benjamin Villalonga, Theodore White, Z.~Jamie Yao, Ping Yeh, Adam Zalcman, Hartmut Neven, and John~M. Martinis.
\newblock {Quantum supremacy using a programmable superconducting processor}.
\newblock {\em Nature}, 574(7779):505--510, 2019.

\bibitem{Movassagh2019}
Ramis Movassagh.
\newblock {Quantum supremacy and random circuits}.
\newblock 2019.

\bibitem{Brown2018}
Adam~R. Brown and Leonard Susskind.
\newblock {Second law of quantum complexity}.
\newblock {\em Physical Review D}, 97(8), 2018.

\bibitem{Hayden2007}
Patrick Hayden and John Preskill.
\newblock {Black holes as mirrors: Quantum information in random subsystems}.
\newblock {\em Journal of High Energy Physics}, 2007(9):120, 9 2007.

\bibitem{Skinner2019}
Brian Skinner, Jonathan Ruhman, and Adam Nahum.
\newblock {Measurement-Induced Phase Transitions in the Dynamics of Entanglement}.
\newblock {\em Physical Review X}, 9(3), 2019.

\bibitem{Bao2020}
Yimu Bao, Soonwon Choi, and Ehud Altman.
\newblock {Theory of the phase transition in random unitary circuits with measurements}.
\newblock {\em Physical Review B}, 101(10), 3 2020.

\bibitem{Jafferis2022}
Daniel Jafferis, Alexander Zlokapa, Joseph~D. Lykken, David~K. Kolchmeyer, Samantha~I. Davis, Nikolai Lauk, Hartmut Neven, and Maria Spiropulu.
\newblock {Traversable wormhole dynamics on a quantum processor}.
\newblock {\em Nature}, 612(7938):51--55, 2022.

\bibitem{Brandao2016}
Fernando~G.S.L. Brand{\~{a}}o, Aram~W. Harrow, and Michał Horodecki.
\newblock {Local Random Quantum Circuits are Approximate Polynomial-Designs}.
\newblock {\em Communications in Mathematical Physics}, 346(2):397--434, 2016.

\bibitem{Haferkamp2022}
Jonas Haferkamp.
\newblock {Random quantum circuits are approximate unitary t-designs in depth O(nt5+o(1))}.
\newblock {\em Quantum}, 6, 2022.

\bibitem{Chen2024b}
Chi-Fang Chen, Jeongwan Haah, Jonas Haferkamp, Yunchao Liu, Tony Metger, and Xinyu Tan.
\newblock {Incompressibility and spectral gaps of random circuits}.
\newblock 2024.

\bibitem{Schuster2024}
Thomas Schuster, Jonas Haferkamp, and Hsin-Yuan Huang.
\newblock {Random unitaries in extremely low depth}.
\newblock 2024.

\bibitem{LaRacuente2024}
Nicholas LaRacuente and Felix Leditzky.
\newblock {Approximate Unitary k-Designs from Shallow, Low-Communication Circuits}.
\newblock 2024.

\bibitem{Belkin2023}
Daniel Belkin, James Allen, Soumik Ghosh, Christopher Kang, Sophia Lin, James Sud, Fred Chong, Bill Fefferman, and Bryan~K. Clark.
\newblock {Approximate t-designs in generic circuit architectures}.
\newblock 2023.

\bibitem{Dalzell2022}
Alexander~M. Dalzell, Nicholas Hunter-Jones, and Fernando~G.S.L. Brand{\~{a}}o.
\newblock {Random Quantum Circuits Anticoncentrate in Log Depth}.
\newblock {\em PRX Quantum}, 3(1):10333, 2022.

\bibitem{Deneris2024}
Andrew~E. Deneris, Pablo Bermejo, Paolo Braccia, Lukasz Cincio, and Marco Cerezo.
\newblock {Exact spectral gaps of random one-dimensional quantum circuits}.
\newblock 2024.

\bibitem{Haferkamp2021}
Jonas Haferkamp and Nicholas Hunter-Jones.
\newblock {Improved spectral gaps for random quantum circuits: Large local dimensions and all-to-all interactions}.
\newblock {\em Physical Review A}, 104(2), 2021.

\bibitem{Hunter-Jones2019}
Nicholas Hunter-Jones.
\newblock {Unitary designs from statistical mechanics in random quantum circuits}.
\newblock 2019.

\bibitem{Znidaric2022}
Marko Znidaric.
\newblock {Solvable non-Hermitian skin effect in many-body unitary dynamics}.
\newblock 5 2022.

\bibitem{Collins2006}
Benoît Collins and Piotr {\'{S}}niady.
\newblock {Integration with respect to the Haar measure on unitary, orthogonal and symplectic group}.
\newblock {\em Communications in Mathematical Physics}, 264(3):773--795, 2006.

\bibitem{Bensa2021}
Jas Bensa and Marko Znidaric.
\newblock {Fastest local entanglement scrambler, multistage thermalization, and a non-Hermitian phantom}.
\newblock 1 2021.

\bibitem{Zinn-Justin2010}
Paul Zinn-Justin.
\newblock {Jucys-Murphy elements and weingarten matrices}.
\newblock {\em Letters in Mathematical Physics}, 91(2):119--127, 2010.

\bibitem{Kostenberger2021}
Georg K{\"{o}}stenberger.
\newblock {Weingarten Calculus}.
\newblock 2021.

\end{thebibliography}

\begin{appendices}
\section{Relationship between spectral gap and $t$-design depth}
This material is generally standard, but we include it here for completeness. 

\begin{definition}
    An ensemble $\varepsilon$ of random quantum circuits is an additive-error $\epsilon$-approximate $t$-design if the diamond norm distance between its $t$-fold channel and the global Haar ensemble's $t$-fold channel is less than $\epsilon$:
    \begin{gather}
        ||\Phi_\varepsilon^{(t)} - \Phi_\text{Haar}^{(t)}||_\diamond \leq \epsilon
    \end{gather}
\end{definition}

\begin{definition}
    The frame potential $\mathcal{F}_\varepsilon^{(t)}$ of a random quantum circuit ensemble $\varepsilon$ is a $4t$-th moment of its unitaries:
    \begin{gather}
        \mathcal{F}_\varepsilon^{(t)} = \int_{\varepsilon^{\otimes 2}} |\tr(U_\varepsilon^\dagger V_\varepsilon)|^{2t} \dd U_\varepsilon \dd V_\varepsilon
    \end{gather}
\end{definition}

\begin{lemma}
    \label{lemma:diamond_norm_to_frame_potential}
    The diamond norm distance between two $t$-fold channels is bounded by the difference in their frame potentials:
    \begin{gather}
        ||\Phi_\varepsilon^{(t)} - \Phi_\text{Haar}^{(t)}||_\diamond^2 \leq q^{2Nt} \left(\mathcal{F}_\varepsilon^{(t)} - \mathcal{F}_\text{Haar}^{(t)}\right)
    \end{gather}
\end{lemma}
The frame potential is the Frobenius norm of the moment operator, 
\begin{gather}
    \mathcal{F}_\varepsilon^{(t)} = \tr\left(\Phi_\varepsilon^{(t) \dagger} \Phi_\varepsilon^{(t)}\right)
\end{gather}

Let $L_O$ and $L_E$ refer to the ``even'' and ``odd'' layers of the 1D brickwork architecture, respectively. The transfer matrix associated with a single period of the architecture is then \(T_{1DB} = L_O L_E\). 
\begin{lemma}
\label{lemma:brickwork_is_quasihermitian}
Let \(\mathcal{F}_\varepsilon(\ell)\) be the frame potential of the \(\ell\)-layer 1D brickwork architecture, \(\ell \geq 2\). Then
\[\mathcal{F}_\varepsilon(\ell) = \tr \left((L_E L_O L_E)^{\ell - 1}\right)\]
\end{lemma}
\begin{proof}
    The first observation is that
    \begin{gather}
        \tr \left((L_E L_O L_E)^{\ell - 1}\right) = \tr \left((L_O L_E^2)^{\ell - 1}\right) = \tr \left((L_O L_E)^{\ell - 1}\right)
    \end{gather}
    
    Next, suppose \(\ell\) is even. Then the moment operator is \(\Phi_\epsilon = (L_O L_E)^{\ell/2}\) and
    \begin{gather}
        \mathcal{F}_\varepsilon(\ell) = \tr \left((L_O L_E)^{\ell/2} (L_O L_E)^{\dagger \ell/2}\right)
    \end{gather}
    Now, $L_O$ is a tensor product of $G$ operators, each of which is an orthogonal projection operator. It follows that $L_O$ is an orthogonal projection operator, which implies both $L_O^\dagger = L_O$ and $L_O^2 = L_O$. The same goes for $L_E$. Using these, we can see
    \begin{align}
        \tr \left(\Phi_\epsilon \Phi_\epsilon^{\dagger}\right)
        &= \tr \left((L_O L_E)^\frac{\ell}{2}(L_E L_O)^\frac{\ell}{2}
        \right)
        \\ &=
        \tr \left((L_O L_E)^\frac{\ell}{2}L_E (L_O L_E)^{\frac{\ell}{2} - 1} L_O
        \right)
        \\ &=
        \tr \left((L_O L_E)^\frac{\ell}{2}(L_O L_E)^{\frac{\ell}{2} - 1}
        \right)
    \end{align}
    where in the third line we have used the cylic property of the trace and the fact that \(L_O L_O = L_O\). 

    Now suppose \(\ell\) is odd. Then \(\Phi_\epsilon^t = T_{1DB}^{\frac{\ell - 1}{2}} L_O\), and we can compute
    \begin{align}
        \tr \left(\Phi_\epsilon \Phi_\epsilon^{\dagger}\right) 
        &= \tr \left((L_O L_E)^\frac{\ell - 1}{2} L_O L_O (L_E L_O)^\frac{\ell - 1}{2}
        \right)
        \\ &= \tr \left((L_O L_E)^\frac{\ell - 1}{2}L_O^2 L_E (L_O L_E)^{\frac{\ell - 1}{2} - 1} L_O\right)
        \\ &= \tr \left((L_O L_E)^\frac{\ell - 1}{2}(L_O L_E)^1 (L_O L_E)^{\frac{\ell - 1}{2} - 1}\right)
    \end{align}
    where again the third line uses the cyclic property of the trace. 
\end{proof}

\begin{lemma}
\label{lemma:frame_potential_from_eigenvalue}
Let \(\lambda_*\) be the magnitude of the largest non-unit eigenvalue of \(L_E L_O L_E\). Then
\begin{gather}
    \mathcal{F}_\varepsilon(\ell) \leq \mathcal{F}_\text{Haar} + q^{2Nt}\lambda_*^{\ell - 1}
\end{gather}
\end{lemma}
\begin{proof}
    From Lemma \ref{lemma:brickwork_is_quasihermitian},
    \begin{gather}
        \mathcal{F}_\varepsilon(\ell) = \tr \left((L_E L_O L_E)^{\ell - 1}\right)
    \end{gather}
    as $L_E$ and $L_O$ are individually Hermitian, so is $L_E L_O L_E$. Therefore 
    \begin{gather}
        \tr \left((L_E L_O L_E)^{\ell - 1}\right) = \sum_j \lambda_j^{\ell-1}
    \end{gather}
    where $j$ runs over all eigenvalues of $L_E L_O L_E$. There are $t!$ unit eigenvalues, which become $\mathcal{F}_\text{Haar}$, and every other eigenvalue is bounded by $\lambda_*$. As the dimension of the system is $q^{2Nt}$, there are $q^{2Nt}$ eigenvalues in total, so 
    \begin{align}
        \mathcal{F}_\varepsilon(\ell) &\leq t! + (q^{2Nt} - t!) \lambda_*^{\ell-1}\n
        &\leq t! + q^{2Nt}\lambda_*^{\ell-1}
    \end{align}
\end{proof}

\begin{lemma}
    \label{lemma:brickwork_power_gap}
    Let \(\lambda(\ell,N,q,t)\) be the largest non-unit eigenvalue of the $\ell$-layer brickwork. Then
    \[\lambda(\ell,N,q,t) = \lambda(2, N,q,t)^{\lfloor \ell/2\rfloor}\]
\end{lemma}
\begin{proof}
    Suppose first that $\ell$ is even. Then the $\ell$-layer moment operator is the $(\ell/2)$\textsuperscript{th} power of the $2$-layer moment operator. It follows that all of its eigenvalues are the same. \\
    Now suppose $\ell$ is odd. Then we may cyclically permute the final $L_E$ to the beginning, where it cancels against the initial $L_E$ via $L_E^2 = L_E$. The spectrum is thus the same as that of the $(\ell-1)$-layer moment operator.
\end{proof}

\begin{lemma}
    \label{lemma:hermitianized_brickwork_is_staircase}
    The eigenspectra of \(L_E L_O L_E\) and the staircase transfer matrix $\Tstair$ are the same.
\end{lemma}
\begin{proof}
    The eigenvalues of a product of matrices are invariant under cyclic permutations, so the spectra of $L_E L_O L_E$ and $L_O L_E^2$ are the same. But $L_E^2 = L_E$, so the spectra of $L_E L_O L_E$ and $L_O L_E$ are the same. We complete the proof by noting that $L_E L_O$ is related to the staircase transfer matrix by cyclic permutations of the individual gates.
\end{proof}

\subsection{Obtaining the approximate $t$-design depth from the spectral gap}\label{app:depth_from_eigval}

In order to obtain a bound on the approximate $t$-design depth from an eigenvalue gap in the transfer matrix, we use the following:

\begin{theorem}
    \label{thm:depth_from_eigval}
     The $N$-site 1D brickwork architecture with local Hilbert space dimension $q$ forms an additive-error $\epsilon$-approximate $t$-design after at most
    \begin{gather}
    1 + \frac{2}{\log \frac{1}{\estair(N,q,t)}}\left(2 N t \log q + \log \frac{1}{\epsilon}\right)
    \end{gather}
    layers and a multiplicative-error\footnote{The multiplicative-error $t$-design condition, as can be seen in ref.~\onlinecite{Brandao2016, Chen2024b, Schuster2024}, requires $(1-\epsilon)\Psi_{\text{Haar}} \preceq \Psi_{\varepsilon} \preceq (1+\epsilon) \Psi_{\text{Haar}}$, where $\Psi_{\nu}$ is the operator formed by taking the Choi-Jamiolkowski isomorphism of $\Phi_\nu$.} $\epsilon$-approximate $t$-design after at most
    \begin{gather}
    \ell \geq 2\left\lceil \frac{1}{2\log \frac{1}{\estair(N,q,t)}} \left(2Nt\log q + \log \frac{1}{\epsilon}\right) \right\rceil
    \end{gather} 
    layers. Here $\estair(N,q,t)$ is the magnitude of the largest eigenvalue of \(\Tstair\). 
\end{theorem}
\begin{proof}
    We first prove the additive-error bound. Combining the bounds of Lemmas \ref{lemma:frame_potential_from_eigenvalue},  \ref{lemma:hermitianized_brickwork_is_staircase}, and \ref{lemma:diamond_norm_to_frame_potential}, we see
    \begin{gather}
    ||\Phi_\varepsilon - \Phi_\text{Haar}||_\diamond^2 \leq  q^{2Nt} \left(q^{2Nt}|\estair|^{\ell - 1}\right)
    \end{gather}
    The diamond-norm distance is guaranteed to be below \(\epsilon\) so long as 
    \begin{gather}
    \epsilon \geq q^{2Nt}\estair^\frac{\ell - 1}{2}
    \end{gather}
    or equivalently
    \begin{gather}
    \log \frac{1}{\epsilon} \leq -2 N t \log q  + \frac{\ell - 1}{2} \log \frac{1}{\estair}
    \end{gather}
    Solving for \(\ell\) gives the desired result.

    For the multiplicative-error bound, we can use Lemma 4 of ref.~\onlinecite{Brandao2016} to see that the 1D brickwork forms a multiplicative-error $\epsilon$-approximate $t$-design when 
    \begin{gather}
        \epsilon \geq q^{2Nt}\lambda(\ell,N,q,t)
    \end{gather}
    We then apply Lemmas \ref{lemma:brickwork_power_gap} and \ref{lemma:hermitianized_brickwork_is_staircase} to write this as
    \begin{gather}
        \epsilon \geq q^{2Nt}\estair^{\lfloor \ell/2 \rfloor}
    \end{gather}
    and rearrange to obtain
    \begin{gather}
        \ell \geq 2\left\lceil \frac{1}{2\log \frac{1}{\estair}} \left(2Nt\log q + \log \frac{1}{\epsilon}\right) \right\rceil
    \end{gather}
\end{proof}
Multiplicative error $\epsilon$ implies additive error at most $2\epsilon$, so the second bound may give a tighter additive-error bound for some parameter values.

\subsection{Proof of Corollary \ref{corollary:depth_bound_upper_lower}}
\label{app:depth_bound_upper_lower}

\begin{corollary*}
    (Restatement of Corollary \ref{corollary:depth_bound_upper_lower})
    There exists a constant $C(N,q,t)$ such that
    \[\ell_*(N,q,t,\epsilon) = \left(C(N,q,t) + o_{\epsilon}(1)\right) \log \frac{1}{\epsilon}\]
    as \(\epsilon \rightarrow 0\). If $t \leq q $, then
    \begin{gather}
        \left[\log \frac{q^2 + 1}{2q} + \log\frac{1}{\cos \frac{\pi}{N}}\right]^{-1} \leq C(N,q,t) \leq \left[\log \frac{q^2 + 1}{2q} + \log\frac{2}{1 + \sqrt{1 + \frac{1}{q^2}}}\right]^{-1} \label{eq:C_upper_lower_bounds}
    \end{gather}
\end{corollary*}
We first note that the lower and upper bounds on $C(N,q,t)$ in Equation~\ref{eq:C_upper_lower_bounds} are equal to $\frac{2}{\log\left(\frac{1}{1-\Delta_{UB}}\right)}$ and $\frac{2}{\log\left(\frac{1}{1-\Delta_{LB}}\right)}$ respectively, where $\Delta_{UB}$ is the upper bound on the spectral gap from Theorem~\ref{thm:upper_bound} and $\Delta_{LB}$ is the lower bound on the spectral gap from Theorem~\ref{thm:gap_t_leq_q}. Therefore, this corollary naturally follows from Theorems~\ref{thm:gap_t_leq_q} and \ref{thm:upper_bound} once we prove the following lemma:
\begin{lemma}
    Suppose the 1D brickwork transfer matrix $T_{1DB} = L_O L_E$ has a spectral gap of $\Delta$. Then that brickwork reaches an $\epsilon$-approximate $t$-design at $\ell_*$ layers, where
    \begin{gather}
        \ell_* = \left(C + o_\epsilon(1)\right)\log\frac{1}{\epsilon}\label{eq:ell_exact_value} \\
        C = \frac{2}{\log\left(\frac{1}{1-\Delta}\right)}. 
    \end{gather}
\end{lemma}
\begin{proof}
    We will prove this lemma by providing equal upper and lower bounds for $\ell_*$ in Equation~\ref{eq:ell_exact_value}.

    \textit{Upper bound:} This proof follows similarly to the previous section. We have from Lemma~\ref{lemma:diamond_norm_to_frame_potential} that
    \begin{gather}
        ||\Phi_\varepsilon^{(t)} - \Phi_\text{Haar}^{(t)}||_\diamond^2 \leq q^{2Nt} \left(\mathcal{F}_\varepsilon^{(t)} - \mathcal{F}_\text{Haar}^{(t)}\right),
    \end{gather}
    and from Lemma~\ref{lemma:frame_potential_from_eigenvalue} that
    \begin{gather}
        \mathcal{F}_\varepsilon(\ell) \leq \mathcal{F}_\text{Haar} + q^{2Nt}\lambda_*^{\ell - 1}
    \end{gather}
    where $\lambda_* = 1-\Delta$. Therefore
    \begin{gather}
        ||\Phi_\varepsilon^{(t)} - \Phi_\text{Haar}^{(t)}||_\diamond \leq q^{2Nt} (1-\Delta)^{\frac{\ell-1}{2}}.
    \end{gather}
    Right before the $\epsilon$-approximate $t$-design depth (i.e. at depth $\ell_*-1$), the left-hand side is still greater than $\epsilon$, so
    \begin{gather}
        \epsilon \leq q^{2Nt} (1-\Delta)^{\frac{(\ell_*-1)-1}{2}}\\
        \log \frac{1}{\epsilon} \geq -2Nt \log q + \frac{\ell_*-2}{2} \log\left( \frac{1}{1-\Delta}\right)\\
        \ell_* \leq 2 + \frac{2}{\log\left(\frac{1}{1-\Delta}\right)}\left(2Nt \log q + \log \frac{1}{\epsilon}\right)
    \end{gather}
    from which the upper bound version of Equation~\ref{eq:ell_exact_value} follows once we collate all lower-order terms into $o_\epsilon(1)$.

    \textit{Lower bound:} By definition of the diamond norm,
    \begin{gather}
        ||\Phi_\varepsilon^{(t)} - \Phi_\text{Haar}^{(t)}||_\diamond \geq \frac{||(\Phi_\varepsilon^{(t)} - \Phi_\text{Haar}^{(t)})\rho||_1}{||\rho||_1}
    \end{gather}
    for all operators $\rho$. Suppose our 1D brickwork had an even number of layers, i.e. it consists of $\frac{\ell}{2}$ copies of the transfer matrix $T_{1DB}$. Now consider the subleading eigenstate $\rho_*$ of the transfer matrix:
    \begin{gather}
        ||T_{1DB}\rho_*||_1 = \lambda_* ||\rho_*||_1 = (1-\Delta) ||\rho_*||_1\\
        ||\Phi_\text{Haar}^{(t)} \rho_*||_1 = 0
    \end{gather}
    We therefore have
    \begin{gather}
        ||\Phi_\varepsilon^{(t)} - \Phi_\text{Haar}^{(t)}||_\diamond \geq (1-\Delta)^{\frac{\ell}{2}}.
    \end{gather}
    At the $\epsilon$-approximate $t$-design depth $\ell_*$, the left-hand side is at most $\epsilon$, so
    \begin{gather}
        \epsilon \geq (1-\Delta)^{\frac{\ell_*}{2}}\\
        \log \frac{1}{\epsilon} \leq \frac{\ell_*}{2} \log\left(\frac{1}{1-\Delta}\right)\\
        \ell_* \geq \frac{2}{\log\left(\frac{1}{1-\Delta}\right)} \log \frac{1}{\epsilon}
    \end{gather}
    from which the lower bound version of Equation~\ref{eq:ell_exact_value} follows.
\end{proof}

\section{Assorted proofs}
\subsection{Matrix bound from block decomposition (Theorem \ref{thm:block_norm_bound})}
\label{app:block_norm_bound}
\begin{theorem}
\label{thm:block_norm_bound}
    Let $P_i$ be a complete orthogonal set of projection operators. Let $T$ be a matrix. Define $A_{ij} = ||P_i T P_j||_{\infty}$. The largest eigenvalue of $A$ gives an upper bound on the largest eigenvalue of $T$.
\end{theorem}
\begin{proof}
    Let \(v\) be an eigenvector of \(T\) with eigenvalue \(c\). Define
    \[x_i = ||P_i v||\]
    \[y_i = ||P_i T v||\]
    \[z_i = \sum_{j} A_{ij} x_i\]
    and let \(a\) be the largest eigenvalue of \(A\).
    We must have \(y_i = cx_i\). By assumption \(y_i \leq z_i\), so \(c \leq \frac{z_i}{x_i}\). In fact, since this holds independently for every $i$,
    \begin{gather}
        c \leq \min_i \frac{z_i}{x_i}
    \end{gather}
    We will now show that
    \[\min_i \frac{z_i}{x_i} \leq a\]
    Define 
    \[f(\vec{x}) = \min_i \frac{z_i}{x_i}\]
    We wish to show that $\max_{\vec{u}} f(\vec{u}) \leq a$. We can compute
    \begin{gather}
        \frac{\partial f(\vec{x})}{\partial x_i} = \begin{cases}
            \min_{j: \frac{z_j}{x_j} = f(\vec{x})} A_{ij} & \frac{z_i}{x_i} > f(\vec{x}) \\
            - \frac{1}{x_i^2}\sum_{k \ne i} A_{ik}x_k & \frac{z_i}{x_i} = f(\vec{x})
        \end{cases}
    \end{gather}
    (Note that this is only a subgradient, as $f$ is not differentiable everywhere.) A maximum of $f$ can occur only if every element of the subgradient is nonpositive. Since $A_{ij} \geq 0$, this happens only when $\frac{z_i}{x_i} = f(\vec{x})$ for every $i$. In other words, $f$ can be maximized only when $A\vec{x} = f(\vec{x}) \vec{x}$, which implies that $\vec{x}$ is an eigenvector of $A$ with eigenvalue $f(\vec{x})$. It follows that $f(\vec{x}) \leq a$. We saw earlier that $c \leq f(\vec{x})$, from which we obtain $c \leq a$.
\end{proof}

\subsection{Proof of Lemma \ref{lemma:gate_t_blindness}}
\label{app:gate_t_blindness}
\begin{lemma*}
(Restatement of Lemma \ref{lemma:gate_t_blindness})
Let $k < t$ and \(\vec{\sigma} \in S_k^{\times N}\), where $S_k^{\times N}$ is the set of tuples$\{(\sigma_1, \sigma_2, ..., \sigma_N) : \sigma_i \in S_k\}$. Suppose
\begin{gather}
G^{(k)} \ket{\vec{\sigma}}^{(k)} = \sum_{\vec{\tau} \in S_k^N} c_{\vec{\tau}} \ket{\vec{\tau}}^{(k)}
\end{gather}
Then
\begin{gather}
    G^{(t)} \ket{\vec{\sigma}}^{(t)} = \sum_{\vec{\tau} \in S_k^N} c_{\vec{\tau}} \ket{\vec{\tau}}^{(t)}
\end{gather}
where the coefficients \(c\) are the same.
\end{lemma*}
\begin{proof}
Consider \(\sigma \in S_k\). Partition a multi index \(\vec{i} \in \mathbb{Z}_q^t\) into its first \(k\) entries \(\vec{i}_{< k}\) and last \(t - k\) entries \(\vec{i}_{\geq k}\). Then
\begin{align}
    \ket{\sigma}^{(t)} 
    &= \frac{1}{\sqrt{q}^t} \sum_{\vec{i} \in \mathbb{Z}_q^t} \ket{\vec{i}_{< k}} \otimes \ket{\vec{i}_{\geq k}} \otimes \ket{\sigma(\vec{i}_{<k})} \otimes \ket{\vec{i}_{\geq k}}
    \\
    &= \left(\frac{1}{\sqrt{q}^k} \sum_{\vec{i} \in \mathbb{Z}_q^k} \ket{\vec{i}} \otimes \ket{\sigma(\vec{i})}\right) \otimes \left(\frac{1}{\sqrt{q}^{t-k}} \sum_{\vec{i} \in \mathbb{Z}_q^{t-k}} \ket{\vec{i}} \otimes \ket{\vec{i}}\right) \\
    &= \ket{\sigma}^{(k)} \otimes \ket{I}^{(t-k)}
\end{align}
where the tensor product is again over copies. Let us now apply a random unitary to a pair of sites belonging to the basis element $\ket{\sigma, \tau}^{(t)}$ with $\sigma, \tau \in S_k$. We obtain
\begin{align}
U^{\otimes t} \otimes U^{*\otimes t} \ket{\sigma, \tau}^{(t)} &= (U^{\otimes k} \otimes U^{*\otimes k})\ket{\sigma, \tau}^{(k)} \otimes (U^{\otimes (t-k)} \otimes U^{*\otimes (t-k)})\ket{I, I}^{(t-k)}
\\
&= (U^{\otimes k} \otimes U^{*\otimes k})\ket{\sigma, \tau}^{(k)} \otimes (U^{\otimes (t-k)}U^{\dagger \otimes (t-k)} \otimes I^{\otimes (t-k)})\ket{I, I}^{(t-k)}
\\
&= (U^{\otimes k} \otimes U^{*\otimes k})\ket{\sigma, \tau}^{(k)} \otimes \ket{I, I}^{(t-k)}
\end{align}
where in the second line we have used the transpose trick. Averaging over the Haar measure, we see that
\begin{gather}
    G^{(t)}\ket{\sigma, \tau}^{(t)} = \left(G^{(k)}\ket{\sigma, \tau}^{(k)}\right) \otimes \ket{I, I}^{(t-k)}
\end{gather}
The lemma above then follows from the linearity of \(G\).
\end{proof}

\subsection{Proof of Lemma \ref{lemma:deranged_symmetry}}
\label{app:deranged_symmetry}
\begin{lemma*}
    (Restatement of Lemma \ref{lemma:deranged_symmetry})
    The subspaces $\vecspan \left(\mathfrak{D}_N^{(t)}\right)$,  \(\vecspan \left(S_t^2 \circ S_{t-1}^{\times N}\right)\), and the orthogonal complement of \(\vecspan \left(S_t^2 \circ S_{t-1}^{\times N}\right)\) are each invariant under the global right-action \(\sigma_R\) of any \(\sigma \in S_t\).
\end{lemma*}
\begin{proof}
    We first show that the set of complete derangement states is invariant under left-action. Consider Definition \ref{def:deranged}. Under some left-action \(\tau_L\), we find \(\sigma_1^{-1}\sigma_2 \rightarrow (\tau \sigma_1)^{-1} (\tau \sigma_2) = \sigma_1^{-1}\sigma_2\), so the set of permutations \((\tau \sigma_1)^{-1} (\tau \sigma_2) ... (\tau \sigma_1)^{-1} (\tau \sigma_N)\) have no common fixed point if and only if the original set also had no common fixed point.

    We now show the same under right-action. Under right-action, \(\sigma_1^{-1}\sigma_2 \rightarrow \tau^{-1} \sigma_1^{-1} \sigma_2 \tau\). If \(j\) is a fixed point of the former, then \(\tau(j)\) is a fixed point of the latter. It follows that the property of having a common fixed point is invariant under right-action. From this we see that the complete derangement subspace is invariant under both left- and right-action.

    The invariance of \(\vecspan \left(S_t^2 \circ S_{t-1}^{\times N}\right)\) is immediate from its definition. Because the representation is unitary, the orthogonal complement of any invariant subspace is also invariant.
\end{proof}

\subsection{Proof of Lemma \ref{lemma:isotype_basis}}
\label{app:isotype_basis}
The projector to the isotypic component corresponding to $\nu$ is the \textbf{canonical idempotent} given by
\begin{gather}
    P_\nu = \frac{\chi_\nu(1)}{t!} \sum_\tau \chi_\nu(\tau^{-1}) \tau_R
\end{gather}
Recall also the operator 
\begin{gather}
    R_\nu^{ij} = \sum_{\rho \in S_t} V_\nu(\rho^{-1})^{ij} \rho_R
\end{gather}
We will first establish a few properties of these operators.
\begin{lemma}
\label{lemma:pr_properties}
The operators $R_{\nu}^{ij}$ obey the following formulas:
\begin{align}
    R_{\nu}^{ij} \tau_R &= \sum_{k} R_{\nu}^{ik} V_\nu(\tau)^{kj} \\
    \tau_R R_{\nu}^{ij}  &= \sum_{k} V_\nu(\tau)^{ik} R_{\nu}^{kj} \\
    R_{\mu}^{ij}R_\nu^{kl} &= \frac{t!}{\chi_\nu(1)} \delta_{\mu \nu} \delta_{ik} R_\nu^{jl} \\
    P_\mu R_{\nu}^{ij} &= \delta_{\mu \nu} R_{\nu}^{ij}
\end{align}
\end{lemma}
\begin{proof}
We compute
\begin{align*}
    R_{\nu}^{ij} \tau_R &= \sum_{\rho \in S_t} V_\nu(\rho^{-1})^{ij} \rho_R \tau_R 
    \\&= \sum_{\rho \in S_t} V_\nu(\rho^{-1})^{ij} (\tau \rho)_R 
    \\& = \sum_{\pi \in S_t} V_\nu(\pi^{-1} \tau)^{ij} \pi_R
    \\& = \sum_{\pi,k} V_\nu(\tau)^{kj} V_\nu(\pi^{-1} \tau)^{ik} \pi_R
    \\& = \sum_{k} V_\nu(\tau)^{kj} R_{\nu}^{ik}
\end{align*}
Similarly
\begin{align*}
    \tau_R R_{\nu}^{ij} &= \sum_{\rho \in S_t} V_\nu(\rho^{-1})^{ij} (\rho \tau)_R 
    = \sum_{\pi,k} V_\nu(\tau)^{ik} V_\nu(\pi^{-1} \tau)^{kj} \pi_R
\end{align*}
Meanwhile
\begin{align*}
    R_{\nu}^{ij}R_\mu^{kl} &= \sum_{\rho} V_{\nu}^{ij}(\rho) \rho_R R_\mu^{kl}
    \\ &=
    \sum_{\rho, n} V_{\nu}^{ij}(\rho) V_{\mu}(\rho)^{kn} R_\mu^{nl}
\end{align*}
The Schur Orthogonality relations tell us
\begin{gather}
    \sum_{\rho} V_{\nu}^{ij}(\rho) V_{\mu}(\rho)^{kn} = \frac{t!}{\chi_\nu(1)} \delta_{\mu \nu}\delta_{ik}\delta_{nj}
\end{gather}
Substituting this formula in allows us to do both sums and complete the proof of the third formula.

For the fourth formula,
\begin{align}
    P_\mu R_{\nu}^{ij} &= \frac{\chi_\mu(1)}{t!} \left(\sum_{\tau \in S_t} \chi_\mu(\tau^{-1}) \tau_R\right) \left(\sum_{\rho \in S_t} V_\nu(\rho^{-1})^{ij} \rho_R \right) 
    \n
    &= \frac{\chi_\mu(1)}{t!} \sum_{\tau,\rho} \chi_\mu(\tau^{-1}) V_\nu(\rho^{-1})^{ij} (\rho \tau)_R 
    \n
    &= \frac{\chi_\mu(1)}{t!} \sum_{\tau,\pi} \chi_\mu(\tau^{-1}) V_\nu(\tau \pi^{-1})^{ij} \pi_R
    \n
    &= \frac{\chi_\mu(1)}{t!} \sum_{\tau,\pi,k,\ell} V_\mu(\tau^{-1})^{kk} V_\nu(\tau)^{i\ell} V_\nu( \pi^{-1})^{\ell j} \pi_R 
    \n
    &= \sum_{\pi,k,\ell} \delta_{\mu \nu} \delta_{ik} \delta_{\ell k} V_\nu(\pi^{-1})^{\ell j} \pi_R  \label{eq:line_uses_got} \\
    &= \delta_{\mu \nu} \sum_{\rho} V_\nu(\rho^{-1})^{i j} \rho_R
\end{align}
where we have used the Great Orthogonality Theorem on Line (\ref{eq:line_uses_got}).
\end{proof}

\begin{lemma*}
(Restatement of Lemma \ref{lemma:isotype_basis})
A basis for the isotypic component of \(\vecspan  \left(\mathfrak{D}_3^{(t)}\right)\) corresponding to \(\nu\) is given by
\begin{gather}
    \left\{\ket{e_{\nu}^{ij,\sigma}} = R_{\nu}^{ij} \ket{I, \sigma} \bigg|\sigma \in \mathbb{D}_t, i,j \in \{1...d_{\nu}\} \right\}
\end{gather}
\end{lemma*}
\begin{proof}
From the fourth property of Lemma \ref{lemma:pr_properties}, 
\begin{gather}
    P_\nu R_{\nu}^{ij} \ket{I, \sigma} = R_{\nu}^{ij} \ket{I, \sigma}
\end{gather}
and so the basis elements lie in the isotypic components. It remains to show that the basis is complete. Applying the canonical idempotent to an arbitrary state, we obtain
\begin{align}
    P_\nu \ket{\pi, \rho} &= \frac{\chi_\nu(1)}{t!} \sum_\tau \chi_\nu(\tau^{-1})\ket{\tau \pi, \tau \rho}
    \\ 
    &= \frac{\chi_\nu(1)}{t!} \sum_\tau \chi_\nu(\pi \tau^{-1} )\ket{\tau, \tau \pi^{-1}\rho} \\
    &= \frac{\chi_\nu(1)}{t!} \sum_{\tau,i,j} V_\nu^{ij}(\tau^{-1}) V_{\nu}^{ji}(\pi)\ket{\tau, \tau \pi^{-1}\rho}
    \\
    &= \frac{\chi_\nu(1)}{t!} \sum_{i,j}  V_{\nu}^{ji}(\pi)\sum_\tau V_\nu^{ij}(\tau^{-1})\ket{\tau, \tau \pi^{-1}\rho}
    \\
    &= \frac{\chi_\nu(1)}{t!} \sum_{i,j}  V_{\nu}^{ji}(\pi)\ket{e_\nu^{ij, \pi^{-1}\rho}}
\end{align}
If the original state \(\ket{\pi,\rho}\) was completely deranged, then \(\pi^{-1}\rho\) is a derangement. The basis is thus complete for any linear combination of completely deranged states that lies in the isotypic component corresponding to \(\nu\).
\end{proof}

\subsection{Proof of bound on Weingarten function (Lemma \ref{lemma:weingarten_bound})}
\label{app:weingarten_bound}
\begin{lemma*} 
(Restatement of Lemma \ref{lemma:weingarten_bound})
    For \(3 \leq t \leq q\) and \(d \geq q^2\), 
    \begin{gather}
        ||W(d)||_\infty \leq e
    \end{gather}
    where the norm is the spectral norm with respect to either the permutation-basis inner product or the Hilbert space inner product.
\end{lemma*}
\begin{proof}
    The Weingarten matrix $W(d)$ has an established eigenvalue decomposition \cite{Zinn-Justin2010,Kostenberger2021}, from its $\mathbb{C}[S_t]$ algebra representation:
\begin{gather}
    W(d) = \sum_\lambda c_\lambda^{-1}(d) P_\lambda
\end{gather}
where $\lambda$ are partitions $\{\lambda_1 ... \lambda_k\}$ of $t$, $P_\lambda$ are a set of complete orthogonal idempotents in the algebra, and 
\begin{gather}
    c_\lambda(d) = d^{-t}\prod_{i=1}^k \prod_{j=1}^{\lambda_i} (d + j - i)
\end{gather}
Since the idempotents are orthogonal to each other, so are the eigenvectors of $W$, so the maximum norm of any unit vector passed through $W$ is bounded by its largest eigenvalue, which corresponds to the partition that minimizes $c_\lambda(Q)$ (keeping $c_\lambda$ nonzero, as the Weingarten function is a pseudo-inverse). For $t \leq d$, this would be the $(1,1,1,...1)$ partition, with an eigenvalue of 
\begin{align*}
    c_{1^t}(d)^{-1} &= \frac{d^t}{\prod_{i=1}^t [d-(i-1)]}\\
    &= \frac{d^t (d-t)!}{d!}
\end{align*}
We now apply a version of Stirling's approximation here:
\begin{gather}
    \sqrt{2\pi n}\left(\frac{n}{e}\right)^n e^{\frac{1}{12n+1}} < n! < \sqrt{2\pi n}\left(\frac{n}{e}\right)^n e^{\frac{1}{12n}}
\end{gather}
giving
\begin{align}
    c_{1^t}(d)^{-1} &\leq d^t \sqrt{\frac{d - t}{d}} e^t \left(\frac{(d - t)^{d - t}}{d^d} \right) e^{\frac{1}{12(d -t)+1} - \frac{1}{12d}}\n
    &= e^t \left(1-\frac{t}{d} \right)^{d-t+\frac{1}{2}} e^{\frac{12t+1}{[12(d -t)+1][12d]}}\n
    &\leq e^t \left(e^{-\frac{t}{d}}\right)^{d - t+\frac{1}{2}} e^{\frac{t+1}{12(d-t)^2}}\n
    &\leq e^{\frac{t^2}{d} - \frac{t}{2d} + \frac{t+1}{12(d-t)^2}}
\end{align}
For $t \leq \sqrt{d}$, this is at most $e$. Because $W$ is Hermitian with respect to both the basis inner product and the Hilbert-space inner product, this bound on its largest eigenvalue implies a bound on its operator norm with respect to either inner product.
\end{proof}

\subsection{Proof of Lemma \ref{lemma:coset_decomposition} (Coset Decomposition)} \label{app:block_triangular_proof}
\begin{lemma}
    \label{lemma:coset_decomposition}
    For $l \leq t$, each right coset of $S_l$ in $S_t$ has at least one permutation $\rho$ such that, for every permutation $\pi = \lambda \rho$ in the coset $S_l \rho$,
    \begin{gather}
        |\pi| = |\rho| + |\lambda|
    \end{gather}
\end{lemma}
We will call $\rho$ the \textit{minimal representative} of the coset $S_l \rho$, as no other element in $S_l \rho$ can have a smaller size than $\rho$.

\begin{proof}
We invoke the unique decomposition of each permutation $\pi$ into its Jucys-Murphy elements~\cite{Zinn-Justin2010}:
\begin{gather}
    \pi = s_{a_1 b_1} s_{a_2 b_2} ... s_{a_k b_k}
\end{gather}
where $s_{ab}$ is the transposition of elements $a$ and $b$, $a_i < b_i$ for all $i$, and $b_i < b_j$ for all $i < j$. These conditions ensure that none of the transpositions cancel with each other, so $k = |\pi|$. 

Now we take the Jucys-Murphy decomposition of $\pi$, and separate the transpositions $s_{a_i b_i}$ with $b_i \leq l$ from the ones with $b_i > l$. Then, we see that the former product is the unique Jucys-Murphy decomposition of a permutation in $S_l$. We call that permutation $\lambda$. The latter product is the unique Jucys-Murphy decomposition of another permutation, which we identify as $\rho$. 

Now, for any other $\lambda' \in S_l$, we can obtain $\pi' = \lambda' \rho$ by replacing the terms corresponding to $\lambda$ with the decomposition of $\lambda'$. This new product of transpositions does not violate the conditions of the unique Jucys-Murphy decomposition, so it must in fact be the unique Jucys-Murphy decomposition of $\pi'$. Hence the decomposition of $\pi'$ is just the decompositions of $\lambda'$ and $\rho$ put together, much like the original permutation $\pi$, so $|\pi'| = |\lambda'| + |\rho|$, and this holds for all $\lambda' \in S_l, \pi' = \lambda'\rho$. Therefore, $\rho$ is a minimal representative of $\pi$'s coset, and the Coset Decomposition Lemma follows.
\end{proof}

\subsection{Proof of Theorem \ref{thm:eigenstate_independence}}
\label{app:eigenstate_independence}
\begin{theorem*} 
(Restatment of Theorem \ref{thm:eigenstate_independence})
Let $P$ be any projector (not necessarily orthogonal) onto \(\vecspan (S_t^2 \circ S_k^{\times N})\). Then $T^{(t)}$ has a block-triangular structure in $P$ and $(I-P)$ and the former block contains the same eigenvalues as $T^{(k)}$. More precisely, let \(\text{eig}^*(M)\) denote the set of nonzero eigenvalues of $M$. Then
\begin{gather}
    \label{eq:eigenvalues_are_subset2}
    \text{eig}^*\left(T^{(k)}\right) = \text{eig}^* \left(P T^{(t)} P\right)
\end{gather}
and
\begin{gather}
    \label{eq:block_triangular_eigenvalues2}
    \text{eig}^*\left(T^{(t)}\right) = 
    \text{eig}^* \left((I - P) T^{(t)}(I - P)\right) \cup \text{eig}^*\left(T^{(k)}\right)
\end{gather}

\end{theorem*}

\begin{proof}
We will first show that \(P T^{(t)} P\) is essentially several copies of $T^{(k)}$. Afterwards, we will prove that $T^{(t)}$ is block-triangular. 

By Corollary \ref{corollary:gate_t_blindness}, \(\vecspan  S_{k}^{\times N}\) is an invariant subspace of $T^{(t)}$, so $T^{(t)}$ may be diagonalized independently within it. By Lemma \ref{lemma:gate_t_blindness}, the restriction of $T^{(t)}$ to this subspace is isomorphic to $T^{(k)}$, so its eigenvalues are exactly those of $T^{(k)}$. Furthermore, by symmetry the same applies to the restriction of $T^{(t)}$ to \(\sigma_L \tau_R \vecspan  S_{k}^{\times N}\) for any \(\sigma, \tau \in S_t\). Each of these subspaces thus contains a complete copy of the eigenvalues of $T^{(k)}$, together with a complete set of corresponding eigenvectors.

Now consider the combined span of all these subspaces. Since each subspace contains a complete set of eigenvectors, the combined span of the subspaces must be spanned by the union of these eigenvectors. In other words, \(\vecspan  S_t^2 \circ S_{k}^{\times N}\) is spanned by some set of eigenvectors whose eigenvalues are copies of the eigenvalues of $T^{(k)}$. This proves Equation \ref{eq:eigenvalues_are_subset2}.

We now show that $T$ has a block-triangular structure. Since \(\vecspan \left(S_t^2 \circ S_k^{\times N}\right)\) is an invariant subspace of each gate $G_i$, we see $G_iP = PG_iP$. It follows that $(I-P)G_iP = 0$, i.e. each gate is block-triangular. Since each gate is block-triangular, a product of the gates is also block-triangular. The eigenvalues of a block-triangular matrix are exactly those of the diagonal blocks. In other words, for any projector $P$ and matrix $M$ such that $(I-P)MP = 0$, we have
\begin{gather}
    \text{eig}^*\left(M\right) = 
    \text{eig}^* \left((I - P) M(I - P)\right) \cup \text{eig}^*\left(P M P\right)
\end{gather}
Combining this with Equation \ref{eq:eigenvalues_are_subset2} proves Equation \ref{eq:block_triangular_eigenvalues2}. 
\end{proof}

\subsection{Proof of Lemma \ref{lemma:elementwise_hbound}}
\label{app:elementwise_hbound}
\begin{lemma*}
(Restatement of Lemma \ref{lemma:elementwise_hbound}) 
Let $Q_1, Q_2 \geq q \geq t$. Let $f_t, g_{\sigma^{-1}\tau}$ be polynomials such that
\begin{gather}
    \Wg\left(\sigma^{-1}\tau, d\right) = \frac{g_{\sigma^{-1} \tau}(d^{-1})}{f_t(d^{-1})}
\end{gather}
Define the polynomial 
\begin{gather}
    h_{\sigma \tau}(x,y) = \sum_{\rho \in S_t} x^{|\rho \sigma^{-1}|} y^{|\rho|} g_{\rho^{-1} \tau}(xy)
\end{gather}
and let $h_{\sigma \tau}^{(i,j)}$ be the corresponding coefficients. Define the matrix
\begin{gather}
    \overline{H}(t)_{\sigma \tau} = \frac{1}{f_t(t^{-2})}\sum_{i,j} \left|h_{\sigma \tau}^{(ij)}\right| t^{-(i+j)}
\end{gather}
Then we may bound the matrix elements of $H$ by
\begin{gather}
        |H_{\sigma \tau}(Q_1, Q_2)| \leq \left(\frac{t}{q}\right)^{\left\lceil \frac{t}{2}\right \rceil} 
        \overline{H}(t)_{\sigma \tau} 
\end{gather}

\end{lemma*}
\begin{proof}
We have
\begin{align}
    H_{\sigma \tau}(Q_1, Q_2) &= \left[W(Q_1 Q_2) D(Q_2) C(Q_1)\right]_{\sigma \tau} \\
    &= \sum_{\rho \in S_t} Q_1^{-|\rho \sigma^{-1}|} Q_2^{-|\rho|} \Wg(\rho^{-1} \tau, Q_1 Q_2) \\
    &= \frac{1}{f_t((Q_1 Q_2)^{-1})} \sum_{\rho \in S_t} Q_1^{-|\rho \sigma^{-1}|} Q_2^{-|\rho|} g_{\rho^{-1} \tau}((Q_1 Q_2)^{-1}) \\
    &= \frac{1}{f_t\big((Q_1 Q_2)^{-1}\big)} \sum_{\rho \in S_t} x^{|\rho \sigma^{-1}|} y^{|\rho|} g_{\rho^{-1} \tau}(xy)\big|_{x=Q_1^{-1}, y=Q_2^{-1}}
    \label{eq:Rsum}
\end{align}
The whole sum in the last line is a single polynomial over $x$ and $y$, with $t^{-1} \geq x, y > 0$. We may thus define coefficients $h_{\sigma \tau}^{(ij)}$ such that 
\begin{gather}
    \sum_{\rho \in S_t} x^{|\rho \sigma^{-1}|} y^{|\rho|} g_{\rho^{-1} \tau}(xy) = \sum_{ij} h_{\sigma \tau}^{(ij)} x^{i} y^{j} \label{eq:h_polynomial}
\end{gather}
This polynomial is upper-bounded by the polynomial obtained by taking absolute values of its coefficients, so we see
\begin{align}
    H_{\sigma \tau}(Q_1, Q_2) &= \frac{1}{f_t\big((Q_1 Q_2)^{-1}\big)} \sum_{ij} h_{\sigma \tau}^{(ij)} x^{i} y^{j}\big|_{x=Q_1^{-1}, y=Q_2^{-1}}
    \\ &\leq \frac{1}{f_t\big((Q_1 Q_2)^{-1}\big)}\sum_{ij} \left|h_{\sigma \tau}^{(ij)}\right| x^{i} y^{j}\big|_{x=Q_1^{-1}, y=Q_2^{-1}}
\end{align}
Furthermore, $x,y \leq q^{-1}$ and each term of the final polynomial is an increasing function of $x$ and $y$, so we have
\begin{align}
    H_{\sigma \tau}(Q_1, Q_2) &\leq \frac{1}{f_t\big((Q_1 Q_2)^{-1}\big)}\sum_{ij} \left|h_{\sigma \tau}^{(ij)}\right| q^{-(i+j)}
\end{align}
We will now show $h_{ij} = 0$ unless $i + j \geq \lceil \frac{t}{2}\rceil$. Notice that in Equation \ref{eq:Rsum} we are interested in the case in which $\sigma$ is a derangement, so $|\sigma| \geq \lceil \frac{t}{2}\rceil$. But $|\rho \sigma^{-1}| + |\rho| \geq |\sigma|$, so every term in the polynomial of Equation \ref{eq:h_polynomial} has total degree at least $\lceil \frac{t}{2}\rceil$. 

From this it follows that 
$\left|h_{\sigma \tau}^{(ij)}\right| q^{-(i+j) -\lceil \frac{t}{2}\rceil }$ is always a decreasing function of $q$. Using the fact that $q \leq t$, we find
\begin{align}
    \left|h_{\sigma \tau}^{(ij)}\right| q^{-(i+j) } &= \left(\frac{t}{q}\right)^{\lceil \frac{t}{2}\rceil}\left|h_{\sigma \tau}^{(ij)}\right| q^{-(i+j) + \lceil\frac{t}{2}\rceil}\, t^{-\lceil \frac{t}{2}\rceil}\\
    &\leq \left(\frac{t}{q}\right)^{\lceil \frac{t}{2}\rceil}\left|h_{\sigma \tau}^{(ij)}\right| t^{-(i+j) + \lceil\frac{t}{2}\rceil -\lceil \frac{t}{2}\rceil}
    \\&\leq \left(\frac{t}{q}\right)^{\lceil \frac{t}{2}\rceil}\left|h_{\sigma \tau}^{(ij)}\right| t^{-(i+j)}
\end{align}
Substituting, we see
\begin{gather}
    H_{\sigma \tau}(Q_1, Q_2) \leq \left(\frac{t}{q}\right)^{\lceil \frac{t}{2}\rceil}\overline{H}(t)_{\sigma \tau}
\end{gather}
where
\begin{gather}
    \overline{H}(t)_{\sigma \tau} = \frac{1}{f_t((Q_1 Q_2)^{-1})}\sum_{i,j} \left|h_{\sigma \tau}^{(ij)}\right| t^{-(i+j)}
\end{gather}
Finally, $f(z)$ is monotonically decreasing in $z$ for all positive $z$. Therefore, for all $Q_1, Q_2 \geq t$, $f_t((Q_1 Q_2)^{-1}) \geq f_t(t^{-2})$.
\end{proof}

\subsection{Proof of Theorem \ref{thm:coderanged_characterization}}
\label{app:coderanged_characterization}
We are interested in understanding the orthogonal complement of non-deranged subspace in the case $q < t$. 
We will work in terms of the embedding $E: \mathbb{C}^{t!^N} \rightarrow \mathcal{H}_q^{2t}$ given by 
\begin{gather}
    E(\mathbf{v}) = \sum_{\vec{\sigma} \in S_t^N} v_{\vec{\sigma}} \ket{\widetilde{\sigma_1,  \dots \sigma_N}} \equiv \ket{\mathbf{v}}
\end{gather}
We will begin with some helpful properties of the operator
\begin{gather}
    X_{\sigma \tau}(q,t) = \braket{\sigma|\widetilde{\tau}}
\end{gather}
\begin{lemma}
\label{lemma:X_properties}
    The matrix $X(q,t)$ has the following properties:
    \begin{enumerate}
        \item $X$ is a projection. 
        \item $\ket{X\mathbf{v}} = \ket{\mathbf{v}}$ 
        \item $\ket{\mathbf{v}} = \ket{0}$ if and only if $X\mathbf{v} = \mathbf{0}$
        \item $$X = \sum_{\nu \vdash t: |\nu| \leq q} P_\nu$$
        where $P_\nu$ is the canonical idempotent with elements $$P_{\nu,\sigma \tau} = \frac{\chi_\nu(1)}{t!} \sum_\rho \chi_\nu(\rho^{-1}) \delta_{\sigma, \rho \circ \tau}$$
    \end{enumerate}
\end{lemma}
\begin{proof}
    We first prove claim (4). 
    \begin{align}
        X_{\sigma \tau} &= \braket{\sigma|\widetilde{\tau}} = \sum_\rho q^{-|\sigma \rho^{-1}|}\Wg(\tau \rho^{-1},q)
        \\ &= \sum_{\rho \in S_t,\lambda \vdash t : |\lambda| \leq q,\nu \vdash t: |\nu| \leq q} c_\lambda(q) P_{\lambda, \sigma \rho} c^{-1}_\nu(q) P_{\nu,\rho \tau}
        \\ &= \sum_{\lambda \vdash t : |\lambda| \leq q,\nu \vdash t: |\nu| \leq q} c_\lambda(q)c^{-1}_\nu(q) \delta_{\lambda \nu}P_{\nu, \sigma \tau}
        \\ &= \sum_{\nu \vdash t: |\nu| \leq q}P_{\nu, \sigma \tau}
    \end{align}
    Claim (1) follows immediately from this and the fact that the $P_\nu$ give an orthogonal decomposition of the space into irreducible representations. Furthermore we see that $X$ is orthogonal with respect to the basis metric $\delta_{\sigma \tau}$.

    For claim (2), we compute
    \begin{align}
        \braket{\sigma|\mathbf{v}} &= 
        \sum_{\tau} \braket{\sigma|\widetilde{\tau}} v_\tau 
        = \sum_\tau X_{\sigma \tau} v_\tau 
        = (X\mathbf{v})_\sigma 
        \\&= (X^2\mathbf{v})_\sigma
        \\&= \braket{\sigma|X\mathbf{v}}
    \end{align}
    where in the second line we have used claim (1) to see that $X = X^2$. Since the states $\ket{\sigma}$ span the space, this implies claim (2). 
    
    For claim (3), one direction follows from claim (2). To show the other, we will prove that $X\mathbf{v} \ne \mathbf{0} \implies \ket{\mathbf{v}} \ne \ket{0}$. If $X\mathbf{v} \ne 0$ then there exists at least one $\sigma$ such that $\left(X \mathbf{v}\right)_\sigma \ne 0$. Then $\braket{\sigma|\mathbf{v}} = \left(X \mathbf{v}\right)_\sigma \ne 0$, which implies that $\ket{\mathbf{v}}$ must be nonzero.
\end{proof}
\begin{theorem*}
    (Restatement of Theorem \ref{thm:coderanged_characterization})
    Define the intersection space to be the intersection of $\image(X(\vec{q},t))$ with the preimage of the coderanged subspace, 
    \begin{gather}
        \mathcal{I} = \left\{\mathbf{v} \in  \mathbb{C}^{t!^N}: \left(v_{\vec{\sigma}} = 0 \forall \vec{\sigma} \in S_t^2 \circ S_{t-1}^{\times N}\right) \wedge \left(\mathbf{v} \in \image(X(\vec{q},t))\right) \right\}
    \end{gather}
    The restriction of the embedding $E$ to $\mathcal{I}$
    is a vector-space isomorphism between $\mathcal{I}$ and the orthogonal complement of the non-deranged subspace.
\end{theorem*}
\begin{proof}
Let $E\left(\mathcal{I}\right)$ be the image of $\mathcal{I}$ under $E$. 
We first show that every element of $E\left(\mathcal{I}\right)$ is orthogonal to the non-deranged subspace. From Lemma \ref{lemma:X_properties}.1 we see $\mathbf{v} \in \image(X)$ implies $X\mathbf{v} = \mathbf{v}$, so we can compute
\begin{gather}
    \braket{\vec{\sigma}|\mathbf{v}} = \sum_{\vec{\tau}} v_{\vec{\tau}}\braket{\vec{\sigma}|\widetilde{\vec{\tau}}}
    = 
    \sum_{\vec{\tau}} X_{\vec{\sigma}\vec{\tau}} v_{\vec{\tau}} 
    = v_{\vec{\sigma}} 
    = 0
\end{gather}
This establishes that $E\left(\mathcal{I}\right)$ is contained in the orthogonal complement of the non-deranged subspace. 

We now show that any vector orthogonal to the non-deranged subspace is represented by an element of $\mathcal{I}$. Consider $\mathbf{v} \in \mathbb{C}^{t!^N}$ such that $\braket{\vec{\sigma}|\mathbf{v}} = 0$ for every $\vec{\sigma} \in S_t^2 \circ S_{t-1}^{\times N}$. Define $\mathbf{u} = X\mathbf{v}$. From Lemma \ref{lemma:X_properties}.2 we have $\ket{\mathbf{u}} = \ket{\mathbf{v}}$, so we see $\braket{\vec{\sigma}|\mathbf{u}} = 0$ for every $\vec{\sigma} \in S_t^2 \circ S_{t-1}^{\times N}$. But from Lemma \ref{lemma:X_properties}.1 we have $X\mathbf{u} = \mathbf{u}$, so $\braket{\vec{\sigma}|\mathbf{u}} = u_{\vec{\sigma}}$. We thus see that $u_{\vec{\sigma}} = 0$ for all $\vec{\sigma} \in S_t^2 \circ S_{t-1}^{\times N}$, and so $\mathbf{u} \in \mathcal{I}$. This implies that the orthogonal complement of the non-deranged subspace is contained in $E\left(\mathcal{I}\right)$. 

We have shown that $E\left(\mathcal{I}\right)$ is exactly the orthogonal complement of the non-deranged subspace. We now show that the embedding is an isomorphism. It suffices to show that its kernel is trivial. From Lemma \ref{lemma:X_properties}.3 we see that $\ket{\mathbf{v}} = \ket{0} \implies X\mathbf{v} = \mathbf{0}$. It follows that if $\ket{\mathbf{v}} = \ket{0}$ and $\mathbf{v} \ne \mathbf{0}$, then $\mathbf{v} \notin \mathcal{I}$. The kernel of the embedding is thus trivial when we restrict to $\mathcal{I}$, so it is an isomorphism.
\end{proof}

\section{Numerical results}
\subsection{Values of $||K_m||_{\mathfrak{D}}$ for $t \leq 6, t \leq q$}
Let's return to the half product in the complete derangement block:
\begin{gather}
    H_{\sigma \tau} = \left[W(Q_1 Q_2) D(Q_2) C(Q_1)\right]_{\sigma \tau} = \sum_{\rho \in S_t} Q_1^{-|\rho \sigma^{-1}|} Q_2^{-|\rho|} \Wg(\rho^{-1} \tau, Q_1 Q_2) 
\end{gather}
From Lemma~\ref{lemma:elementwise_hbound}, there is an elementwise bound of
\begin{gather}
    H_{\sigma \tau} \leq \frac{1}{f_t(t^{-2})}\left(\frac{t}{q}\right)^{\lceil \frac{t}{2}\rceil}\sum_{i,j} \left|h_{\sigma \tau}^{(ij)}\right| t^{-(i+j)}
\end{gather}
Where $h_{\sigma \tau}$ was a polynomial dependent on the Weingarten numerator $g_{\sigma^{-1} \tau}(z)$, while $f_t(z)$ was the Weingarten denominator. We can choose a specific set of numerators $g_{\sigma^{-1} \tau}(z)$ and denominators $f_t(z)$, such that\footnote{For general $t$, we can define $f_t(z)$ as the inverse of the least common multiple of $z^t c_\lambda(z^{-1})$ for all partitions $\lambda$ of $t$, where $c_\lambda(q)$ is the eigenvalue of the cycle matrix $C(q)$ corresponding to partition $\lambda$.}
\begin{gather}
    f_t(z) = \prod_{i=1}^{t-1} \left(1-i^2 z^2\right)\begin{cases}  1 & t \leq 5\\ (1-z^2) & t = 6\end{cases}.
\end{gather}
Under this choice of $f_t(z)$, the Weingarten numerators $g_c(z)$ (where $c$ is the conjugacy class of $\sigma^{-1}\tau$, as that is the only thing that $g$ depends on) is given by Table~\ref{tab:gc_z} for each choice of conjugacy class $c$, while the values of $f_t(t^{-2})$ are in Table~\ref{tab:ft_t2}.

\begin{table}
    \begin{tabular}{|c||p{2.5cm}|p{2.2cm}|p{2.2cm}|p{2.2cm}|p{2.2cm}|p{2.2cm}|p{2.2cm}|}
        \hline
        $c,t=3$ & 111 & 21 & 3 & & & &\\
        \hline
        $g_c(z)$ & $1-2z^2$ & $-z$ & $2z^2$ & & & &\\ 
        \hline \hline
        $c,t=4$ & 1111 & 211 & 31 & 22 & 4 & & \\
        \hline
        $g_c(z)$ & $1-8z^2+6z^4$ & $-z+4z^3$ & $2z^2-3z^4$ & $z^2+6z^4$ & $-5z^3$ & & \\ 
        \hline \hline
        $c,t=5$ & 11111 & 2111 & 311 & 221 & 41 & 32 & 5\\
        \hline
        $g_c(z)$ & $1 - 20z^2 + 78z^4$ & $-z + 14z^3 - 42z^5$ & $2z^2-18z^4$ & $z^2-2z^4$ & $-5z^3 + 24z^5$ & $-2z^3 - 12z^5$ & $14z^4$\\ 
        \hline \hline
        $c,t=6$ & 111111 & 21111 & 3111 & 2211 & 411 & 321 & 51\\
        \hline
        $g_c(z)$ & $1-41z^2+458z^4-1258z^6+240z^8$ & $-z+33z^3-254z^5+342z^7$ & $2z^2-51z^4+229z^6-60z^8$ & $z^2-19z^4+58z^6-160z^8$ & $-5z^3 + 93z^5 - 208z^7$ & $-2z^3+5z^5+117z^7$ & $14z^4 -154 z^6 + 140 z^8$\\ \hline
         & 222 & 42 & 33 & 6 & & &\\
        \hline
        $g_c(z)$ & $-z^3-z^5+358z^7$ & $5z^4 + 75z^6+40z^8$ & $4z^4+116z^6-360z^8$ & $42z^5$ &  & & \\ \hline
    \end{tabular}
    \caption{Values of the Weingarten numerator $g_c(z)$ for $t \leq 6$.}\label{tab:gc_z}
\end{table}

\begin{table}
    \begin{tabular}{|c||c||c|c|}
        \hline
        t & $f_t(t^{-2})$ & $||H(t)||_{\mathfrak{D}} \leq$ & $||K_m||_{\mathfrak{D}} \leq$\\
        \hline 
        2 & 0.9375 & $\frac{q}{q^2+1}$  & $\frac{1}{q^2+1}$\\
        \hline
        3 & 0.9389 & $0.1845 \left(\frac{t}{q}\right)^2$ & $\frac{0.307}{q^2}\left(\frac{t}{q}\right)^2$\\
        \hline
        4 & 0.9460 & $0.1018 \left(\frac{t}{q}\right)^2$ & $\frac{0.167}{q^2}\left(\frac{t}{q}\right)^2$\\
        \hline
        5 & 0.9527 & $0.01818 \left(\frac{t}{q}\right)^3$ & $\frac{8.27\times 10^{-3}}{q^2}\left(\frac{t}{q}\right)^4$\\
        \hline
        6 & 0.9581 & $8.297 \times 10^{-3} \left(\frac{t}{q}\right)^3$ & $\frac{2.48 \times 10^{-3}}{q^2}\left(\frac{t}{q}\right)^4 $\\
        \hline
    \end{tabular}
    \caption{Values of the Weingarten denominator $f_t(t^{-2})$ for $t \leq 6$, as well as bounds on the singular values of $H$ and $K_m$.}\label{tab:ft_t2}
\end{table}

Using these reductions, we obtain the bounds on the singular values of $H$ and $K_m$ as shown in Table~\ref{tab:ft_t2}.
All of these are less than $\frac{1}{q^2+1}$, which is lower bounded by $\frac{0.8}{q^2}$ for all $2 \leq t \leq 6, t \leq q$. Therefore, the subleading singular value bound of $K_m$ is bounded by $\frac{1}{q^2+1}$ for all  $2 \leq t \leq 6$, $t\leq q$, proving Lemma~\ref{lemma:Kbound_t34}.

\subsection{Values of $||K_m||_{\mathfrak{D}}$ for $q < t$ (Proof of Theorem \ref{thm:gap_finite_N})}
\label{app:gap_finite_N}
In order to prove Theorem \ref{thm:gap_finite_N}, we must calculate $||K_m||_{\mathfrak{D}}$ for every $q < t \leq 6$ and $m \leq 1000$. Our numerics use a variant of the Lanczos algorithm. Results are presented in Figure \ref{fig:larget_numerics}.

\begin{figure}
    \centering
    \begin{tikzpicture}
        \begin{scope}
            \node[anchor=north west,inner sep=0] (image_a) at (0,0)
            {\includegraphics[width=0.8\columnwidth]{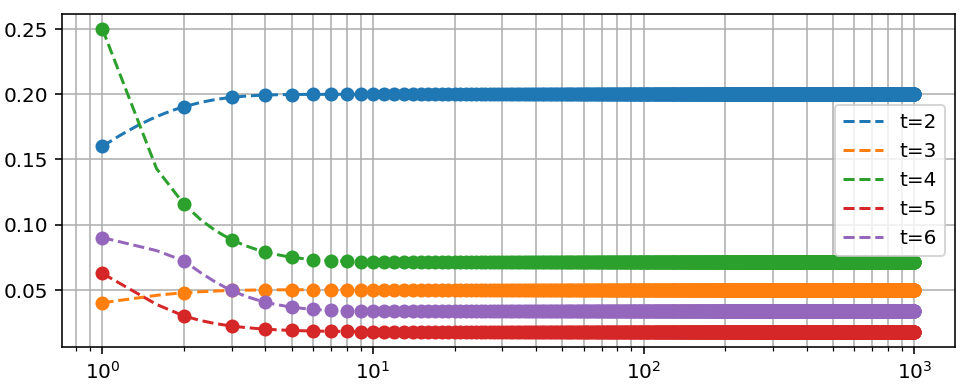}};
            \node [anchor=north west] (note) at (7,-5.3) {\large{$m$}};
            \node [anchor=north west] (note) at (-0.5,-2.4) {\large{$\lambda$}};
        \end{scope}
    \end{tikzpicture}
    \vspace*{-0.4cm}
    \caption{Largest eigenvalues in the deranged subspace for $q = 2$ and various values of $t$. Each curve remains essentially flat out to at least $m = 1000$ (not shown).}
    \label{fig:larget_numerics}
\end{figure}

Our computational basis states $\mathbf{e}_\sigma$ are related to the physical Hilbert space by the embedding $E$ described in the previous section. 
This embedding is an isometry under the metric 
\begin{gather}
    \langle \mathbf{e}_{\sigma}, \mathbf{e}_{\tau} \rangle = \Wg(\tau \sigma^{-1}, d)
\end{gather}
Since the metric is degenerate, $E^\dagger$ is not uniquely defined. However, if we define
\begin{multline}
    X_m = 
    \sum_{\sigma_i, \tau_i, \alpha, \beta} \Bigg[\delta_{\sigma_1 \sigma_2} \Wg(\sigma_1 \alpha^{-1}, q^{m+1}) q^{-m|\alpha^{-1} \tau_1|} q^{-|\alpha^{-1} \sigma_3|}
    \\ \times
    \Wg(\sigma_3 \beta^{-1}, q^2) q^{-|\beta^{-1} \tau_2|}q^{-|\beta^{-1} \tau_3|} 
    \mathbf{e}_{\tau_1 \tau_2 \tau_3}\mathbf{e}_{\sigma_1 \sigma_2 \sigma_3}^T\Bigg]
\end{multline}
then we have
\begin{gather}
    \Pi_{m+1} G_m = E X_m E^\dagger
\end{gather}
for any choice of adjoint $E^\dagger$.

We may cyclically permute the leading $\delta_{\sigma_1 \sigma_2}$ operation to obtain a new operator with the same spectrum
\begin{gather}
    O_m = \sum_{\sigma, \tau, \pi, \rho} \left(\sum_{\alpha, \beta} \Wg(\sigma \alpha^{-1}, q^{m+1}) q^{-m|\alpha \pi|}q^{-|\alpha \tau|} \Wg(\tau \beta^{-1}, q^2) q^{-|\beta \pi|} q^{-|\beta \rho|}\right) \mathbf{e}_{\pi \rho} \mathbf{e}_{\sigma \tau}^T
\end{gather}
This operator is now Hermitian with respect to the inherited metric. 

To find new eigenvalues, we wish to restrict our search to the intersection space $\mathcal{I}$ described in Theorem \ref{thm:coderanged_characterization}. This subspace is invariant under $O_m$, so in principle it suffices to begin our Lanczos iterations with an initial vector in $\mathcal{I}$. However, numerical instability makes the invariance only approximate, so at each iteration of the algorithm we must ensure that our vectors are constrained to $\mathcal{I}$.

We first find a starting point. Beginning with a random (normalized) vector, we use the implicitly restarted Lanczos method to find an eigenvector of $X|_{\vecspan\mathfrak{D}_2}$ with the largest eigenvalue. The corresponding eigenvalue will be $1$ so long as $\mathcal{I}$ is nontrivial. We then run the usual Lanczos algorithm to find large eigenvalues of $O_m$. At each step of the iteration we ensure that the vector remains (nearly) invariant under $X$ by repeated applications of $X$. The size of the subspace is increased iteratively until the change in the largest eigenvalue found is sufficiently small.

\section{Properties of the Derangement Polynomial}
\subsection{Proof of Lemma \ref{lemma:dc_bound}}\label{app:proof_of_dc_bound}
\begin{lemma*} (Restatement of Lemma \ref{lemma:dc_bound})
Define the \textit{derangement polynomial}
    \begin{gather}\label{eq:dt_definition_copy2}
        d_t(x) = \sum_{\sigma \in S_t^{(\mathfrak{D})}} \sum_{\tau \in S_t} x^{|\sigma^{-1} \tau| + |\tau|}
    \end{gather}
    Let $Q_1, Q_2 \geq q$. Then
    \begin{gather}
        ||D(Q_2) C(Q_1)||_{\mathfrak{D}, \text{basis}}  \leq \sqrt{d_t(q^{-2})}
    \end{gather}
\end{lemma*}

\begin{proof}
    We can compute 
    \begin{gather}
        (D(Q_2)C(Q_1))_{\sigma, \tau} = \sum_{\rho \in S_t}\delta_{\tau \rho} Q_2^{-|\tau|}Q_1^{-|\sigma^{-1} \rho|}  = Q_2^{-|\tau|}Q_1^{-|\sigma^{-1} \tau|}
    \end{gather}
    The spectral norm of a matrix is upper-bounded by its Frobenius norm, so
    \begin{align}
        ||D(Q_2)C(Q_1)||_2^2 
        &\leq \sum_{\sigma \in D_t}\sum_{\tau \in S_t} \left(Q_2^{-|\tau|}Q_1^{-|\sigma^{-1} \tau|} \right)^2
        \\
        &\leq \sum_{\sigma \in D_t}\sum_{\tau \in S_t} Q^{-2(|\tau|+|\sigma^{-1} \tau|)} 
    \end{align}
\end{proof}

\subsection{Product and Matrix Forms}
\label{app:form_of_derangement_polynomial}
\begin{lemma}\label{lemma:derangement_polynomial_productform}
    The derangement polynomial can be expressed as the following sum:
    \begin{gather}
        d_t(x) = (-1)^t\sum_{k=0}^t (-1)^k \binom{t}{k} \prod_{\ell = k+1}^t \big(1+(\ell - 1)x^2\big) \prod_{\ell = 1}^{k} \big(1+(\ell - 1)x\big)^2 
    \end{gather}
\end{lemma}
\begin{proof}
We start with the original expression of the derangement polynomial:
\begin{gather}
    d_t(x) = \sum_{\sigma \in S_t^{(\mathfrak{D})}} \sum_{\tau \in S_t} x^{|\sigma^{-1} \tau| + |\tau|}
\end{gather}
We consider the more general sum
\begin{gather}
    c_t(x) \equiv \sum_{\tau \in S_t} x^{|\tau|}\\
    c_t(x)^2 = \sum_{\sigma \in S_t} \sum_{\tau \in S_t} x^{|\sigma^{-1} \tau| + |\tau|}
\end{gather}
Since the term in the sum is invariant upon simultaneous conjugation of $\sigma$ and $\tau$, $\sum_{\tau \in S_t} x^{|\sigma^{-1} \tau| + |\tau|}$ only depends on the conjugacy class of $\sigma$. Therefore, we can separate the sum over derangements $D_k$:
\begin{align*}
    c_t(x)^2 &= \sum_{k=0}^t \binom{t}{k} \sum_{\sigma \in D_k} \sum_{\tau \in S_t} x^{|\sigma^{-1} \tau| + |\tau|}
\end{align*}
For $\sigma \in D_k$, $\tau \in S_t$, we can use Lemma \ref{lemma:coset_decomposition} to split $\tau = \rho \pi$, $\rho \in S_k$ with $\pi$ being a minimal representative, and get
\begin{align}
    c_t(x)^2 &= \sum_{k=0}^t \binom{t}{k} \sum_{\pi \in S_t : S_k} x^{2|\pi|} \sum_{\sigma \in D_k} \sum_{\tau \in S_k} x^{|\sigma^{-1} \rho| + |\rho|}\n
    &= \sum_{k=0}^t \binom{t}{k} d_k(x) \sum_{\pi \in S_t : S_k} x^{2|\pi|} 
\end{align}
Moreover, for any specific $k$,
\begin{gather}
    c_t(x^2) = \sum_{\tau \in S_t} x^{2|\tau|} = \sum_{\pi \in S_t : S_k} x^{2|\pi|} \sum_{\rho \in S_k} x^{2|\rho|}\\
    \sum_{\pi \in S_t:S_k} x^{2|\pi|} = \frac{c_t(x^2)}{c_k(x^2)}
\end{gather}
So
\begin{gather}
    c_t(x)^2 = \sum_{k=0}^t \binom{t}{k} \frac{c_t(x^2)}{c_k(x^2)} d_k(x)
\end{gather}
A system of linear equations from $d_k(x)$ to $c_k(x)^2$ which we can invert as
\begin{gather}
    d_t(x) = \sum_{k=0}^t (-1)^{t-k} \binom{t}{k} \frac{c_t(x^2)}{c_k(x^2)} c_k(x)^2
\end{gather}
From the Jucys-Murphy decomposition of permutations in $S_k$,
\begin{gather}
    c_k(x) = \sum_{\sigma \in S_k} x^{|\sigma|} = \prod_{\ell = 1}^t \big(1+(\ell - 1)x\big)
\end{gather}
So we can also write
\begin{gather}
    d_t(x) = (-1)^t\sum_{k=0}^t (-1)^k \binom{t}{k} \prod_{\ell = k+1}^t \big(1+(\ell - 1)x^2\big) \prod_{\ell = 1}^{k} \big(1+(\ell - 1)x\big)^2 \label{eq:derangement_polynomial}
\end{gather}
\end{proof}

\subsection{Numerical Bound on the Derangement Polynomial (Proof of Lemma~\ref{lemma:dt_bound_numerical})} \label{app:dt_numerical_proof}
Using Eq.~\ref{eq:derangement_polynomial}, we can obtain an exact form of $d_t(x)$ for very large $t$ values in a time that is only polynomial in $t$. This in turn lets us bound $d_t(x)$ for these $t$ values, which allows us to prove Lemma~\ref{lemma:dt_bound_numerical}. 

The minimal degree of $d_t(x)$ must be smallest transposition cost of any derangement $\sigma \in S_t^{(\mathfrak{D})}$, i.e. it is $\ceil{\frac{t}{2}}$. Next, take $x_0 \equiv t^{-2} \geq x$. Because $d_t(x)$, in its original form (\ref{eq:dt_definition_copy2}), is a polynomial with positive coefficients, each term is monotonic in $x$. Therefore,
\begin{align}
    d_t(x) &\leq d_t(x_0) \left(\frac{x}{x_0}\right)^{\ceil{\frac{t}{2}}}\\
    &\leq d_t(t^{-2}) \left(\frac{t}{Q}\right)^{2\ceil{\frac{t}{2}}}
\end{align}
\begin{figure}
    \centering
    \begin{tikzpicture}
        \begin{scope}
            \node[anchor=north west,inner sep=0] (image_a) at (0.4,0)
            {\includegraphics[width=0.4\columnwidth]{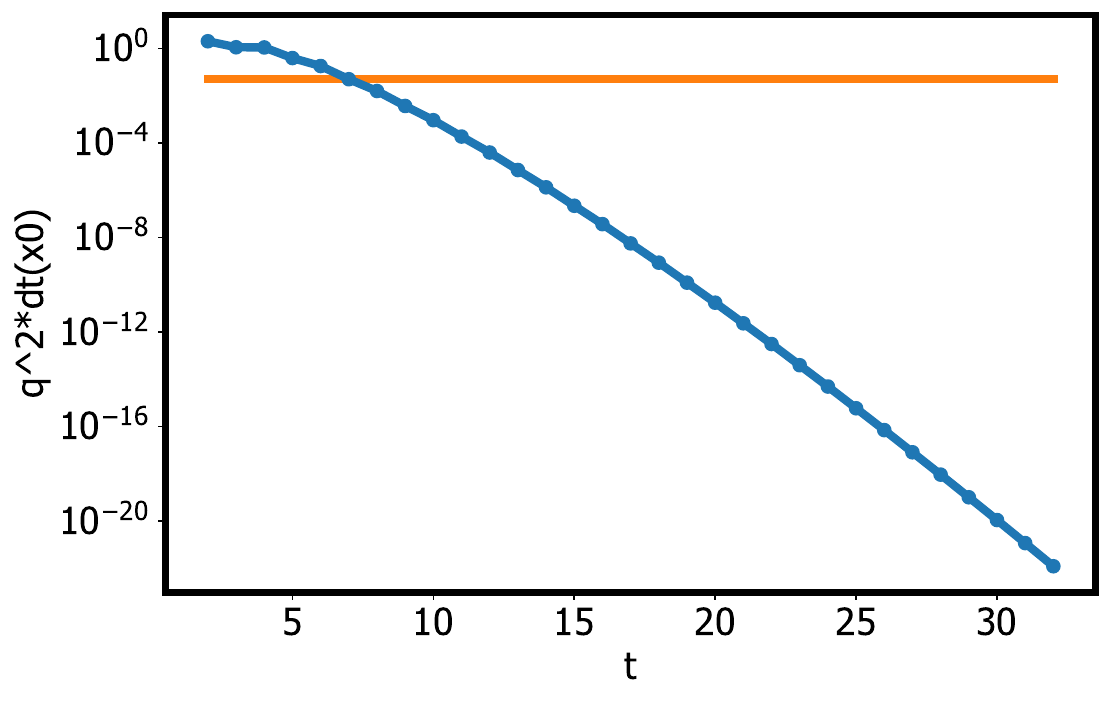}};
        \end{scope}
    \end{tikzpicture}
    \vspace*{-0.4cm}
    \caption{Bounds on the scaled derangement polynomial $q^2 d_t(t^{-2})$ for different values of $t \leq 32$. The $t=7$ value of $0.0496$ is also given by the orange line.}
    \label{fig:derangement_coeffs}
\end{figure}
The values of $d_t(t^{-2})$ for $t \leq 32$ are given in Fig.~\ref{fig:derangement_coeffs}, and monotonically decrease in $t$ for $t \geq 7$. Therefore, we can bound each of their coefficients by the $t=7$ value of $1.013\times 10^{-3}$, giving
\begin{gather}
    d_t(Q^{-2}) \leq 1.013\times 10^{-3} \left(\frac{t}{Q}\right)^{2\ceil{\frac{t}{2}}} \leq 1.013\times 10^{-3} \left(\frac{t}{q}\right)^t\\
    q^2 d_t(Q^{-2}) \leq 0.0496 \left(\frac{t}{q}\right)^{t-2}
\end{gather}
for all $7 \leq t \leq 32$. By Lemmas \ref{lemma:main_product} and \ref{lemma:weingarten_bound}, this therefore implies a bound on $K_m$ of
\begin{gather}
    \lambda = e^2 d_t(Q^{-2}) \leq \frac{0.367}{q^2} \left(\frac{t}{q}\right)^{t-2}
\end{gather}

\subsection{Analytic Bound on the Derangement Polynomial (Proof of Lemma \ref{lemma:dt_bound_analytic})} \label{app:proof_of_dt_bound}
\begin{lemma*}
(Restatement of Lemma \ref{lemma:dt_bound_analytic})
We have
\begin{gather}
    d_t(q^{-2}) \leq \half \frac{e^{\frac{t(t-1)}{2q^4}}}{2}\left(\frac{t}{q}\right)^t \left[\left(\frac{1}{2} + \frac{3}{t}\right)^t + \left(\frac{4}{\sqrt{t}}\right)^t\right]
\end{gather}
Furthermore, when $t \leq q$, this bound can be simplified to
\begin{gather}
    d_t(q^{-2}) \leq \frac{t^2}{2} e^{\frac{1}{2t^2}} \left[\left(\frac{1}{2} + \frac{3}{t}\right)^t + \left(\frac{4}{\sqrt{t}}\right)^t\right] \frac{1}{q^2}
\end{gather}
\end{lemma*}

We will need several intermediate results before we are ready to prove Lemma \ref{lemma:dt_bound_analytic}. 
\begin{lemma}
\label{lemma:derangement_polynomial_matrixform}
\begin{gather}
        d_t(x) = c_t(x^2) \hat{e}_0^T \left(\Delta + TG(x)\right)^t \vec{u}
    \end{gather}
where \(\{\hat{e}_i\}\) is the standard basis for \(\mathbb{R}^{t+1}\) and we have defined
\begin{align}
    \vec{u} &= \sum_{i=1}^{t+1} \hat{e}_i \\
    T &= \sum_{i=1}^t \hat{e}_{i+1}\hat{e}_i^T \\
    \Delta &= T - I \\
    G(x) &= \sum_{j} \left(\frac{\big(1+j x\big)^2}{1+j x^2} - 1\right) \hat{e}_j \hat{e}_j^T \\
    c_t(x) &= \sum_{\tau \in S_t} x^{|\tau|}
\end{align}
\end{lemma}
\begin{proof}
Define $g_\ell(x) = G_{\ell \ell}(x)$. We start with Lemma $\ref{lemma:derangement_polynomial_productform}$, which we express as
\begin{align}
    d_t(x) &= c_t(x^2) (-1)^t \sum_{k=0}^t (-1)^{t-k} \binom{t}{k} \prod_{\ell = 0}^{k-1} (1+g_\ell(x)) \label{eq:dt_sum}
\end{align}
Note that $g_\ell(0) = 0$. Ignoring the prefactor $c_t(x^2)$ for now, this expression can be described as a weighted sum over paths. Each path occurs over $t$ total time steps. At each time step we either move forward a space, acquiring weight $\frac{(1 + \ell x)^2}{1 + \ell x^2} = 1+g_\ell(x)$ if we are originally at position $\ell$, or we do not advance, acquiring a weight of $-1$. Any path that advances $k$ steps in total therefore acquires a weight $(-1)^{t-k} \prod_{\ell = 0}^{k-1} (1+g_\ell(x))$, and there are $\binom{t}{k}$ such paths. At the end of this process, we sum over all weights of all possible paths, regardless of the final position. This weighted path sum can therefore be represented as a matrix element
\begin{gather}
    u^T M^t \hat{e}_0
\end{gather}
where $M$ is the transfer matrix representing all possible actions at a single time step. The decision to advance and acquire a weight $(1+g_\ell(x))$ corresponds to the transfer matrix $T + TG(x)$, while the decision to remain still and acquire a weight $(-1)$ corresponds to $-I$. Therefore $M=T+TG-I$, so
\begin{align*}
    d_t(x) &= c_t(x^2) u^T \big[T + TG(x) - I\big]^t \hat{e}_0\\
    &= c_t(x^2) u^T \big[\Delta + TG\big]^t \hat{e}_0\\
\end{align*}-
\end{proof}

\begin{lemma}
    \label{lemma:derangement_rearrangement}
    Define an operator $\delta$ which maps diagonal matrices to diagonal matrices by
    \begin{gather}
    \delta M = \sum_{i} (M_{i,i} - M_{i-1,i-1})e_i e_i^T
    \end{gather}
    We have the following bound on the derangement polynomial:
    \begin{gather}
    d_t(x) \leq  c_t(x^2)\sum_A \binom{t}{A} A^{t-A} \max_{\lambda} \left| u^T \prod_{i=1}^{A} T^{\lambda_i+1} \delta^{\lambda_i} G(x) \hat{e}_0\right|. \label{eq:dt_inequality}
\end{gather}
\end{lemma}
\begin{proof}
    From Lemma \ref{lemma:derangement_polynomial_matrixform} we have
    \begin{gather}
        d_t(x) = c_t(x^2) u^T \left(\Delta + TG(x)\right)^t \hat{e}_0
    \end{gather}
    The matrix power $\big[\Delta + TG(x)\big]^t$ can be expanded into $2^t$ terms. Our general claim is that each of these terms decays exponentially with $t$. To demonstrate this claim, we start with the following observations:
\begin{itemize}
    \item $\Delta u = 0$
    \item $\Delta$ follows a pseudo product rule: $\Delta AB = [\Delta,A]B + A \Delta B$
    \item $[\Delta, T^k D] = [T-I, T^k D] = T^{k+1} \delta D$, where $D_{ij} = \delta_{ij} d_j$ is any diagonal matrix and $\delta D = \delta_{ij}(d_{j} - d_{j-1})$ is its finite difference (note that this is also a diagonal matrix). 
\end{itemize}
So
\begin{gather}
    \Delta\left(\prod_{i} T^{k_i} D_i\right)  \hat{e}_0 = \sum_{j} \prod_i\left[(1-\delta_{ij})T^{k_i} D_i + \delta_{ij}T^{k_i+1} \delta D_i \right] \hat{e}_0
\end{gather}
for any sequence of integers $k_i$ and diagonal matrices $D_i$. That is, applying $\Delta$ to a product of $\ell$ factors $T^{k_i} D_i$ gives us a sum of $\ell$ products, where each term in the sum is the original product with one of the factors replaced by $T^{k_i+1} \delta D_i$. This lets us impose the inequality
\begin{gather} \label{eq:delta_inequality}
    \big|\hat{e}_m^T\Delta\left(\prod_{i=0}^\ell T^{k_i} D_i\right)  \hat{e}_0\big| \leq \ell \max_j \big|\hat{e}_m^T\prod_{i=0}^\ell\left[(1-\delta_{ij})T^{k_i} D_i + \delta_{ij}T^{k_i+1} \delta D_i \right] \hat{e}_0\big|
\end{gather}
Now we generalize the situation to a sequence of $t$ operators which could either be $\Delta$ or $TG$. Specifically, we have two vectors $a \in \{0...t\}^m, b\in \{0...t\}^m$ for some integer $m$ such that $\sum_i a_i + b_i = t$. Let's also say that $\sum_i a_i \equiv A$. We want to bound the product
\begin{gather}
    \big|u^T \prod_{i=1}^m \Delta^{b_i} (TG)^{a_i} \hat{e}_0\big|
\end{gather}
By imposing inequality (\ref{eq:delta_inequality}) on each $\Delta$ individually, we get
\begin{gather}
    \big|u^T \prod_{i=1}^m \Delta^{b_i} (TG)^{a_i} \hat{e}_0\big| \leq \left(\prod_{k=1}^m \big(\sum_{j=1}^k a_j\big)^{b_i}\right) \max_{\lambda} \big|u^T \prod_{i=1}^{A} T^{\lambda_i+1} \delta^{\lambda_i} G(x) \hat{e}_0\big|
\end{gather}
where $\lambda$ is taken as any vector in $\{0...t-A\}^A$ such that $\sum_i \lambda_i = t-A$. The prefactor $\prod_{k=1}^m \big(\sum_{j=1}^k a_j\big)^{b_i}$ is counting the total number of terms in our sum after passing through each $\Delta$ individually, and for a fixed $A$ can be maximized if all the $\Delta$'s were applied at the end, i.e. $x_1 = A, y_1 = t-A$. This gives
\begin{gather}
    \big|u^T \prod_{i=1}^m \Delta^{b_i} (TG)^{a_i} \hat{e}_0\big| \leq A^{t-A} \max_{\lambda} \big|u^T \prod_{i=1}^{A} T^{\lambda_i+1} \delta^{\lambda_i} G(x) \hat{e}_0\big|
\end{gather}
and therefore
\begin{gather}
    d_t(x) = c_t(x^2) u^T \left[\Delta + TG(x)\right]^t \hat{e}_0 \leq  c_t(x^2)\sum_A \binom{t}{A} A^{t-A} \max_{\lambda} \big|u^T \prod_{i=1}^{A} T^{\lambda_i+1} \delta^{\lambda_i} G(x) \hat{e}_0\big|. 
\end{gather}
\end{proof}

\begin{lemma}
\label{lemma:derangement_Gbound}
\begin{gather}
        \left(\delta^n G(x)\right)_{ij} \leq 3x^\frac{n+1}{2} \delta_{ij}
\end{gather}
where $G$ is as defined in Lemma \ref{lemma:derangement_polynomial_matrixform} and $\delta$ is as defined in Lemma \ref{lemma:derangement_rearrangement}.
\end{lemma}
\begin{proof}
We again use the shorthand $g_k(x) = G_{kk}(x)$. We first write
\begin{gather}
    \big[\delta^{n} G(x)\big]_{kk} = \delta^{n} g_k
(x)
\end{gather}
where for scalar-valued sequences we have defined the usual finite difference operator
\begin{gather}
    \delta f_k(x) = f_{k}(x) - f_{k-1}(x)
\end{gather}
Then
\begin{gather}
    \delta^n g_k(x) = \sum_{i=0}^n \binom{n}{i}(-1)^{i} g_{k-i}(x)
\end{gather}
Let's write 
\[j_x(y) \equiv \frac{(1+y)^2}{1+xy} - 1\]
\begin{gather}
    g_k(x) = \frac{(1+kx)^2}{1+kx^2} - 1 = j_x(kx)
\end{gather}
\begin{gather}
    g_{k-i}(x) = j_x(kx-ix) = \sum_{\ell = 0}^{\infty} j_x^{(\ell)}(kx) \frac{(-ix)^\ell}{\ell!} 
\end{gather}
Therefore,
\begin{align*}
    \delta^n g_k(x) &= \sum_{\ell = 0}^{\infty} j_x^{(\ell)}(kx) \frac{(-x)^\ell}{\ell!} \sum_{i=0}^n \binom{n}{i}(-1)^{i} i^\ell\\
    &= \sum_{\ell = 0}^{\infty} j_x^{(\ell)}(kx) \frac{(-x)^\ell}{\ell!} \frac{d^\ell}{dz^\ell} \sum_{i=0}^n \binom{n}{i}(-1)^{i} e^{iz} \big|_{z=0}\\
    &= (-1)^n \sum_{\ell = 0}^{\infty} j_x^{(\ell)}(kx) \frac{(-x)^\ell}{\ell!} \frac{d^\ell}{dz^\ell} (e^z-1)^n \big|_{z=0}\\
    &= (-1)^n \sum_{\ell = n}^{\infty} j_x^{(\ell)}(kx) \frac{(-x)^\ell}{\ell!} \frac{d^\ell}{dz^\ell} (e^z-1)^n \big|_{z=0}
\end{align*}
In the last step we used the fact that $e^z-1$ has no constant term in its Taylor series, so $\frac{d^\ell}{dz^\ell} (e^z-1)^n \big|_{z=0} = 0$ for all $\ell < n$. Also, $\frac{d^\ell}{dz^\ell} (e^z-1)^n \big|_{z=0} \leq \frac{d^\ell}{dz^\ell} (e^z)^n \big|_{z=0}$ (the latter has more terms in each Taylor series and all coefficients of the Taylor expansion are positive).
Therefore, in terms of magnitudes,
\begin{align*}
    \delta^n g_k(x) &\leq \sum_{\ell = n}^{\infty} |j_x^{(\ell)}(kx)| \frac{(nx)^\ell}{\ell!}
\end{align*}

Next, we bound each $|j_x^{(\ell)}(kx)|$ using the Cauchy integral theorem:
\begin{align*}
    j_x^{(\ell)}(y) &= \frac{\ell!}{2\pi i} \int_\gamma \frac{j_x(z)}{(z-y)^{\ell+1}} dz\\
    &= \frac{\ell!}{2\pi i} \int_\gamma \frac{(1+z)^2}{(1+xz)(z-y)^{\ell+1}} dz
\end{align*}
For $\ell \geq 2$, the contour integral vanishes at infinity. Of course, if we take the contour integral to infinity, we capture the extra pole at $z=-1/x$, so the residues of the two poles must sum up to zero. Therefore,
\begin{align*}
    j_x^{(\ell)}(y) &= -\ell! \, \text{Res} \frac{(1+z)^2}{(1+xz)(z-y)^{\ell+1}} \big|_{z = -1/x}\\
    &= -\frac{\ell!}{x} \frac{(1-1/x)^2}{(-y-1/x)^{\ell+1}}\\
    &= (-1)^{\ell} \ell! x^{\ell-2} \frac{(1-x)^2}{(1+xy)^{\ell+1}}\\
    |j_x^{(\ell)}(y)| &\leq \ell! \, x^{\ell-2}
\end{align*}
For $\ell = 0$, 
\begin{gather}
    j_x^{(0)}(y) = \frac{(1+y)^2}{(1+xy)} - 1 \leq 2y+y^2
\end{gather}
and for $\ell = 1$,
\begin{gather}
    j_x'(y) = \frac{2(1+y)(1+xy) - x(1+y)^2}{(1+xy)^2} \leq 2(1+y)
\end{gather}
These inequalities hold in magnitude as well, provided $x < y < 1$, which is consistent with our $x = Q^{-2} \leq t^{-2}, y = kx \leq tx \leq t^{-1}$ regime. Therefore, we have (where we define $\theta_{ij} = 1$ if $i \geq j$ and 0 otherwise):
\begin{align*}
    \delta^n g_k(x) &\leq (2kx + k^2 x^2) \theta_{0 n} + 2nx(1+kx) \theta_{1 n} + x^{-2}\sum_{\ell = \max(2,n)}^\infty (nx^2)^\ell\\
    &\leq kx(2+kx) \theta_{0 n} + 2nx(1+kx) \theta_{1 n} + \frac{(nx^2)^{\max(2,n)}}{x^2(1-nx^2)}
\end{align*}
Given our initial condition $n \leq t \leq x^{-1/2}$, with $t \geq 2$, we have $\delta^0 g_k(x) \leq 2.5 x^{1/2}$, $\delta^1 g_k(x) \leq 3x$, $\delta^2 g_k(x) \leq 2/(1-1/8) x^{3/2}$, and for $n > 2$, $\delta^n g_k(x) \leq 2x^{1.5 n - 2} \leq 2x^{(n+1)/2}$. Therefore, we always have
\begin{align*}
    \delta^n g_k(x) &\leq 3x^{\frac{n+1}{2}}
\end{align*}
    
\end{proof}

\begin{lemma}
\label{lemma:derangement_cbound}
We have
\begin{gather}
    c_t(x^2) \leq e^{\frac{t(t-1)}{2}x^2}
\end{gather}
which when $x \leq \frac{1}{t^2}$ implies
\begin{gather}
    c_t(x^2) < e^\frac{1}{2t^2}
\end{gather}
\end{lemma}
\begin{proof}
    \begin{align*}
    c_t(x^2) &= \prod_{\ell = 1}^{t} \left(1+(\ell-1)x^2\right)\\
    &\leq e^{\sum_{\ell=1}^t (\ell-1)x^2}\\
    &\leq e^{\frac{t(t-1)}{2}x^2}\\
\end{align*}
which is an increasing function of $x$. When  $x \leq \frac{1}{t^2}$ this bound is at most
\begin{align*}
    e^{\frac{t(t-1)}{2}\left(\frac{1}{t^2}\right)^2} = e^{\frac{1}{2t^2} - \frac{1}{2t^3}} <  e^{\frac{1}{2t^2}}
\end{align*}
\end{proof}

We are now ready to combine these pieces to prove Lemma \ref{lemma:dt_bound_analytic}.
\begin{proof} 
(Proof of Lemma \ref{lemma:dt_bound_analytic})
We start with Lemma \ref{lemma:derangement_rearrangement}. We apply Lemma \ref{lemma:derangement_Gbound} to the $\delta G$ term to obtain
\begin{align*}
    d_t(x) &\leq c_t(x^2)\sum_A \binom{t}{A} A^{t-A} 3^A x^{\frac{t}{2}}\\
    &\leq c_t(x^2) x^{\frac{t}{2}}\sum_{A=0}^{\lfloor \frac{t}{2} \rfloor} \binom{t}{A} [A^{t-A} 3^A  + (t-A)^A 3^{t-A}]\\
    &\leq c_t(x^2) x^{\frac{t}{2}}\sum_{A=0}^{\lfloor \frac{t}{2} \rfloor} \binom{t}{A} [(t/2)^{t-A} 3^A  + t^{t/2} 3^{t-A}]\\
    &\leq \half c_t(x^2) x^{\frac{t}{2}}\left[\left(\frac{t}{2} + 3\right)^t + t^{t/2} \left(1+3\right)^t\right] \\
    &\leq \half c_t(x^2) x^{\frac{t}{2}}t^t\left[\left(\frac{1}{2} + \frac{3}{t}\right)^t + \left(\frac{4}{\sqrt{t}}\right)^t\right]
\end{align*}
We can now use Lemma \ref{lemma:derangement_cbound} to obtain
\begin{gather}
    d_t(x) \leq \half e^{\frac{t(t-1)}{2}x^2} x^{\frac{t}{2}}t^t\left[\left(\frac{1}{2} + \frac{3}{t}\right)^t + \left(\frac{4}{\sqrt{t}}\right)^t\right]
\end{gather}
Now, if we use the condition $x = q^{-2} \leq t^{-2}$, we get
\begin{align}
    d_t(q^{-2}) &\leq \half e^{\frac{1}{2t^2}} \left(\frac{t}{q}\right)^t \left[\left(\frac{1}{2} + \frac{3}{t}\right)^t + \left(\frac{4}{\sqrt{t}}\right)^t\right]\n
    &\leq \frac{t^2}{2} e^{\frac{1}{2t^2}} \left[\left(\frac{1}{2} + \frac{3}{t}\right)^t + \left(\frac{4}{\sqrt{t}}\right)^t\right] \left(\frac{t}{q}\right)^{t-2} \frac{1}{q^2}
\end{align}
as desired. 
\end{proof}

Finally, we will also show that the $t = 2$ eigenvalue dominates for sufficiently large $t$. 
\begin{lemma}
Let $29 \leq t \leq q$. Then
\begin{gather}
d_t(q^{-2}) \leq \frac{1}{e^2(q^2 + 1)}
\end{gather}
\end{lemma}
\begin{proof}
    We begin with Lemma \ref{lemma:dt_bound_analytic}. 
     Writing the bound as $\tilde{d}(t)  \left(\frac{t}{q}\right)^{t-2} \frac{1}{q^2}$, we have $\tilde{d}(29) = 0.076 < 0.108 \leq \frac{q^2}{e^2(q^2+1)}$ 
     for all $q \geq 2$. It remains to show that the derivative of $\tilde{d}(t)$ over $t$ is negative for all $t$ values beyond 29. We have
\begin{gather}
    \frac{\dd \tilde{d}(t)}{\dd t} = e^{\frac{1}{2t^2}} \left[a(t) \left(\half + \frac{3}{t} \right)^t + b(t) \left(\frac{4}{\sqrt{t}}\right)^t\right]
\end{gather}
where
\begin{gather}
    a(t) = -t\left(\frac{3t}{t+6} - 1\right) - \frac{t^2}{2} \ln\left[\left(\half + \frac{3}{t}\right)^{-1}\right] - \frac{1}{2t}\\
    b(t) = - \frac{t^2}{4} \ln\left(\frac{t}{16}\right) - t\left(\frac{t}{4} - 1\right)- \frac{1}{2t}. 
\end{gather}
Since both are negative once $t$ exceeds 16, the whole derivative is negative, so $\tilde{d}(t)$ remains decreasing below the threshhold of 0.8 for all $t \geq 29$.
\end{proof}

\section{Irreps of generalized action}\label{app:lr_irreps}
In the analytic part of our proof, we used the right-action symmetry of the effective 3-site operator to separate the basis into isotypical components, as described in Lemma~\ref{lemma:isotype_basis}. However, the effective 3-site system also has symmetry under global left-action of any permutation, which is not an equivalent symmetry to right-action and can potentially split up the basis even further. In this section, we will generalize the right action isotypical basis of Lemma~\ref{lemma:isotype_basis} to symmetries under both left and right action.
We will start with a state that lives in a right-action irrep indexed by a partition $\lambda_R$ and a left-action irrep indexed by $\lambda_L$. For each conjugacy class $c$, we can select an arbitrary representative element $\sigma_c$, and note the centralizer $C(\sigma_c)$ with size $C_c$. Then, we define a basis for this isotypic component, indexed by left/right irrep indices $ij, k\ell$ respectively, as well as a conjugacy class index $c$, as follows:
\begin{gather}
    \left\{\ket{e_{\nu_L, \nu_R}^{ij, k\ell, c}} = \frac{1}{|C_c|}L_{\nu_L}^{ij} R_{\nu_R}^{k\ell}|I, I, \sigma_c\rangle \bigg| i,j \in \{1...d_{\nu_L}\}, k,l \in \{1...d_{\nu_R}\}\right\}\\
\end{gather}
where as before we have used
\begin{gather}
    R_\nu^{ij} = \sum_{\rho \in S_t} V_\nu(\rho^{-1})^{ij} \rho_R
\end{gather}
and the analogous
\begin{gather}
    L_\nu^{ij} = \sum_{\rho \in S_t} V_\nu(\rho^{-1})^{ij} \rho_L
\end{gather}
Note that we are allowed to index our basis element only in terms of the conjugacy class $c$, instead of a full permutation $\sigma$, as attempting to replace $\sigma_c$ with any other element $\tau_c$ in the same conjugacy class is equivalent to applying the tensor 
\begin{gather}
    V_{\nu_L}(\pi)^{hi} V_{\nu_R}(\pi^{-1})^{\ell m}
\end{gather}
where $\pi$ is such that $\tau_c = \pi \sigma_c \pi^{-1}$. The above operation consists of left-action of $V_{\nu_L}(\pi)$ to the left irrep indices, combined with right-action of $V_{\nu_R}(\pi^{-1})$ to the right irrep indices. Therefore, it is actually a group homomorphism - we will call this operation $T(\pi)_{\nu_L, \nu_R}^{hi, \ell m}$.

Each state $\ket{\gamma} = \sum_{ijk\ell c} \gamma^{ij, k\ell, c}\ket{e_{\nu_L, \nu_R}^{ij, k\ell, c}}$is then characterized by the tensor $\gamma^{ij, k\ell, c}$ of size $d_{\lambda_R}^2 d_{\lambda_L}^2 N_c$, where $N_c$ is the number of conjugacy classes in $S_t$ - note that unlike the $\beta$ tensor, $\gamma$ is not indexed by permutations in $S_t$. Unlike the right-action ansatz, this is not surjective - i.e. there exists some nonzero $\gamma^{ij, k\ell, c}$ tensors that yield an identically zero wavefunction (for example, if $\lambda_L = $Trv, $\lambda_R = $Alt, $\gamma^c = \delta_{c,I}$). This, combined with the fact that the single site metric in permutation space $g(\ket{\sigma}, \ket{\tau}) = q^{-|\sigma^{-1}\tau|}$ becomes singular if $q < t$, means that there is a possibility of spurious $\gamma$ tensors which are eigenstates of the gate but do not actually exist in the original space.

Like before, the left and right-action twirls commute with all gates, so applying the $G_m$ gate yields
\begin{align*}
    G_m \ket{e_{\nu_L, \nu_R}^{ij, k\ell, c}} &= \frac{1}{C_c} L_{\nu_L}^{ij} R_{\nu_R}^{k\ell} G\ket{I, I, \sigma_c}\\
    &= \frac{1}{C_c} L_{\nu_L}^{ij} R_{\nu_R}^{k\ell} \sum_{\pi, \tau} \Wg(\pi \tau^{-1}, q^2) q^{-|\tau| - |\tau \sigma_c^{-1}|} \ket{I, \pi, \pi}
\end{align*}
Since
\begin{gather}
    \sum_{\pi \in S_t} f(\pi) = \sum_a \frac{1}{C_a} \sum_\pi f(\pi \sigma_a \pi^{-1}),
\end{gather}
for conjugacy class index $a$, we can change this sum into
\begin{align*}
    G_m \ket{e_{\nu_L, \nu_R}^{ij, k\ell, c}} &= \frac{1}{C_c} L_{\nu_L}^{ij} R_{\nu_R}^{k\ell} \sum_{a,b} \frac{1}{C_a C_b} \sum_{\pi, \tau} \Wg(\pi \sigma_a \pi^{-1} \tau \sigma_b^{-1} \tau^{-1}, q^2) q^{-|\sigma_b| - |\tau \sigma_b \tau^{-1} \sigma_c^{-1}|} \pi_L \pi_R^{-1}\ket{I, \sigma_a, \sigma_a}\\
    &= \frac{1}{C_c} L_{\nu_L}^{ij} R_{\nu_R}^{k\ell} \sum_{a,b} \frac{1}{C_a C_b} \sum_{\pi, \tau} \Wg(\pi \sigma_a \pi^{-1} \sigma_b^{-1}, q^2) q^{-|\sigma_b| - |\tau \sigma_b \tau^{-1} \sigma_c^{-1}|} (\tau \pi)_L (\tau \pi)^{-1}_R \ket{I, \sigma_a, \sigma_a} \qquad (\pi \rightarrow \tau \pi)
\end{align*}
We then use the commutation relations
\begin{gather}
    L_{\nu_L}^{ij} \sigma_L = V_{\nu_L}^{ih}(\sigma) L_{\nu_L}^{hj}\\
    R_{\nu_R}^{k\ell} \sigma_R = R_{\nu_L}^{k n} V_{\nu_R}^{n\ell}(\sigma)
\end{gather}
(using Einstein summation notation for all repeated irrep indices). 
In fact,
\begin{gather}
    L_{\nu_L}^{ij} R_{\nu_R}^{k\ell} \sigma_L \sigma_R^{-1} = T_{\nu_L, \nu_R}(\sigma)^{ih, n\ell} L_{\nu_L}^{hj} R_{\nu_L}^{k n}
\end{gather}
Therefore
\begin{align*}
    G_m \ket{e_{\nu_L, \nu_R}^{ij, k\ell, c}} &= \frac{1}{C_c} \sum_{a,b} \frac{1}{C_a C_b} \sum_{\pi, \tau} \Wg(\pi \sigma_a \pi^{-1} \sigma_b^{-1}, q^2) q^{-|\sigma_b| - |\tau \sigma_b \tau^{-1} \sigma_c^{-1}|} T_{\nu_L, \nu_R}^{ih, n\ell}(\tau \pi) L_{\nu_L}^{hj} R_{\nu_R}^{kn} \ket{I, \sigma_a, \sigma_a} \\
    &= \mathcal{W}^{ig, o\ell}_{ab}(q^2) \mathcal{D}_{bb}(q) \mathcal{C}^{gh, no}_{bc}(q) \frac{1}{C_a} L_{\nu_L}^{hj} R_{\nu_R}^{kn} \ket{I, \sigma_a, \sigma_a}
\end{align*}
where
\begin{gather}
    \mathcal{W}^{ig, o\ell}_{ab}(Q) = \frac{1}{C_b}\sum_\pi \Wg(\pi \sigma_a \pi^{-1} \sigma_b^{-1} , Q) T(\pi)_{\nu_L, \nu_R}^{ig, o\ell}\\
    \mathcal{D}_{ab}(Q) = Q^{-|\sigma_a|} \delta_{ab}\\
    \mathcal{C}^{gh, no}_{bc}(Q) = \frac{1}{C_c}\sum_\tau Q^{-|\tau \sigma_b \tau^{-1} \sigma_c^{-1} |} T(\tau)_{\nu_L, \nu_R}^{gh, no}
\end{gather}
Applying $\Pi_m$ afterwards gives
\begin{align*}
    \Pi_m G_m \ket{e_{\nu_L, \nu_R}^{ij, k\ell, c}} &= \mathcal{W}^{ig, o\ell}_{ab}(q^2) \mathcal{D}_{bb}(q) \mathcal{C}^{gh, no}_{bc}(q) \frac{1}{C_a} L_{\nu_L}^{hj} R_{\nu_R}^{kn} \sum_{\pi, \tau} \Wg(\pi \tau^{-1}, q^{m}) q^{-|\tau| - |\tau \sigma_a^{-1}|} \ket{\pi, \pi, \sigma_a}\\
    &= \mathcal{W}^{ig, o\ell}_{ab}(q^2) \mathcal{D}_{bb}(q) \mathcal{C}^{gh, no}_{bc}(q) \frac{1}{C_a} L_{\nu_L}^{hj} R_{\nu_R}^{kn} \sum_{\pi, \tau} \Wg(\pi \tau^{-1}, q^{m}) q^{-|\tau \sigma_a| - |\tau|} (\sigma_a)_R \ket{\pi, \pi, I} 
\end{align*}
Where we took $(\pi \rightarrow \pi \sigma_a, \tau \rightarrow \tau \sigma_a)$ in the final line. Dropping the $\mathcal{W}, \mathcal{D}, \mathcal{C}$ matrices in front, and converting the sums over $\pi, \tau$ into conjugacy class sums over $\pi \sigma_e^{-1} \pi, \tau \sigma_d^{-1} \tau$, this expression becomes
\begin{align*}
    &=  \frac{1}{C_a} L_{\nu_L}^{hj} R_{\nu_R}^{kn} \sum_{\pi, \tau, d, e}\frac{1}{C_d C_e} \Wg(\pi \sigma_e^{-1} \pi^{-1} \tau \sigma_d \tau^{-1}, q^{m+1}) q^{-m|\tau \sigma_d^{-1} \tau^{-1} \sigma_a| - |\sigma_d|} (\sigma_a)_R \ket{\pi \sigma_e^{-1} \pi^{-1}, \pi \sigma_e^{-1} \pi^{-1}, I} \\
    &=  \frac{1}{C_a}  \sum_{\pi, \tau, d, e}\frac{1}{C_d C_e} \Wg(\pi \sigma_e^{-1} \pi^{-1} \tau \sigma_d \tau^{-1}, q^{m+1}) q^{-m|\tau \sigma_d^{-1} \tau^{-1} \sigma_a| - |\sigma_d|} L_{\nu_L}^{hj} R_{\nu_R}^{kn} (\sigma_a)_R \pi_L \pi_R^{-1} (\sigma_e)_R^{-1}\ket{I, I, \sigma_e} \\
    &=  \frac{1}{C_a}  \sum_{\pi, \tau, d, e}\frac{1}{C_d C_e} \Wg(\pi \sigma_e^{-1} \pi^{-1} \tau \sigma_d \tau^{-1}, q^{m+1}) q^{-m|\tau \sigma_d^{-1} \tau^{-1} \sigma_a| - |\sigma_d|} V_{\nu_R}(\sigma_a)^{on}L_{\nu_L}^{hj} R_{\nu_R}^{ko} \pi_L \pi_R^{-1} (\sigma_e)_R^{-1}\ket{I, I, \sigma_e}\\
    &=  \frac{1}{C_a}  \sum_{\pi, \tau, d, e}\frac{1}{C_d C_e} \Wg(\pi \sigma_e^{-1} \pi^{-1} \tau \sigma_d \tau^{-1}, q^{m+1}) q^{-m|\tau \sigma_d^{-1} \tau^{-1} \sigma_a| - |\sigma_d|} T(\pi)^{hg, po}V_{\nu_R}(\sigma_a)^{on}L_{\nu_L}^{gj} R_{\nu_R}^{kp} (\sigma_e)_R^{-1}\ket{I, I, \sigma_e}\\
    &=  \frac{1}{C_a}  \sum_{\pi, \tau, d, e}\frac{1}{C_d C_e} \Wg(\pi \sigma_e^{-1} \pi^{-1} \sigma_d, q^{m+1}) q^{-m|\tau \sigma_d^{-1} \tau^{-1} \sigma_a| - |\sigma_d|} V_{\nu_R}(\sigma_e^{-1})^{qp} T(\tau \pi)^{hg, po}V_{\nu_R}(\sigma_a)^{on}  L_{\nu_L}^{gj} R_{\nu_R}^{kq} \ket{I, I, \sigma_e}\\
    &=  \sum_{d, e}\mathcal{W}^{hf, pr}_{ed}(q^{m+1}) \mathcal{D}_{dd}(q) \mathcal{C}^{fg, ro}_{da}(q^m)V_{\nu_R}(\sigma_e^{-1})^{qp} V_{\nu_R}(\sigma_a)^{on} \ket{e_{\nu_L, \nu_R}^{gj, kq, e}} \\
    &=  \sum_{d, e} (\mathcal{D}_{\nu_R, ee}^{qp})^{-1} \mathcal{W}^{hf, pr}_{ed}(q^{m+1}) \mathcal{D}_{dd}(q) \mathcal{C}^{fg, ro}_{da}(q^m) D_{\nu_R,aa}^{on} \ket{e_{\nu_L, \nu_R}^{gj, kq, e}} 
\end{align*}
where
\begin{gather}
    \mathcal{D}_{\nu, ab}^{ij} = \delta_{ab} V_\nu(\sigma_a)^{ij}
\end{gather}
In total, we can therefore write the $\Pi_m G_m \Pi_m$ operator as
\begin{gather}
    \Pi_m G_m \Pi_m \ket{e_{\nu_L, \nu_R}^{ij, k\ell, c}} = \sum_{h, n} \mathcal{M}_{\nu_L, \nu_R}^{ih, n\ell, cb} \ket{e_{\nu_L, \nu_R}^{hj, kn, b}}
\end{gather}
where $\mathcal{M}$ is a product of tensors, analogous to the right-action matrix product:
\begin{gather}
    \mathcal{M}_{\nu_L, \nu_R}^{ih, n\ell, cb} = \left[(\mathcal{D}_{\nu_R})^{-1}\right]^{o_5 \ell}_{cc} \mathcal{W}(q^{m+1})^{ig_3, o_4 o_5}_{c a_3} \mathcal{D}(q)_{a_3 a_3} \mathcal{C}(q^m)^{g_3 g_2, o_3 o_4}_{a_3 a_2} \mathcal{D}_{\nu_R, a_2 a_2}^{o_2 o_3} \mathcal{W}(q^2)^{g_2 g_1, o_1 o_2}_{a_2 a_1} \mathcal{D}(q)_{a_1 a_1} \mathcal{C}(q)^{g_1 h, n o_1}_{a_1 b}
\end{gather}
Note that the indices corresponding to the right-action irrep are applied as a right action, while the indices corresponding to the left-action irrep, and the conjugacy class indices, are applied as left action.

One advantage of working in this basis is that the size of each isotypical subspace is significantly reduced. For example, under only right-action symmetry, the largest isotypical subspace dimension $\max_\lambda d_\lambda (t!)^2$ is 11520 at $t=6$. Under both left-action and right-action, the largest possible size of an isotypical subspace is $\max_{\lambda} d_\lambda^2 N_c = 2816$ for $t=6$, corresponding to the partition $\lambda=(3,2,1)$. However, these matrices tend to be extremely sparse, which reduces this number even further - for example, there are only 50 nonzero singular values in $\mathcal{C}(2)$ for the maximal dimension partition. The Gelfand-Tsetlin basis can be used to reduce the number of relevant vectors further. Since every leaf in the Bratteli diagram gives the same spectrum, only one path must be computed. 

\section{Special bound for $q=2$}
\label{app:special_q2_bound}
The bound
\begin{gather}
    k_m \leq \frac{1}{q^2 + 1}
\end{gather}
holds in every case we have checked except one. When $q = 2$, $t = 4$, and $m = 1$, the eigenvalue is $\frac{1}{4}$ instead of the expected $\frac{1}{5}$. We thus cannot prove Theorem \ref{thm:gap_finite_N} in this case by using Lemma \ref{lemma:gershgorin_block_bound}. We will now give an alternative to Lemma \ref{lemma:gershgorin_block_bound} that holds in this special case.

\begin{lemma}
    \label{lemma:special_t2_blockbound}
    Suppose that there exists $\lambda$ such that the following conditions hold: 
    \begin{itemize}
        \item $0.06963 \leq \lambda \leq \frac{1}{2}$
        \item $\lambda_i \leq \lambda \leq \frac{1}{2}$ for all $i > 1$
        \item $\lambda \leq \lambda_1 \leq \frac{1}{2}$
    \end{itemize}
    Then
    \begin{gather}
        \estair \leq \max\left((1 + \sqrt{1 - \lambda})^2 \lambda, \lambda_1 + (1 + \sqrt{1-\lambda})\sqrt{(1-\lambda_1)\lambda \lambda_1}\right)
    \end{gather}
\end{lemma}
\begin{proof}
    From the assumptions in the proof, we see
    \[P_i T_N P_j \leq A'_{ij}\]
    where
    \begin{equation}
    A' = \begin{bmatrix}
    \lambda_1 & \lambda \mu_1 & \lambda \mu \mu_1 & \dots & \lambda \mu^{N-2} \mu_1 & \lambda \mu^{N-1}\mu_1 \\
    \mu_1 & \lambda & \lambda \mu & \dots & \lambda \mu^{N-2} & \lambda \mu^{N-1}\\
    0 & \mu & \lambda & \dots & \lambda \mu^{N-3} & \lambda \mu^{N-2} \\
    \vdots & \vdots & \vdots & \ddots  & \vdots & \vdots \\
    0 & 0 & 0 & \dots & \mu & \lambda \\
\end{bmatrix}
\end{equation}
If we apply the Gershgorin Circle Theorem to $M^{-1} A' M$ as in the proof of Lemma \ref{lemma:gershgorin_block_bound}, the first column sum is
\begin{gather}
    C_1 = \lambda_1 + (\mu + \sqrt{\lambda})\mu_1 = \lambda_1 + (1 + \sqrt{1-\lambda})\sqrt{(1-\lambda_1)\lambda \lambda_1}
\end{gather}
For $j > 1$, the column sum is
\begin{align}
    C_j &= (\mu + \sqrt{\lambda})\mu\ + \lambda \sum_{k=0}^{j-1} \left(\frac{\mu}{\mu + \sqrt{\lambda}}\right)^k + \lambda \left(\frac{\mu}{\mu + \sqrt{\lambda}}\right)^{j-1}\left(\frac{\mu'}{\mu +\sqrt{\lambda}}\right) \\
    &= (\mu + \sqrt{\lambda})\mu\ + \lambda \sum_{k=0}^{j-1} \left(\frac{\mu}{\mu + \sqrt{\lambda}}\right)^k + \lambda \left(\frac{\mu}{\mu + \sqrt{\lambda}}\right)^j \frac{\mu'}{\mu}
    \\
    &= (\mu + \sqrt{\lambda})\mu\ + \lambda \sum_{k=0}^{\infty} \left(\frac{\mu}{\mu + \sqrt{\lambda}}\right)^k + \lambda\left[\left(\frac{\mu}{\mu + \sqrt{\lambda}}\right)^j \frac{\mu_1}{\mu} - \sum_{k = j}^{\infty}\left(\frac{\mu}{\mu + \sqrt{\lambda}}\right)^k\right] \\
    &= (\mu + \sqrt{\lambda})\mu\ + \lambda \sum_{k=0}^{\infty} \left(\frac{\mu}{\mu + \sqrt{\lambda}}\right)^k + \lambda\left(\frac{\mu}{\mu + \sqrt{\lambda}}\right)^j \left[\frac{\mu_1}{\mu} - \sum_{k = 0}^{\infty}\left(\frac{\mu}{\mu + \sqrt{\lambda}}\right)^k\right] \\
    &= (\mu + \sqrt{\lambda})\mu\ + \lambda \frac{1}{1 - \frac{\mu}{\mu + \sqrt{\lambda}}} + \lambda\left(\frac{\mu}{\mu + \sqrt{\lambda}}\right)^j \left[\frac{\mu_1}{\mu} - \frac{1}{1 - \frac{\mu}{\mu + \sqrt{\lambda}}}\right] \\
    &=\left(\left(1 + \sqrt{1 - \lambda}\right)^2  + \left(\frac{1}{1 + \frac{1}{\sqrt{1 - \lambda}}}\right)^j \left[\frac{\mu_1}{\mu} - \left(1 + \sqrt{1 - \lambda}\right)\right]\right)\lambda \\
\end{align}
We will now show that
\begin{gather}
    \frac{\mu_1}{\mu} - \left(1 + \sqrt{1 - \lambda}\right) \leq 0
\end{gather}
so that we may drop the last term of the bound. We know \(\mu_1 \leq \frac{1}{2}\), so it suffices to show that
\begin{align}
    \frac{1}{2} \leq \left(1 + \sqrt{1 - \lambda}\right) \sqrt{\lambda(1 - \lambda)}
\end{align}
This holds whenever $0.6963 \leq \lambda \leq 0.8474$. 
We may then apply Theorem \ref{thm:block_norm_bound} to complete the proof.
\end{proof}
This argument can be extended to the case where the first few columns of $A$ are exceptional, but we will not need it here. Note also that the $0.6963$ condition is not tight when $\mu_1 < \frac{1}{2}$. We check numerically that this bound holds for many cases.
\begin{lemma*}
    When  $q = 2$, $t \leq 6$, and $N \leq 1000$,
    \begin{gather}
        ||K_m|| \leq \begin{cases}
            \frac{1}{4} & m = 1\\
            \frac{1}{5} & m > 1
        \end{cases}
    \end{gather}
\end{lemma*}
\begin{proof}
This is obtained by checking each case numerically using the algorithm described in Appendix \ref{app:gap_finite_N}. Results are shown in Figure \ref{fig:larget_numerics}.
\end{proof}
Substituting into Lemma \ref{lemma:special_t2_blockbound}, we see that all three conditions are satisfied when we take \(\lambda_1 = \frac{1}{4}\) and \(\lambda = \frac{1}{q^2 + 1} = \frac{1}{5}\). In this case
\begin{align*}
    (1 + \sqrt{1 - \lambda})^2 \lambda &\approx 0.7178\\
    \lambda_1 + (1 + \sqrt{1-\lambda})\sqrt{(1-\lambda_1)\lambda \lambda_1} &\approx 0.6169 < 0.7178
\end{align*}
and so the dominant bound is, just as in the $q \ne 2$ case, $(1 + \sqrt{1 - \lambda})^2 \lambda$.

\end{appendices}
\end{document}